\documentclass[a4paper, 10pt]{article}
\usepackage{geometry}
\newcounter{somecounter}
\usepackage{booktabs, array}

\usepackage{setspace}
\usepackage{float}
\usepackage{standalone}
\usepackage{amsmath,amsfonts,amssymb}
\usepackage{amsthm}

\usepackage{color}
\usepackage{xcolor}
\usepackage[english]{babel}

\newtheorem{theorem}{Theorem}[section]
\newtheorem{assumption}{Assumption}[section]

\usepackage{caption}
\usepackage{subcaption}
\usepackage{bm}
\usepackage{pdflscape}
\usepackage{authblk}
\usepackage{graphicx}
\usepackage{hyperref}
\usepackage[normalem]{ulem}
\usepackage{multirow}
\usepackage{multicol}
\usepackage[flushleft]{threeparttable}
\usepackage{enumitem}
\usepackage{url}
\usepackage[utf8]{inputenc}
\usepackage{dsfont}

\usepackage{lipsum}
\makeatletter

\renewenvironment{proof}[1][\proofname]{\par
  \pushQED{\qed}%
  \normalfont \topsep6\p@\@plus6\p@\relax
  \list{}{\leftmargin=1em
          \rightmargin=\leftmargin
          \settowidth{\itemindent}{\itshape#1}%
          \labelwidth=\itemindent
          \parsep=0pt \listparindent=\parindent 
  }
  \item[\hskip\labelsep
        \itshape
    #1\@addpunct{.}]\ignorespaces
}{
  \popQED\endlist\@endpefalse
}

\makeatother
\usepackage{natbib}

\usepackage{mathtools}

\newcommand{\interior}[1]{%
 {\kern0pt#1}^{\mathrm{o}}%
}
\usepackage [english]{babel}
\usepackage [autostyle, english = american]{csquotes}
\MakeOuterQuote{"}

\usepackage{algorithm,algpseudocode}
\usepackage{caption}

\usepackage[makeroom]{cancel}

\makeatletter
\newcommand*\bigcdot{\mathpalette\bigcdot@{.5}}
\newcommand*\bigcdot@[2]{\mathbin{\vcenter{\hbox{\scalebox{#2}{$\m@th#1\bullet$}}}}}
\makeatother

\makeatother
\usepackage{natbib}

\usepackage{mathtools}

\usepackage{algorithm,algpseudocode}
\usepackage{caption}
\makeatother

\title{\textbf{Bayesian Function-on-Function Regression for Spatial Functional Data}}

\author[1]{Heesang Lee}
\author[1]{Dagun Oh}
\author[3]{Sunhwa Choi}
\author[1,2]{Jaewoo Park}
\affil[1]{Department of Statistics and Data Science, Yonsei University}
\affil[2]{Department of Applied Statistics, Yonsei University}
\affil[3]{Innovation Center for Industrial Mathematics, National Institute for Mathematical Sciences}

\begin{document}

\maketitle

\begin{abstract}
Spatial functional data arise in many settings, such as particulate matter curves observed at monitoring stations and age population curves at each areal unit. Most existing functional regression models have limited applicability because they do not consider spatial correlations. Although functional kriging methods can predict the curves at unobserved spatial locations, they are based on variogram fittings rather than constructing hierarchical statistical models. In this manuscript, we propose a Bayesian framework for spatial function-on-function regression that can carry out parameter estimations and predictions. However, the proposed model has computational and inferential challenges because the model needs to account for within and between-curve dependencies. Furthermore, high-dimensional and spatially correlated parameters can lead to the slow mixing of Markov chain Monte Carlo algorithms. To address these issues, we first utilize a basis transformation approach to simplify the covariance and apply projection methods for dimension reduction. We also develop a simultaneous band score for the proposed model to detect the significant region in the regression function. We apply our method to both areal and point-level spatial functional data, showing the proposed method is computationally efficient and provides accurate estimations and predictions.
\end{abstract}


\noindent%
{\it Keywords: dimension reduction; function-on-function regression; functional kriging; Markov chain Monte Carlo; Gaussian process}

\section{Introduction}
\label{sec:intro}
With the improvement in data collection technology, data have become large, complex, and high-dimensional. To analyze such datasets, functional data analysis has been developed in many disciplines, including linear models, nonparametric methods, and multivariate techniques  \citep[cf.][]{ramsay2005}. Although most functional data analysis methods assume independence between functions, there is an increasing interest in modeling dependent functions. An important example is spatial functional data \citep{delicado2010statistics,  martinez2020recent}, where we observe a curve with a spatial component. Although several analysis tools, including functional kriging and significance testing, have been developed (see \cite{kokoszka2019some} for a comprehensive review), fully Bayesian approaches have been under-explored.
In this manuscript, we propose a Bayesian approach for spatial functional regression, allowing both the predictor and response as functional variables. The proposed method can carry out functional kriging and study the relationships between functional variables. For efficient computation, we utilize basis transformation approaches \citep{zhang2016functional, kowal2020bayesian} and projection methods \citep{hughes2013dimension, park2022projection, lee2022picar}.

There is a vast literature on modeling functional variables in the regression context. \cite{morris2015functional} classified functional regression methods into three categories: (1) scalar-on-function regression models \citep{cardot1999functional, reiss2017methods}, (2) function-on-scalar regression models \citep{ramsay2005, reiss2010fast}, and (3) function-on-function regression models \citep{yao2005functional, kim2018additive}. These models extend the standard regression by taking responses or covariates (or both) as functional variables. However, most previous studies assume independent random errors, which is unsuitable for spatial functional data. 

There have been several works to account for spatial correlations in the model. \cite{hu2023bayesian} developed a Bayesian clustering method by combining Markov random fields and mixture models to account for geographical information. \cite{zhong2024bayesian} proposed a spatial functional clustering by utilizing random spanning trees for partitioning and latent Gaussian models for describing within-cluster dependency. In the regression context, \cite{zhang2016functional} developed a function-on-scalar regression based on conditional autoregressive models. For scalar-on-function regression, \cite{pineda2019functional} proposed a functional simultaneous autoregressive model to study econometric data. \cite{korte2022multivariate} developed a function-on-function regression with mixture components by modeling spatial dependence through a Markov random field. They fit the model by maximizing the likelihood function with the Monte Carlo EM algorithm.
Recently, \cite{kang2023fast} proposed a sparse functional spatial generalized linear mixed model in a Bayesian framework that is computationally efficient. Note that the prediction of curves at unobserved locations is one of the main interests in many geostatistical problems. The previous works focused on estimating functional parameters, not functional prediction (i.e., functional kriging). Furthermore, Bayesian inference for spatial function-on-function regression models has not been studied. 

In terms of functional kriging, \cite{goulard1993geostatistical} proposed the exploratory method based on the stationarity assumption. They reconstructed the complete curves via a parametric model by assuming that the curves were only known at a finite set of points. By extending this, \cite{giraldo2011ordinary} proposed nonparametric approaches based on the B-splines basis representation. \cite{caballero2013universal} developed a functional universal kriging (UK) that allows a different mean trend of functions. \cite{menafoglio2013universal} established a theoretical framework for UK under a general Hilbert space. All of these approaches focused on curve prediction based on variogram fittings; therefore, it is challenging to study the relationships between functional variables. Furthermore, it is not trivial to quantify uncertainties for functional parameters. 

In this manuscript, we develop a spatial function-on-function regression (SFoFR), extending the standard FoFR by including spatial random functions. There are two main inferential challenges for SFoFR.
\begin{itemize}
\item {\textit{Complex covariance}}: We need to account for both within-curve dependence and between-curve (spatial) dependence, resulting in a complex covariance function. To address this, we utilize a basis transformation approach \citep{zhang2016functional, kowal2020bayesian} that can convert every observed function from the data space into the basis space. From this, SFoFR can be represented as matrices of basis coefficients with a simpler covariance. 
\item {\textit{High-dimensional spatial random effects}}: Although the basis transformation can simplify complex covariance, Bayesian inference for SFoFR is still challenging, even with a moderate sample size. With an increasing sample size, the dimension of the spatial random effects matrix grows, resulting in the slow mixing of the Markov chain Monte Carlo (MCMC) algorithm. To address this, we propose a projection-based SFoFR (PSFoFR) by adopting dimension reduction methods \citep{hughes2013dimension, lee2022picar}. 
\end{itemize}
This manuscript makes the following main contributions.
\begin{itemize}    
    \item To our knowledge, this study is the first attempt to develop Bayesian FoFR for fitting spatial functional data. Our Bayesian framework allows for flexible hierarchical model specification and provides uncertainties using posterior samples. 
    \item We provide a theoretical justification for the method. We show that the regression function obtained from MCMC samples converges in distribution to the true function. We also show the consistency of the estimated regression functions from the posterior mean. 
    \item We investigate the performance of PSFoFR through empirical studies. We observe that PSFoFR is computationally efficient and provides accurate estimations and predictions. Furthermore, our framework is applicable to both discrete and continuous spatial domains, covering the broader class of spatial functional data.
\end{itemize}

The remainder of this manuscript is organized as follows. In Section~\ref{sec:FunctionalKrig}, we explain the background of the existing methods. In Section~\ref{sec:method}, we propose a PSFoFR and describe functional kriging. Furthermore, we provide theoretical justification for our method. In Section~\ref{sec:simulation}, we provide simulation studies to investigate the performance of the proposed methods. In Section~\ref{sec:appli}, we apply the method to real datasets, including Japan PM2.5 data and Korea mobility data. We conclude with a summary and discussion in Section~\ref{sec:discuss}.

\section{Background}
\label{sec:FunctionalKrig}

\subsection{Spatial Functional Regression}
\label{sec:sfreg}

Functional regression models have been studied for both areal and point-level data. For areal data, \cite{zhang2016functional} developed function-on-scalar regressions with nonseparable and nonstationary covariance structures. They used a basis transformation approach to make inference scalable to higher dimensional problems. \cite{kang2023fast} proposed scalar-on-function regression based on sparse representation of high-dimensional spatial random effects. Both of these approaches modeled spatial correlations through Gaussian Markov random fields. For point-level data, \cite{aristizabal2019analysis} proposed the analysis of variance of spatial functional data, and \cite{vrimalova2022inference} developed a permutation testing approach. Both methods are based on variogram fittings from the residuals of independent function-on-scalar regression models. \cite{park2023crop} developed a scalar-on-function regression model with a low-rank representation of spatially dependent functional covariates. 

Although functional kriging is one of the important geostatistical goals, most previous approaches focused on parameter estimation. Furthermore, existing methods take either responses or covariates as functional variables, but not both. Therefore, it is challenging to study the relationships between functional variables in the presence of spatial correlations. In Section~\ref{sec:method}, we propose a Bayesian framework for SFoFR that can be applicable to both areal and point-level data. The Bayesian approach is convenient for complex hierarchical models because we can carry out estimation and prediction using posterior samples.

\subsection{Spatial Functional Kriging}
\label{sec:sfkrig}

As a seminal work, \cite{goulard1993geostatistical} considered both multivariate and functional kriging approaches to predict function values. In terms of multivariate perspectives, they regarded the vector of observed values from each curve as a multivariate random variable. Then, they used cokriging to predict the values of the multivariate random vector. As a functional approach, \cite{goulard1993geostatistical} defined the best linear unbiased predictor (BLUP) of curves at unobserved sites and fitted each curve with a parametric model. Extending such approaches, \cite{giraldo2011ordinary} proposed functional ordinary kriging (OK) by replacing parametric models with nonparametric models based on cubic B-splines. Then, they constructed BLUP of the OK predictor. However, they still assume stationarity among functions, which has limited applicability. To address this, universal kriging (UK) approaches that allow different mean functions have been developed
\citep{caballero2013universal,menafoglio2013universal}. Here, we describe the UK following \cite{caballero2013universal}, and the results can be extended to general Hilbert spaces \citep{menafoglio2013universal}.

Consider we have a non-stationary stochastic process $\{Y_{\boldsymbol{s}}(t), t \in \mathcal{T}, \boldsymbol{s}\in \mathcal{D}\}$
in $\mathcal{L}^2(\mathcal{T})$, the set of all square integrable functions on $\mathcal{T} \in \mathbb R$ defined over the spatial domain $\mathcal{D} \in \mathbb R^{2}$. Then, the $i$th observed curve realized from the random process can be modeled as \begin{equation}\label{defRandomProcess}
    Y_{\boldsymbol{s}_i}(t) =   \mathbf{X}^{'}_{\boldsymbol{s}_i}\boldsymbol{\beta}(t) + \upsilon_{\boldsymbol{s}_i}(t),
\end{equation}
where $\mathbf{X}^{'}_{\boldsymbol{s}_i} \in \mathbb{R}^{L}$ is a vector of spatial variables at location $\boldsymbol{s}_i$ and $\boldsymbol{\beta}(t)=(\beta_1(t),\cdots,\beta_L(t))'$ is the corresponding regression function vector. In addition, $\upsilon_{\boldsymbol{s}_i}(t)$ is a functional error term from a zero-mean second-order stationary and isotropic random process $\{\upsilon_{\boldsymbol{s}}(t), t \in \mathcal{T}, \boldsymbol{s}\in \mathcal{D}\}$
in $\mathcal{L}^2(\mathcal{T})$. Note that \eqref{defRandomProcess} allows the non-constant spatial mean function (drift) and belongs to a function-on-scalar regression framework. 

The UK approaches can select an optimal linear model for a given $n$ number of observations $Y_{\boldsymbol{s}_1}(t),\cdots,Y_{\boldsymbol{s}_n}(t)$. Furthermore, UK approaches can predict $Y_{\boldsymbol{s}^{\ast}_{1}}(t)$ for an unobserved location $\boldsymbol{s}^{\ast}_{1}$ by accounting for the spatial dependence structure in the model. The BLUP of the UK predictor at $\boldsymbol{s}^{\ast}_{1}$ is
\begin{equation}\label{UkPred}
   \widehat{Y}_{\boldsymbol{s}^{\ast}_{1}}(t) = \sum_{i=1}^n \lambda_i {Y_{\boldsymbol{s}_i}(t)},
\end{equation}
where the weight $\boldsymbol{\lambda}=(\lambda_1,\cdots,\lambda_n)'$ is  obtained by minimizing the global variance of the prediction error under the following unbiasedness constraint:
\begin{equation}\label{UKconstraint}
    \text{min}_{\lambda_1,...\lambda_n} \text{Var} (\widehat{Y}_{\boldsymbol{s}^{\ast}_{1}}(t) - {Y_{\boldsymbol{s}^{\ast}_{1}}(t)}),\quad \text{s.t} \quad \mathbb{E} \bigl [\widehat{Y}_{\boldsymbol{s}^{\ast}_{1}}(t) - {Y_{\boldsymbol{s}^{\ast}_{1}}(t)} \bigr] = 0.
\end{equation}
\cite{caballero2013universal,menafoglio2013universal} showed that the optimal weight $\boldsymbol{\lambda}$ in \eqref{UKconstraint}, can be obtained by solving the following linear equation.
 \begin{equation}\label{UKlinearsystem}
    \begin{pmatrix} 
        C_{11} & \cdots & C_{1n} & \mathbf{X}'_{\boldsymbol{s}_1} \\
        \vdots & \ddots & \dots & \vdots \\
        C_{n1} & \cdots & C_{nn} & \mathbf{X}'_{\boldsymbol{s}_n} &  \\
        \mathbf{X}_{\boldsymbol{s}_1} & \cdots & \mathbf{X}_{\boldsymbol{s}_n}  &  \mathbf{0} \\ 
    \end{pmatrix}
    \begin{pmatrix}
        \boldsymbol{\lambda} \\ \boldsymbol{\mu} 
    \end{pmatrix}
    = 
    \begin{pmatrix}
        C_{01} \\ \vdots \\ C_{0n} \\  \mathbf{X}_{\boldsymbol{s}^{\ast}_{1}}
    \end{pmatrix},
 \end{equation}
where $C_{ij}= \text{Cov}(Y_{\boldsymbol{s}_i}(t), Y_{\boldsymbol{s}_j}(t))$ denotes the trace-covariogram function  and $\boldsymbol{\mu}\in \mathbb R^{L}$ represents the Lagrange multiplier corresponding to the $L$ unbiasedness constraints in \eqref{UKconstraint}. The UK approaches described above can be implemented through the publicly available \texttt{R} package \texttt{fdagstat} \citep{fdagstat_2017}. 

Although UK methods are valuable for exploratory data analysis, they are not based on formal statistical inference. The fitted results from the trace-covariogram function are plugged into \eqref{UkPred} by ignoring the uncertainties in the estimation step. In terms of the model, the drift term in \eqref{defRandomProcess} is limited to the scalar covariates; therefore, we cannot study the relationships between functional variables. Furthermore, the uncertainty quantification of regression functions is not trivial for UK approaches. To address these challenges, we propose Bayesian spatial functional
regression methods in the following section.

\section{Functional Regression for Spatial Data}
\label{sec:method}

Let $Y_{\boldsymbol{s}_1}(t),..., Y_{\boldsymbol{s}_n}(t)$ be the observed functional response from a second-order stochastic process $\{Y_{\boldsymbol{s}}(t), t \in \mathcal{T}, \boldsymbol{s}\in \mathcal{D}\}$ in $\mathcal{L}^2(\mathcal{T})$. Similarly, we can define $X_{\boldsymbol{s}_1}(r),..., X_{\boldsymbol{s}_n}(r)$ as the observed functional covariates from a second-order stochastic process $\{X_{\boldsymbol{s}}(r), r \in \mathcal{V}, \boldsymbol{s}\in \mathcal{D}\}$ in $\mathcal{L}^2(\mathcal{V})$. To account for the spatial dependence between curves, we introduce a zero-mean Gaussian process $\{W_{\boldsymbol{s}}(t), t \in \mathcal{T}, \boldsymbol{s}\in \mathcal{D}\}$ that belongs to $\mathcal{L}^2(\mathcal{T})$.
Here, we propose the spatial function-on-function regression (SFoFR) as
\begin{equation}\label{fregmodel}
   Y_{\boldsymbol{s}_i}(t) = \int_{\mathcal{V}} \Psi(r,t)X_{\boldsymbol{s}_i}(r)dr + W_{\boldsymbol{s}_i}(t)+ \epsilon_{\boldsymbol{s}_i}(t),
\end{equation}
where $\epsilon_{\boldsymbol{s}_i}(t)$ is an independent error function. The functional parameter (regression function) $\Psi$ belongs to $\mathcal{L}(\mathcal{V}) \otimes \mathcal{L}(\mathcal{T})$, where $\otimes$ indicates a tensor product. By extending the standard spatial regression models, SFoFR can study the relationships between functional variables and interpolate functions at unobserved locations. Following \cite{kowal2020bayesian},  we expand the functional variables in  \eqref{fregmodel} with respect to the orthonormal basis functions $\{\phi_k\}_{k\in \mathbb{N}} \in \mathcal{L}(\mathcal{T}),\ \{\phi_g\}_{g\in \mathbb{N}} \in \mathcal{L}(\mathcal{V})$. For instance, we can orthonormalize the basis functions through QR factorization or directly use orthonormal ones (e.g., FPC, Fourier).
Then, we have
\begin{equation}\label{fregmodelcomp}
Y_{\boldsymbol{s}_i}(t) = \sum_{k=1}^{\infty} y_{ik}\phi_k(t), \quad X_{\boldsymbol{s}_i}(r)=  \sum_{g=1}^{\infty} x_{ig}\phi_g(r), \quad W_{\boldsymbol{s}_i}(t)=  \sum_{k=1}^{\infty} w_{ik}\phi_k(t), \quad \epsilon_{\boldsymbol{s}_i}(t)=  \sum_{k=1}^{\infty} e_{ik}\phi_k(t). 
\end{equation}
We can also construct the orthonormal tensor basis $\{\phi_{gk}\}_{g,k \in \mathbb{N}} = \{\phi_g\}_{g\in \mathbb{N}} \otimes\{\phi_k\}_{k\in \mathbb{N}} \in \mathcal{L}(\mathcal{V}) \otimes \mathcal{L}(\mathcal{T})$ and the regression function can be represented as
\begin{equation}\label{psi}
\Psi(r, t) = \sum_{k=1}^\infty\sum_{g=1}^\infty\psi_{gk}\phi_{gk}(r,t).
\end{equation}
In \eqref{fregmodelcomp}, $\{y_{ik}\}_{k \in \mathbb{N}}$, $\{x_{ig}\}_{g \in \mathbb{N}}$, and $\{w_{ik}\}_{k \in \mathbb{N}}$ are correlated random variables across different spatial locations. Without loss of generality, we assume that they have mean zero and satisfy spatial homoscedasticity \citep{pineda2019functional, kang2023fast}. We assume that the error coefficients $\{e_{ik}\}_{k \in \mathbb{N}}$ follow $\mathcal{N}(0, \tau^2)$ independently. 

There are computational and inferential challenges for the model specified above. Since each functional component is correlated in terms of both $t$ (within-curve dependence) and $\boldsymbol{s}_i$ (between-curve dependence), we need to model covariance functions, resulting in complex likelihood functions. To address this, we utilize a basis transformation approach \citep{zhang2016functional, kowal2020bayesian}. We transform each observed curve from the data space to the basis space and fit the model in the basis space. Then, we transform it back to the data space; the estimation and prediction results are interpreted in the data space. From this procedure, we can fit the model with a simpler covariance in the basis space while still modeling complex between and within-curve covariances in the data space. To transform to the basis space, we first take the inner product with $\phi_{k'}(t)$  on both sides of \eqref{fregmodel} as
\begin{equation}\label{fregmodeltransform}
   \int_{\mathcal{T}} \phi_{k'}(t) Y_{\boldsymbol{s}_i}(t)dt = \int_{\mathcal{T}} \phi_{k'}(t) \int_{\mathcal{V}} \Psi(r,t)X_{\boldsymbol{s}_i}(r)drdt + \int_{\mathcal{T}} \phi_{k'}(t) W_{\boldsymbol{s}_i}(t) dt + \int_{\mathcal{T}} \phi_{k'}(t)\epsilon_{\boldsymbol{s}_i}(t)dt.
\end{equation}
Since $\{\phi_k\}_{k\in \mathbb{N}}$ is an orthonormal basis function, using dominated convergence, the left-hand side of \eqref{fregmodeltransform} becomes
\begin{equation}\label{dominatedConvergence}
    \int \phi_{k'}(t)Y_{\boldsymbol{s}_i}(t)dt = \int \phi_{k'}(t)\left( \sum_{k=1}^\infty y_{ik}\phi_{k}(t) \right) dt = \sum_{k=1}^{\infty} y_{ik} \left( \int \phi_{k'}(t)\phi_{k}(t)dt \right) = y_{ik},
\end{equation} 
when $k=k'$. In a similar fashion, $\int_{\mathcal{T}} \phi_k(t)W_{\boldsymbol{s}_i}(t)dt$ and $\int_{\mathcal{T}} \phi_k(t)\epsilon_{\boldsymbol{s}_i}(t)dt$ terms can be represented as $w_{ik}$ and $e_{ik}$, respectively. Lastly, with dominated convergence, we have 
\begin{align}\label{PsiDominatedConvergence}
   \int_{\mathcal{T}} \phi_{k'}(t) \int_{\mathcal{V}} \Psi(r,t)X_{\boldsymbol{s}_i}(r)drdt  &= \int_{\mathcal{T}} \phi_{k'}(t) \int_{\mathcal{V}} \left( \sum_{k=1}^\infty \sum_{g=1}^\infty\psi_{gk}\phi_{gk}(r,t) \right) \left( \sum_{g=1}^\infty x_{ig}\phi_{g}(r) \right) drdt\nonumber\\ 
    &= \sum_{g=1}^\infty \psi_{gk} x_{ig},
    \end{align}
when $k=k'$. In practice, we approximate the above infinite basis expansions with the finite number of basis functions $k_n$ and $g_n$, where $k_n$ and $g_n$ increase as $n \rightarrow \infty$ \citep{muller2005fglm}. Under this finite approximation, if we repeat \eqref{dominatedConvergence}, \eqref{PsiDominatedConvergence} with different $k'$, we have matrices of the coefficients in the basis space. Let $\widetilde{\mathbf{Y}}=(y_{ik})_{n\times k_n}$ be the matrix of transformed curves in the basis space and $\widetilde{\mathbf{X}}=(x_{ig})_{n \times g_n}$ be the design matrix with the coefficients $\widetilde{\boldsymbol{\psi}}=(\psi_{gk})_{g_n \times k_n}$. Similarly, let  $\widetilde{\mathbf{W}}=(w_{ik})_{n \times k_n}$ be the transformed random effects matrix
and $\widetilde{\boldsymbol{\epsilon}}=(e_{ik})_{n \times k_n}$ be the transformed error matrix. Then, we can represent the SFoFR in the basis space with finite truncation as
\begin{equation}\label{fregmodeltrunc}
\widetilde{\mathbf{Y}} = \widetilde{\mathbf{X}} \widetilde{\boldsymbol{\psi}}+\widetilde{\mathbf{W}} +\widetilde{\boldsymbol{\epsilon}}, 
\end{equation}
where $\widetilde{\mathbf{W}} \sim \mathcal{MGP}(\boldsymbol{\Lambda}, \widetilde{\boldsymbol{\Sigma}})$. Following \cite{morris2006wavelet}, $\mathcal{MGP}(\boldsymbol{\Lambda}, \widetilde{\boldsymbol{\Sigma}})$ refers to a multivariate Gaussian process with a between-curve covariance $\boldsymbol{\Lambda}\in \mathbb R^{n\times n}$ and a within-curve covariance $\widetilde{\boldsymbol{\Sigma}}\in \mathbb R^{k_n\times k_n}$. From this definition, the covariance of $\text{vec}(\widetilde{\mathbf{W}}) \in R^{n k_n}$ is $\boldsymbol{\Lambda} \otimes \widetilde{\boldsymbol{\Sigma}}$, where $\otimes$ denotes a Kronecker product. (1) For continuous domains, we can use the Mat\'{e}rn class covariance to model $\boldsymbol{\Lambda}$ that depends on the variance parameter $\sigma^2$ and range parameter $\rho$. (2) For discrete domains, we can use intrinsic autoregressive models. Let $\mathbf{D} \in \mathbb R^{n\times n}$ be an adjacency matrix, where $D_{ij}=1$, if locations $i,j$ are neighbors, and otherwise $D_{ij}=0$. Then the model is defined with a precision matrix $\nu[\mbox{diag}(\mathbf{D}\mathbf{1})-\mathbf{D}]$, where $\nu$ is a spatial smoothness parameter. For both continuous and discrete domains, we can model $\widetilde{\boldsymbol{\Sigma}}$ through a simple diagonal matrix in the basis space (e.g., $\widetilde{\boldsymbol{\Sigma}}=s\mathbf{I}$). Such an independence assumption in the basis space is justified for a wide variety of basis functions \citep{zhang2016functional}. Note that while we use a simple covariance structure for $\widetilde{\boldsymbol{\Sigma}}$ in the basis space, it still allows a flexible covariance structure in the data space. In the following paragraph, we provide more details about the within-curve covariance in the data space.  

Once we specify the priors of the parameters in \eqref{fregmodeltrunc}, we can construct the MCMC algorithm to generate posterior samples from the basis space model. We provide the conditional distributions for SFoFR in the supplementary material. These posterior samples will be transformed back into the data space for interpreting the regression function and kriging. Let $\boldsymbol{\Phi} \in \mathbb R^{k_n \times n_t}$ be a matrix of basis functions $\phi_k(t)$ evaluated at $n_t$ grid points over $\mathcal{T}$. Similarly, let $\boldsymbol{\Xi} \in \mathbb R^{g_n \times n_v}$ be the matrix of basis functions $\phi_g(r)$ evaluated at $n_v$ grid points over $\mathcal{V}$. If we multiply both sides of \eqref{fregmodeltrunc} with $\boldsymbol{\Phi}$, we can transform back into the data space as
$$\widetilde{\mathbf{Y}}\boldsymbol{\Phi} = \widetilde{\mathbf{X}} \boldsymbol{\Xi} \boldsymbol{\Xi}'\widetilde{\boldsymbol{\psi}}\boldsymbol{\Phi}+\widetilde{\mathbf{W}}\boldsymbol{\Phi} +\widetilde{\boldsymbol{\epsilon}}\boldsymbol{\Phi}.$$
The left-hand side term $\widetilde{\mathbf{Y}}\boldsymbol{\Phi}$ is a finite basis expansion of the observed curves evaluated over $n_t$ grid points, which can be defined as $\mathbf{Y} \in \mathbb R^{n\times n_t}$. Similarly, $\widetilde{\mathbf{X}}\boldsymbol{\Xi}$, $\boldsymbol{\Xi}'\widetilde{\boldsymbol{\psi}}\boldsymbol{\Phi}$, $\widetilde{\mathbf{W}}\boldsymbol{\Phi}$, and $\widetilde{\boldsymbol{\epsilon}}\boldsymbol{\Phi}$ can be defined as $\mathbf{X} \in \mathbb R^{n \times n_v}$, $\boldsymbol{\psi} \in \mathbb R^{n_v \times n_t}$, $\mathbf{W} \in \mathbb R^{n\times n_t}$, and $\boldsymbol{\epsilon} \in \mathbb R^{n\times n_t}$, respectively, which are finite representations of the functional components in 
\eqref{fregmodel}. Then we have the SFoFR in the data space with finite truncation as 
\begin{equation}\label{fregmodeldata}
\mathbf{Y} = \mathbf{X} \boldsymbol{\psi}+\mathbf{W} +\boldsymbol{\epsilon}, 
\end{equation}
where $\mathbf{W} \sim \mathcal{MGP}(\boldsymbol{\Lambda}, \boldsymbol{\Sigma})$. Therefore, in the data space, we allow a flexible form of the within-curve covariance as $\boldsymbol{\Sigma}=\boldsymbol{\Phi}'\widetilde{\boldsymbol{\Sigma}}\boldsymbol{\Phi}$, while the MCMC algorithms are implemented in the basis space with a simple covariance structure as in \eqref{fregmodeltrunc}. From $\boldsymbol{\psi}=\boldsymbol{\Xi}'\widetilde{\boldsymbol{\psi}}\boldsymbol{\Phi}$, we can also interpret the regression function in the data space. A graphical description of the basis transformation approach is presented in Figure~\ref{fig:outlineTrans}.

\begin{figure}[htbp]
     \centering
         \centering
         \includegraphics[width=\linewidth]{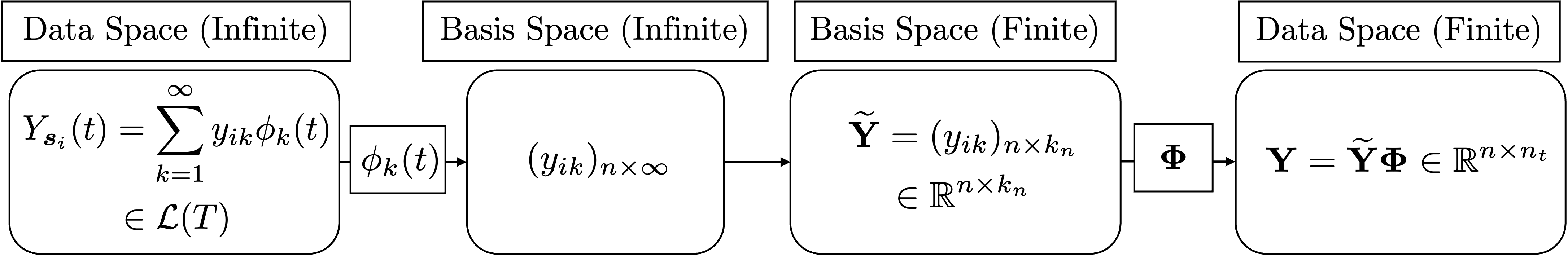}
         \caption{Outline of the basis transformation approach.}
         \label{fig:outlineTrans}
\end{figure}

However, implementing the MCMC algorithm for \eqref{fregmodeltrunc} can still be computationally demanding with increasing numbers of observations. The dimension of the transformed random effects $\widetilde{\mathbf{W}} \in \mathbb R^{n \times k_n}$ becomes larger with larger $n$, and they are spatially correlated; this can lead to the slow mixing of the chain \citep{kang2023fast}. To address this, we develop a projection-based SFoFR that is computationally efficient.

\subsection{Projection-Based Function-on-Function Regression}
\label{sec:psfofr}

Projection methods \citep{hughes2013dimension, lee2022picar} have been studied for both areal and point-level data for efficient computation. They can reduce the dimension of the spatial latent process from $n$ to $p$ $(<<n)$ based on the projection matrix $\mathbf{P} \in \mathbb R^{n \times p}$. Here, we apply these methods to the transformed random effect $\widetilde{\mathbf{W}}$ in \eqref{fregmodeltrunc}.

For the areal data over the discrete spatial domain, we construct the projection matrix $\mathbf{P}$ by taking the leading $p$ eigencomponents of Moran's operator. Its eigenvalues correspond to the degree of spatial dependence, and the corresponding eigenvectors present the spatial clustering patterns of the data exhibit \citep{hughes2013dimension}. To obtain $\mathbf{P}$, we first calculate $\mathbf{H}^{\perp}$ which is a complement of $\mathbf{H} = \widetilde{\mathbf{X}}(\widetilde{\mathbf{X}}'\widetilde{\mathbf{X}})^{-1}\widetilde{\mathbf{X}}'$. Here, $\widetilde{\mathbf{X}} \in \mathbb R^{n \times g_n}$ is a design matrix in \eqref{fregmodeltrunc}. Then we calculate Moran's operator $\mathbf{H}^{\perp} \mathbf{D} \mathbf{H}^{\perp}$, where $\mathbf{D} \in \mathbb R^{n\times n}$ is an adjacency matrix for observations. We take the leading $p$ eigenvectors of Moran's operator to construct $\mathbf{P} \in \mathbb R^{n \times p}$ and the transformed random effect $\widetilde{\mathbf{W}}$ in \eqref{fregmodeltrunc} can be approximated as $\mathbf{P}\boldsymbol{\delta}$. 

For point-level data over the continuous spatial domain, we follow \cite{lee2022picar} to construct the projection matrix $\mathbf{P}$. First, we generate triangular mesh covering the spatial domain $\mathcal{D}$ using the {\texttt{R}} package {\texttt{INLA}} \citep{lindgren2015bayesian}. The mesh constructs an undirected graph $\mathbf{G}= (\mathbf{V}, E)$ where $\mathbf{V}=\lbrace \mathbf{v}_1,\cdots,\mathbf{v}_m\rbrace$ is the set of locations of the mesh vertices and $E$ is the set of edges. Then, we construct Moran's operator 
$(\mathbf{I-11'}/m)\mathbf{N}(\mathbf{I-11'}/m)$, where $\mathbf{N}\in \mathbb R^{m\times m}$ is the adjacency matrix from the graph $\mathbf{G}$. As in the discrete case, we take the leading $p$ eigenvectors of Moran's operator to obtain $\mathbf{M} \in \mathbb R^{m \times p}$. Note that the eigenvectors present spatial dependence among the mesh vertices, not about point-level observations \citep{lee2022picar}. Therefore, to interpolate observations within the mesh, we use piecewise linear basis functions $\mathbf{A} \in \mathbb{R}^{n\times m}$. The rows of $\mathbf{A}$ correspond to the observed locations, and the columns correspond to the mesh vertices; the $i$-th row of $\mathbf{A}$ contains weights to interpolate onto point-level locations. From this, we can construct a projection matrix   $\mathbf{P}=\mathbf{AM} \in \mathbb R^{n \times p}$
and approximate the transformed random effect $\widetilde{\mathbf{W}}$ in \eqref{fregmodeltrunc} as $\mathbf{P}\boldsymbol{\delta}$. 

Therefore, for both discrete and continuous spatial domains, we have the projection-based SFoFR (PSFoFR) in the basis space as
\begin{equation}\label{projmodeltrunc}
\widetilde{\mathbf{Y}} = \widetilde{\mathbf{X}} \widetilde{\boldsymbol{\psi}}+\mathbf{P}\boldsymbol{\delta} +\widetilde{\boldsymbol{\epsilon}}.
\end{equation}
As in SFoFR, we use the MCMC algorithm to sample from the posterior distributions in the basis space model \eqref{projmodeltrunc} and transform them back into the data space for inference. We provide the details about the full conditionals for PSFoFR in the supplementary material (Appendix A). There are two clear computational advantages from the projection step. First, the projection methods significantly reduce the number of parameters that have to be updated in the MCMC sampler. In SFoFR, we have $\widetilde{\boldsymbol{\psi}} \in \mathbb R^{g_n \times k_n}$, $\widetilde{\mathbf{W}} \in \mathbb R^{n \times k_n}$, and spatial covariance parameters (2 or 3 depends on model specification). On the other hand, we have $\widetilde{\boldsymbol{\psi}} \in \mathbb R^{g_n \times k_n}$, $\boldsymbol{\delta} \in \mathbb R^{p \times k_n}$, and spatial covariance parameters in PSFoFR. Since $p<<n$, there is a significant reduction in the number of parameters. For instance, the number of parameters is reduced from 20,103 to 2,434 (about 88\% reduction) in the PM 2.5 data example (Section~\ref{sec:pm25}) due to the projection step. Second, the projection step can improve the mixing of the MCMC by reducing the correlation among the random effect parameters \citep{hughes2013dimension, lee2022picar}. To assess the improvement in MCMC mixing, we compute the effective sample size and provide the trace plots in the supplementary material (Appendix B). We observe that PSFoFR has better mixing compared to SFoFR for the same MCMC iterations, especially for the random effect parameters. 

For the point-level data, predictions at unobserved locations (i.e., kriging) are of great interest. Consider we have $n_{cv}$ number of unobserved locations $\boldsymbol{s}^{\ast}=(\boldsymbol{s}^{\ast}_{1},\cdots,\boldsymbol{s}^{\ast}_{n_{cv}}) \in \mathcal{D}$. Similar to \eqref{fregmodelcomp}, let $\widetilde{\mathbf{Y}}_{\boldsymbol{s}^{\ast}} =(y_{ik})_{n_{cv}\times k_n}$ and $\widetilde{\mathbf{X}}_{\boldsymbol{s}^{\ast}}=(x_{ig})_{n_{cv} \times g_n}$ be the matrices of basis coefficients obtained from response curves and covariate curves at $\boldsymbol{s}^{\ast}$, respectively. For PSFoFR, we can construct a piecewise linear basis function matrix  $\mathbf{A}_{\boldsymbol{s}^{\ast}} \in \mathbb{R}^{n_{cv} \times m}$, where the $i$-th row of $\mathbf{A}_{\boldsymbol{s}^{\ast}}$ contains weights to interpolate the vertices onto $\boldsymbol{s}^{\ast}$. As before, we can obtain a projection matrix $\mathbf{P}_{\boldsymbol{s}^{\ast}}=\mathbf{A}_{\boldsymbol{s}^{\ast}}\mathbf{M} \in \mathbb R^{n_{cv} \times p}$. Then, the linear predictor $\widetilde{\boldsymbol{\eta}}_{\boldsymbol{s}^{\ast}}=(\widetilde{\eta}_{ik})_{n_{cv}\times k_n}$ at the basis space is defined as
\begin{equation}\label{krigbasis}
\widetilde{\boldsymbol{\eta}}_{\boldsymbol{s}^{\ast}} = \widetilde{\mathbf{X}}_{\boldsymbol{s}^{\ast}} \widetilde{\boldsymbol{\psi}}+\mathbf{P}_{\boldsymbol{s}^{\ast}}\boldsymbol{\delta}.
\end{equation}
From $u$th posterior sample, we have $\widetilde{\eta}^{(u)}_{ik}$ and $\tau^{2(u)}$. Therefore, we can generate a sample from the posterior predictive distribution as $y_{ik}^{(u)} \sim N(\widetilde{\eta}^{(u)}_{ik},\tau^{2(u)})$. Note that $\widetilde{\mathbf{Y}}^{(u)}_{\boldsymbol{s}^{\ast}}=(y_{ik}^{(u)})_{n_{cv}\times k_n}$ is from a posterior predictive distribution in the basis space. For interpretation, we transform it back to the data space by multiplying $\boldsymbol{\Phi}$ as $\mathbf{Y}^{(u)}_{\boldsymbol{s}^{\ast}}=\widetilde{\mathbf{Y}}^{(u)}_{\boldsymbol{s}^{\ast}}\boldsymbol{\Phi} \in \mathbb R^{n_{cv}\times n_t}$, which represent predictive curves from $n_{cv}$ unobserved locations and each of them is evaluated at $n_t$ grid points. From $\lbrace \mathbf{Y}^{(u)}_{\boldsymbol{s}^{\ast}} \rbrace_{u=1}^{U}$, we compute the posterior mean $\widehat{\mathbf{Y}}_{s^{\ast}}$ (the point estimate of predictive curves) and simultaneous credible intervals (prediction uncertainties) described in Section~\ref{sec:simulcre}. 

To effectively implement PSFoFR, we need to tune the following: (1) the number of basis functions ($k_n, g_n$), (2) rank ($p$), and (3) the number of mesh vertices ($m$). Note that $m$ needs to be tuned only for point-level data. In our analysis, $k_n$ and $g_n$ are determined by minimizing the generalized cross-validation (GCV) error. Following the suggestions in \cite{hughes2013dimension, lee2022picar}, we choose the rank from the desired proportion of variation (e.g., 90\% using the leading eigencomponents). In our simulation studies, we show that setting $p$ between 5\% and 10\% of the sample size can achieve both computational and inferential efficiencies. Lastly, $m$ should be large enough to discretize spatial domain, though extremely large $m$ (i.e., too small triangles) can affect computational costs. In {\texttt{inla.mesh.2d}} function, we can control $m$ by choosing the maximum triangle edge length for the inner and outer regions of the mesh via $\texttt{max.edge}$ option. For the normalized spatial domain $\mathcal{D} = [0,1]^2$, we observe that edge lengths between 0.1 to 0.5 provide robust results. 
\subsection{Theoretical Justifications}
\label{sec:theo}

Consider the PSFoFR with the following representations.
\begin{enumerate}
\item Let $\Psi \in \mathcal{L}(\mathcal{V}) \otimes \mathcal{L}(\mathcal{T})$ be the true regression function in the data space defined as $\Psi=\sum_{k=1}^\infty\sum_{g=1}^\infty \psi_{gk}\phi_{gk}$.
\item Let ${\Psi}_{g_nk_n}=\sum_{k=1}^{k_n}\sum_{g=1}^{g_n} \psi_{gk}\phi_{gk}$ be the truncated regression function in the data space with large enough $g_n$ and $k_n$ so that $\sum_{k>k_n} \sum_{g>g_n}\psi_{gk} \phi_{gk}$  is negligible. And let $\widetilde{\boldsymbol{\psi}}_{g_nk_n}=\begin{pmatrix}
\psi_{11} &\cdots&\psi_{1k_n} \\
\vdots & \ddots & \vdots \\
\psi_{g_n1}&\cdots&\psi_{g_nk_n}\\
\end{pmatrix}$.
\item Let $\widetilde{\boldsymbol{\psi}}_{g_nk_n}^{(u)}=\begin{pmatrix}
\psi_{11}^{(u)}&\cdots&\psi_{1k_n} ^{(u)}\\
\vdots & \ddots & \vdots \\
\psi_{g_n1}^{(u)}&\cdots&\psi_{g_nk_n}^{(u)}\\
\end{pmatrix}$ be the posterior sample from $u$th iteration of MCMC.
\item  Let $\Psi_{g_nk_n}^{(u)}=\sum_{k=1}^{k_n}\sum_{g=1}^{g_n} \psi_{gk}^{(u)}\phi_{gk} \in \mathcal{L}(\mathcal{V}) \otimes \mathcal{L}(\mathcal{T})$ be a function constructed using MCMC samples.
\item Let  $\widehat{\Psi}_{g_nk_n}=\sum_{k=1}^{k_n}\sum_{g=1}^{g_n} \widehat{\psi}_{gk}\phi_{gk}$ be the estimated regression function where $\widehat{\psi}_{gk}=\frac{1}{U}\sum_{u=1}^U \psi_{gk}^{(u)}$. And let $\widehat{\boldsymbol{\psi}}_{g_n k_n} =$$\begin{pmatrix}
\widehat{\psi}_{11}&\cdots&\widehat{\psi}_{1k_n}\\
\vdots & \ddots & \vdots \\
\widehat{\psi}_{g_n1}&\cdots&\widehat{\psi}_{g_nk_n}\\
\end{pmatrix}.$ \par
\end{enumerate} 
\vspace{0.5cm}

\begin{assumption}
\label{assumption3.1}
The presumptions are summarized as follows. 
    \begin{enumerate}
    \item \textit{The basis systems $\{\phi_g\}_{g=1}^{\infty}$ and $\{\phi_k\}_{k=1}^{\infty}$ are fixed orthonormal system of $\mathcal{L}(\mathcal{V})$ and $\mathcal{L}(\mathcal{T})$. From this, we can a construct tensor basis system $\{\phi_{gk}\}_{g,k=1}^{\infty} \in \mathcal{L}(\mathcal{V}) \otimes \mathcal{L}(\mathcal{T}).$}
    \item \textit{The truncation $k_n \rightarrow \infty$ as $n \rightarrow \infty$, and truncation $g_n \rightarrow \infty$ as $n \rightarrow \infty$.}
    \item \textit{The grid points of the response curve $\lbrace t \rbrace_{t=1}^{n_t}$ are fixed and for any interval $A,B \in \mathcal{T}$, there exist constants $c_2 \geq c_1>0$ satisfy the following.
    \begin{enumerate}
    \item $c_1 n_t |A|-1 \leq \sum_{t=1}^{n_t}\mathbf{1}(t \in A) \leq \max\lbrace c_2 n_t, 1\rbrace$,
    \item $c_1 n^2_t |A||B|-1 \leq \sum_{t=1}^{n_t}\sum_{t'=1}^{n_t} \mathbf{1}(t \in A)\mathbf{1}(t' \in B)\leq \max\lbrace c_2n_t|A||B|,1\rbrace$,
    \end{enumerate}
    where $|\cdot|$ is the length of the interval. The same condition holds for the grid points of the covariate curve $\lbrace r \rbrace_{r=1}^{n_v}$.
    }
    \item \textit{$\Psi, \Psi_{g_nk_n}, \Psi_{g_nk_n}^{(u)},$ and $\widehat{\Psi}_{g_nk_n}$ are defined on the same probability space.}
    \end{enumerate}
\end{assumption}

Assumption~\ref{assumption3.1} is not restrictive. The first assumption is that we are working under a given fixed tensor product basis system defined in $\mathcal{L}(\mathcal{V}) \otimes \mathcal{L}(\mathcal{T})$. In the second assumption, $k_n, \ g_n$ depends on $n$, and both go to infinity when $n$ goes to infinity, which is commonly assumed. \cite{muller2005fglm} showed that the approximation error from the finite truncation $\sum_{k>k_n} \sum_{g>g_n}\psi_{gk} \phi_{gk}$ vanishes asymptotically by assuming that the number of basis ($k_n,g_n)$ increases as $n \rightarrow \infty$. We also follow this strategy to have convergence of $\Psi_{g_nk_n}$ to $\Psi$ over the function space $\mathcal{L}(\mathcal{V}) \otimes \mathcal{L}(\mathcal{T})$. The third assumption guarantees that each curve is densely observed over the entire function domain so that we can smooth out measurement noise; this is a standard setting in the FDA \citep{ramsay2005, kokoszka2017introduction} and is also assumed in \cite{zhou2022theory}. Note that the grid points do not have to be equally spaced for our implementation. The last assumption is for implementing computation on these functions. Note that $\|\cdot \|$ denotes the norm in $\mathcal{L}(\mathcal{V}) \otimes \mathcal{L}(\mathcal{T})$, and $\|\cdot\|_{\mathbb{R}^{g_n \times k_n}}$ is the norm in $\mathbb{R}^{g_n \times k_n}$. Under these assumptions, Theorem~\ref{3.1} shows that the function obtained through the MCMC algorithm converges in distribution to the true function. From this, we can construct the corresponding credible intervals in Section~\ref{sec:simulcre}.

\begin{theorem}
\label{3.1}
Consider the sequence of random functions $\{\Psi_{g_nk_n}^{(u)}\}_{n,u} \in \mathcal{L}(\mathcal{V})\otimes  \mathcal{L}(\mathcal{T})$. Then, $\Psi_{g_nk_n}^{(u)} \xrightarrow[]{D} \Psi$ as $n \rightarrow \infty$ and $u \rightarrow \infty$.
\end{theorem}

\begin{proof}
In terms of the total variation distance, an ergodic Markov chain $\left\{\widetilde{\boldsymbol{\psi}}_{g_n k_n}^{(u)}\right\}_{u \in \mathbb{N}}$ converges to the stationary posterior distribution $\pi(\widetilde{\boldsymbol{\psi}}|\textbf{Y})$. Therefore, we have $\widetilde{\boldsymbol{\psi}}_{g_n k_n}^{(u)} \xrightarrow[]{D} \widetilde{\boldsymbol{\psi}}_{g_n k_n}$ as $u \rightarrow \infty$ for every $g_n$ and $k_n \in \mathbb{N}$. Then, by the continuous mapping theorem $\Psi_{g_nk_n}^{(u)} \xrightarrow[]{D} \Psi_{g_nk_n}$ for every $g_n$ and $ k_n$. Since $\Psi_{g_nk_n} \xrightarrow[]{D} \Psi$ as $n \rightarrow \infty$, we have $\Psi_{g_nk_n}^{(u)} \xrightarrow[]{D} \Psi$ as $n \rightarrow \infty$ and $u \rightarrow \infty$. 
\end{proof}

Theorem~\ref{3.2} shows that the estimated regression function $\widehat{\Psi}_{g_nk_n}^{(u)}$ converges in probability to the true regression function $\Psi$ as the $n$ (sample size) and $u$ (MCMC sample size) go to infinity. From this, we can guarantee that the point estimate $\widehat{\Psi}_{g_nk_n}^{(u)}$ obtained from the posterior mean is consistent, which is stronger than converge in distribution. 

\begin{theorem}
\label{3.2}
    Consider the estimated regression function $\widehat{\Psi}_{g_nk_n}$. Under aforementioned assumptions, we have $\widehat{\Psi}_{g_nk_n} \xrightarrow{P} \Psi$ as $n \rightarrow \infty$ and $u \rightarrow \infty$.
\end{theorem}

\begin{proof}
First, we fix $g_n$ and $k_n$. Since $\left\{\widetilde{\boldsymbol{\psi}}_{g_n k_n}^{(u)}\right\}_{u \in \mathbb{N}}$ is an ergodic and stationary Markov chain, we have $\widehat{\boldsymbol{\psi}}_{g_n k_n} \xrightarrow{P} \widetilde{\boldsymbol{\psi}}_{g_nk_n}$ as $u \rightarrow \infty$ by the weak law of large numbers. As
$$P(\|\Psi_{g_nk_n}^{(u)} - \Psi_{g_nk_n}\|^2 \geq \frac{\epsilon}{2}) = P(\|\widehat{\boldsymbol{\psi}}_{g_n k_n} - \widetilde{\boldsymbol{\psi}}_{g_n k_n}\|^2_{\mathbb{R}^{g_n \times k_n}} \geq \frac{\epsilon}{2}) \rightarrow 0,$$
we have $\Psi_{g_nk_n}^{(u)} \xrightarrow{P} \Psi_{g_nk_n}$ as $u \rightarrow \infty$ for every $g_n, k_n \in \mathbb{N}$. Since $\Psi_{g_nk_n}$ is truncated version of $\Psi$, we have $\Psi_{g_nk_n} \xrightarrow{P} \Psi$ as $n \rightarrow \infty$ (i.e. $g_n, k_n \rightarrow \infty$), resulting in $P(\|\Psi_{g_nk_n} - \Psi\|^2 \geq \frac{\epsilon}{2}) \rightarrow 0$. Then,
\begin{align*}
P(\|\Psi_{g_nk_n}^{(u)} - \Psi\|^2 \geq \epsilon) &\leq P(\|\Psi_{g_nk_n}^{(u)} - \Psi_{g_nk_n}\|^2 + \|\Psi_{g_nk_n} - \Psi\|^2  \geq \epsilon) \\
&\leq P(\|\Psi_{g_nk_n}^{(u)} - \Psi_{g_nk_n}\|^2 \geq \frac{\epsilon}{2}) + P(\|\Psi_{g_nk_n} - \Psi\|^2 \geq \frac{\epsilon}{2}) \\
&\rightarrow 0+0.
\end{align*}
as $n \rightarrow \infty$ and $u \rightarrow \infty$.
\end{proof}

The ergodicity of the Markov chain $\{\widetilde{\boldsymbol{\psi}}^{(u)}_{g_nk_n}\}$ is important to ensure convergence. For PSFoFR, the Metropolis-Hastings algorithm with a random walk proposal produces an ergodic Markov chain. 
This is directly from Corollary 2 in \cite{tierney1994markov} and Lemmas 1.1, 1.2 in \cite{mengersen1996rates}, which implies that an aperiodic, irreducible, positive Harris recurrent Markov chain is ergodic. In our MCMC implementation, we used a random walk proposal in {\tt{nimble}}. Since we represent a spatial dependence with a fixed basis matrix $\mathbf{P}$ with independent errors $\widetilde{\boldsymbol{\epsilon}}$, Theorems~\ref{3.1}, \ref{3.2} do not require more than the ergodicity of the standard linear models. However, investigating the convergence rate of the algorithm is a challenging problem because these are sampler specific. In generalized linear mixed models (GLMMs), \cite{christensen2001geometric} showed the geometric ergodicity for random walk Metropolis-Hastings and Langevin-Hastings. \cite{roman2015geometric} showed the geometric ergodicity of the block Gibbs Markov chain in GLMMs using a geometric drift condition. Extending such results to functional regression would be an interesting direction for future research. 

\subsection{Simultaneous Credible Intervals for Functional Parameters}
\label{sec:simulcre}

We construct simultaneous credible intervals \citep{meyer2015bayesian} for both the predicted curves and the regression function. Here, we describe the credible intervals in terms of the regression function without loss of generality. The procedure can be described as follows. 
\begin{enumerate}
\item We construct $\Psi_{g_nk_n}^{(u)}(r,t)=\sum_{k=1}^{k_n}\sum_{g=1}^{g_n} \psi_{gk}^{(u)}\phi_{gk}(r,t)$ from the MCMC samples $(\psi^{(u)}_{gk})_{g_n \times k_n}$, for $u=1,\cdots,U$.
\item We evaluate $\Psi_{g_nk_n}^{(u)}(r, t)$ for each grid point $r = 1,\cdots, n_v$ and $t = 1, \cdots, n_t$, where $1 < 2 < \cdots < n_v$ and $1 < 2 < \cdots < n_t$, respectively.
\item We obtain ${M_\alpha}$ from the $(1-\alpha)$ sample quantile of 
 $$Z^{(u)} = \max_{1 \leq r \leq n_v, 1 \leq t \leq n_t}\left|\frac{\Psi_{g_nk_n}^{(u)}(r, t)-\widehat{\Psi}_{g_nk_n}(r, t)}{\text{SD}(\Psi_{g_nk_n}(r, t))}\right|,$$
where $\widehat{\Psi}_{g_nk_n}(r, t)$ and $\text{\text{SD}}(\Psi_{g_nk_n}(r, t))$ be the sample mean and sample standard deviation of $\left\{\Psi_{g_nk_n}^{(u)}(r, t)\right\}_{u=1}^U$, respectively.
\item We construct $(1-\alpha)$ simultaneous credible intervals can be obtained as
$$I(\alpha, r, t)=\widehat{\Psi}_{g_nk_n}(r, t)\pm {M_\alpha} \text{SD}(\Psi_{g_nk_n}(r, t)).$$
\end{enumerate}
In Step 1 and Step 2, we reconstruct the regression functions from posterior samples and evaluate them at grid points; such steps are illustrated in Figure~\ref{outline_credible}(a). In Step 3, we standardize the reconstructed regression functions $\left\{\Psi_{g_nk_n}^{(u)}(r,t)\right\}_{u=1}^U$ (Figure~\ref{outline_credible}(b)) and obtain the maximum absolute value across the entire grid (Figure~\ref{outline_credible}(c)). Then, we use the $(1-\alpha)$ quantile of the maximum variability to compute the simultaneous credible intervals. This implies that the simultaneous credible intervals consider the distribution of the largest variability across the entire grid (i.e., $\left\{Z^{(u)}\right\}_{u=1}^{U}$) instead of only using point-wise variability. 
\begin{figure}[htbp]
     \centering
         \centering
         \includegraphics[width=\linewidth]{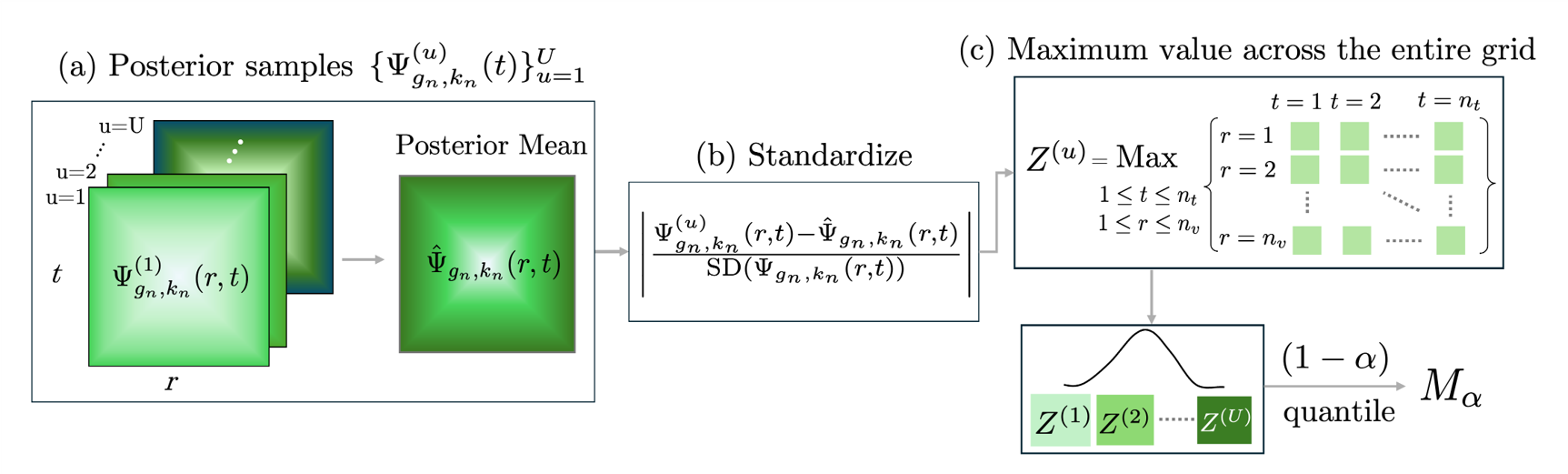}
         \caption{A graphical illustration for constructing simultaneous credible intervals. }
         \label{outline_credible}
\end{figure}

To detect the significant region of the regression function, we also compute Simultaneous Band Scores (SimBaS) \citep{meyer2015bayesian} by converting simultaneous intervals. As in \cite{meyer2015bayesian}, we first construct $I(\alpha, r, t)$ for multiple levels of $\alpha$. Then, for each $(r, t)$, we determine the minimum $\alpha$ at which each $I(\alpha, r, t)$ excludes zero (i.e., $P_\text{SimBaS}(r, t) = \min \{\alpha: 0 \notin I(\alpha, r, t)\}$). SimBaS can be computed as 
$$P_\text{SimBaS}(r, t) = \frac{1}{U} \sum_{u=1}^U 1 \left\{\left|\frac{\widehat{\Psi}_{g_nk_n}(r, t)}{\text{SD}(\Psi_{g_nk_n}(r, t))}\right| \leq  Z^{(u)} \right\}. $$
Here, we can specify the level $\alpha$ and identify the significant region $(r, t)$ that satisfies $P_\text{SimBaS}(r, t)<\alpha$. This is equivalent to checking whether $I(\alpha, r, t)$ covers zero at a chosen $\alpha$-level. 

There is a connection between the simultaneous credible intervals \citep{meyer2015bayesian} and the contour avoiding regions \citep{bolin2015excursion}. The contour avoiding region is defined as the union of two sets: one where the set is significantly above the reference level and another where it is significantly below the reference level. If we set the reference level as 0, we can detect the statistically significant region. Specifically, we can define the contour avoidance function \citep{bolin2015excursion} as
$$F_0(r,t)=\sup\{1-\alpha: (r,t) \in E_{0,\alpha}\},$$ 
where $E_{0,\alpha}$ is the contour avoiding region with probability $1-\alpha$ at reference level 0. By specifying $\alpha = 0.05$, we can identify the statistically significant region $(r,t)$ that satisfies $F_0(r,t) > 1-\alpha$. We provide computational details about contour avoiding regions for the regression function in the supplementary material (Appendix C). We compute the contour avoiding regions using {\tt{excursions}} package \citep{excursion2023} in {\tt R}. We observe that the simultaneous credible intervals and the contour avoiding regions show similar patterns in general, though the contour avoiding regions are slightly wider. Note that the direct comparison is challenging because the package assumed a latent Gaussian process for the Monte Carlo samples $\left\{\Psi_{g_nk_n}^{(u)}(r,t)\right\}_{u=1}^U$ in its computation. 

\section{Simulation Study}
\label{sec:simulation}

In this section, we apply the proposed method to simulated spatial functional datasets. To represent the performance of PSFoFR, we compare it with FoFR (i.e., independent model), SFoFR and standard UK. We implement PSFoFR, SFoFR, and FoFR in {\texttt{nimble}} \citep{de2017programming}, and implement UK through {\texttt{R}} package {\texttt{fdagstat}} \citep{fdagstat_2017}. We show the main results with B-spline basis functions and provide fitted results using the Fourier and FPC basis functions in the supplementary material (Appendix D). We also provide a simulation study with a different regression function in the supplementary material (Appendix E). We observe that the performance of PSFoFR is robust in different scenarios. We run MCMC algorithms until the Monte Carlo standard errors \citep{jones2006fixed} for PSFoFR are at or below 0.04 (70,000 iterations). All codes were run on 16-core AMD Ryzen 9 7950X processors. The source code can be downloaded from \url{https://github.com/codinheesang/PSFoFR}. The simulation procedures are summarized as follows. 
\begin{enumerate}
\item We simulate 1,000 locations in the $\mathcal{D} = [0,1]^2$ domain. For each location,  we generate $\{x_{ig}\}_{1,000 \times 15}$ from a normal distribution with the mean as the $x$-axis of the location and a variance of 1. By using 15 B-spline basis functions, we create $X_{\boldsymbol{s}_i}(r)=  \sum_{g=1}^{15} x_{ig}\phi_g(r)$ for $i=1,\cdots,1,000$.
\item Then we generate $\{{w}_{ik}\}_{1,000 \times 15}$ from a Gaussian process with a mean 0 and a Mat\'{e}rn covariance, where the distance matrix is calculated from the simulated locations. We use a variance $0.5$, range $0.2$, and a smoothness 0.5. Similarly, we create $W_{\boldsymbol{s}_i}(t)=  \sum_{k=1}^{15} w_{ik}\phi_k(t)$ for $i=1,\cdots,1,000$.
\item Following \cite{meyer2015bayesian}, we use the true regression function defined over the $[0,225] \times [0,225]$ domain as
$$\Psi(r,t) = \frac{7}{500} \frac{1}{\sqrt{0.006\pi}} \exp\left[-\frac{1}{0.006} \left(\frac{t-r}{225}\right)^2\right].$$
Then, from the finite basis expansion $\Psi(r, t) = \sum_{k=1}^{15}\sum_{g=1}^{15}\psi_{gk}\phi_{gk}(r,t)$, we obtain $\psi_{gk}$ for $g,k=1,\cdots,15$.
\item We simulate $\widetilde{\mathbf{Y}}=(y_{ik})_{1,000\times 15}$ from \eqref{fregmodeltrunc} and obtain $Y_{\boldsymbol{s}_i}(t) = \sum_{k=1}^{15} y_{ik}\phi_k(t)$ for $i=1,\cdots,1,000$. When we do not consider the measurement error, we simulate $\{y_{ik}\}_{1,000\times 15}$ from $\widetilde{\mathbf{Y}} = \widetilde{\mathbf{X}} \widetilde{\boldsymbol{\psi}}+\widetilde{\mathbf{W}}$. 
\end{enumerate}


We use $n=700$ samples for training and $n_{cv}=300$ samples for prediction. For UK, we use a default Gaussian model for variogram fitting and plug in the simulated $\{x_{ig}\}_{n \times 15}$ as covariates; note that UK cannot directly model functional covariates. For PSFoFR, SFoFR, and FoFR, we run the MCMC algorithm for 70,000 iterations, with 50,000 discarded for burn-in, and 1,000 thinned samples obtained from the remaining 20,000. We choose the number of basis functions $k_n$ and $g_n$ by minimizing the generalized cross-validation (GCV) error. For PSFoFR, we use 1,520 triangular meshes to construct the projection matrix over the continuous spatial domain. We study PSFoFR with different rank values $p=50, 100, 150$, corresponding to 5\%, 10\%, and 15\% of the sample size. To assess the performance of the algorithms, we compare the mean square prediction error (MSPE) defined as $\sqrt{\frac{1}{n_{cv}n_t}\sum_{\forall \boldsymbol{s}^{\ast},t} (Y_{\boldsymbol{s}^{\ast}}(t)-\widehat{Y}_{\boldsymbol{s}^{\ast}}(t))^2}$ with $n_t=225$ grid points over the $[0,225]$ domain. For functional regression models, we also compute the mean square error (MSE) of the regression function as $\sqrt{\frac{1}{n_t n_v}\sum_{\forall r,t}(\Psi(r,t)-\widehat{\Psi}(r,t))^2}$ with $n_t \times n_v = 225 \times 225$ grid points over the $[0,225]\times[0,225]$ domain.

\newcolumntype{M}[1]{>{\centering\arraybackslash}m{#1}}
\newcolumntype{P}[1]{>{\centering\arraybackslash}p{#1}}
\begin{table}[htbp]
  \centering
    \begin{tabular}{M{30mm}M{20mm}M{20mm}M{20mm}M{20mm}} \toprule
    \textbf{Model} &\multicolumn{1}{M{20mm}}{\textbf{MSPE}}&\multicolumn{1}{M{20mm}}{\textbf{Coverage}} &\multicolumn{1}{M{20mm}}{\textbf{MSE }} &\multicolumn{1}{M{20mm}}{\textbf{Time}} \\ \midrule
        PSFoFR $(p=50)$ & 0.106&  99   &0.004& 34 \\ \cmidrule(l){2-5}
        PSFoFR $(p=100)$ & 0.099 &  99  &0.004& 39 \\ \cmidrule(l){2-5}
        PSFoFR $(p=150)$  & 0.098 &  99  &0.004& 43  \\ \cmidrule(l){1-5}
        SFoFR & 0.091 &  98 &0.004& 287  \\  \cmidrule(l){1-5}
        FoFR & 0.182 &   99  &0.004& 19  \\  \cmidrule(l){1-5}
        UK   &0.123  &-&-& $< 1$ \\  \bottomrule
    \end{tabular}
    \caption{Inference results for the simulated dataset. MSPE, prediction coverage (\%), MSE, and computing time (min) are reported.}
    \label{tab1}
\end{table}

\begin{figure}[htbp]
     \centering
        \centering
         \includegraphics[width=\linewidth]{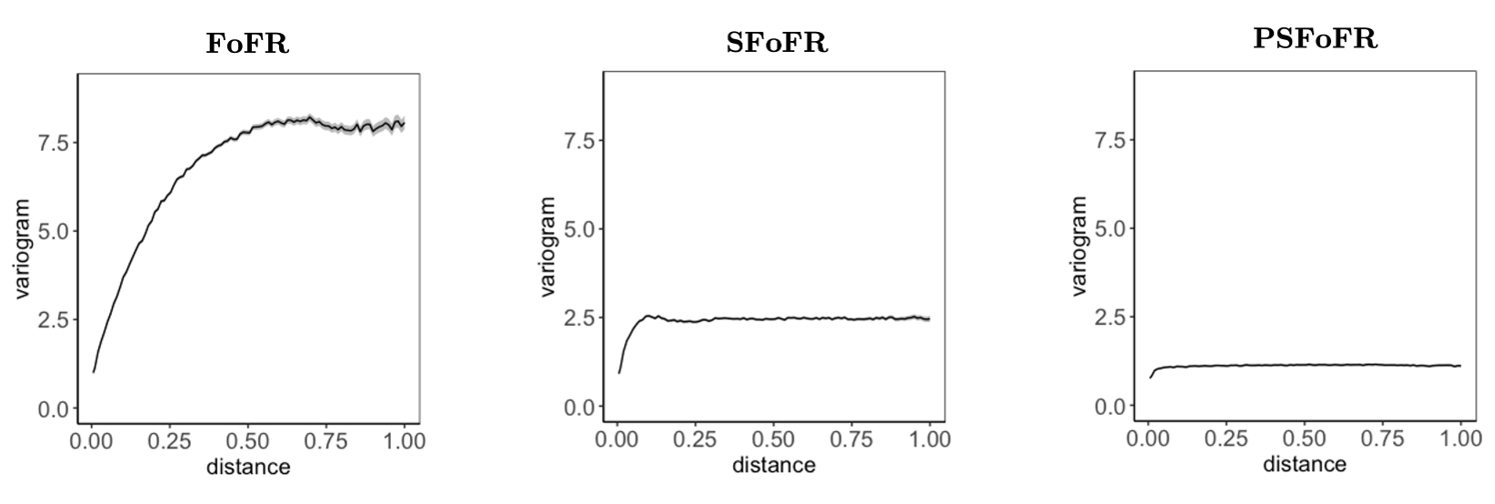}
         \caption{Fitted trace-variograms for the simulated datasets. Black lines indicate the point estimates obtained from posterior means. Grey areas represent the corresponding 95\% credible intervals.}
         \label{fig:SimulVario}
\end{figure}

Table~\ref{tab1} shows that functional regression approaches can accurately estimate $\Psi(r,t)$ with a small MSE, while UK cannot. Compared to SFoFR, PSFoFR is computationally much faster due to the projection step. The performance of PSFoFR is similar across different ranks. This implies that choosing $p$ as 5\% to 10\% of the samples can quickly provide accurate inference results, as suggested in \cite{hughes2013dimension, lee2022picar}. We observe that FoFR has the highest MSPE value. To analyze the correlation among residuals, we compute the trace-variogram through {\tt fdagstat} package. Figure~\ref{fig:SimulVario} indicates that residuals from FoFR are more correlated than those from PSFoFR and SFoFR. To quantify uncertainties, we compute 95\% simultaneous credible intervals of the trace-variograms. We observe that the uncertainties about the variogram estimates are small across all models.

\begin{figure}[htbp]
     \centering
         \centering
         \includegraphics[width=\linewidth]{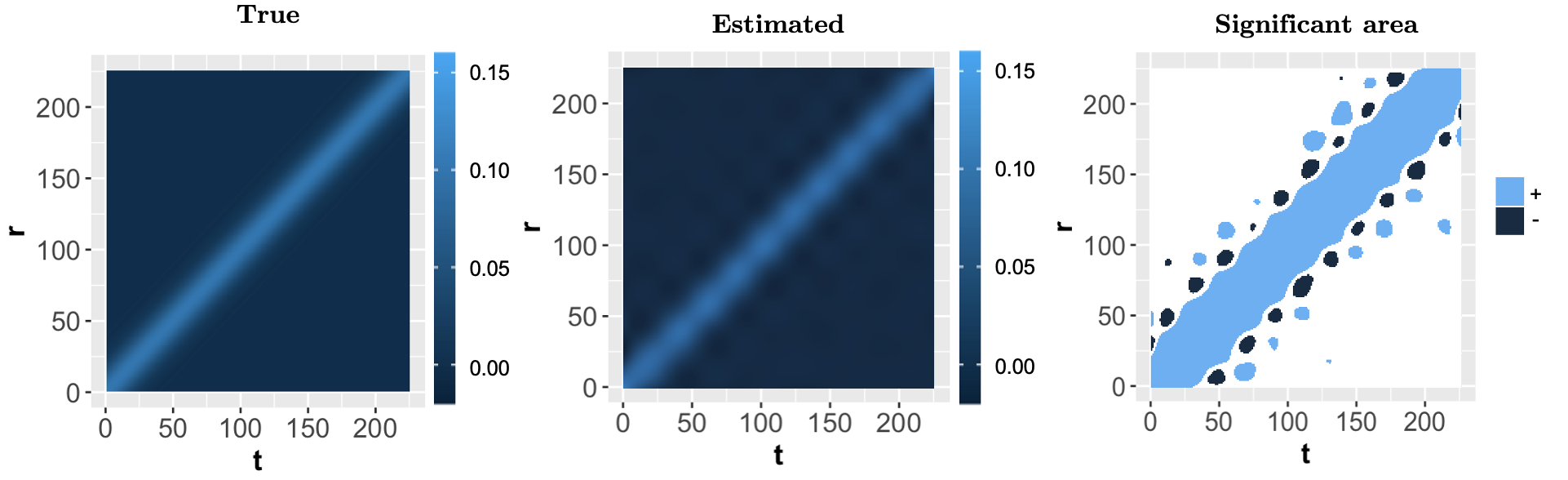}
         \caption{The estimated $\Psi(r,t)$ obtained from the posterior mean of PSFoFR. Significant areas in the functions are detected from $\alpha=0.05$. Sky blue colors represent positive significant areas, whereas dark blue colors represent negative significant areas. 
         }
         \label{fig:SimulParam}
\end{figure}

Here, we visualize the inference results from PSFoFR with $p=50$ since the results from other rank choices are quantitatively similar. Figure~\ref{fig:SimulParam} indicates that the estimated regression functions from PSFoFR are similar to the true regression function. Furthermore, the significant areas detected from the 95\% simultaneous credible intervals and SimBaS are well aligned with the important regions of the true function. Note that UK approaches cannot quantify such relationships between functional variables. 

\begin{figure}[htbp]
     \centering
         \centering
         \includegraphics[width=\linewidth]{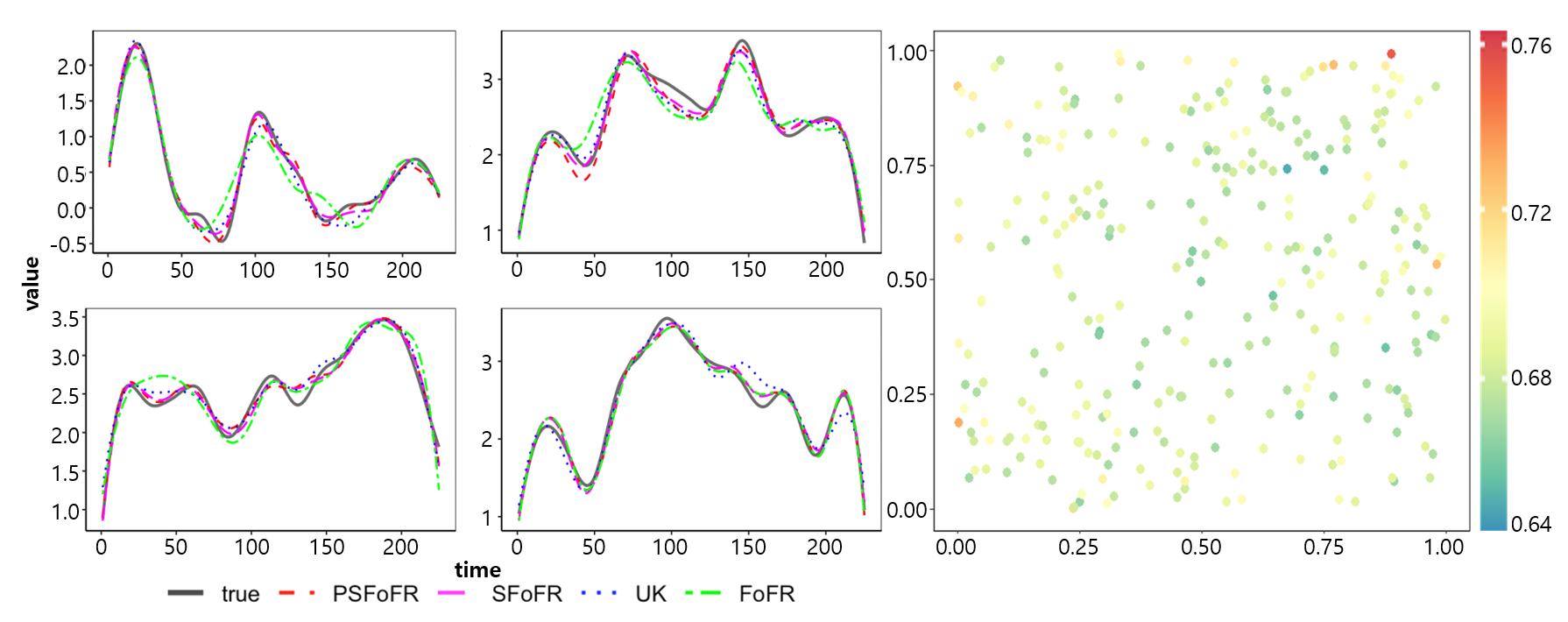}
         \caption{Visualizations of predicted curves (left panel) and prediction uncertainties (right panel) from PSFoFR. }
         \label{fig:SimulKrig}
\end{figure}

Figure~\ref{fig:SimulKrig} provides four predicted curves from 300 test observations. All methods can capture the trend of the true response functions, while FoFR shows slight differences, which are aligned with the results in Table~\ref{tab1}. We compute the mean coverage for 300 predicted curves. For each curve, we calculate the proportion of $n_t$ grid points where the 95\% simultaneous credible intervals include the true curve value. Then, we take the average of the proportions computed from the 300 predicted curves to obtain mean coverage. Table~\ref{tab1} indicates that the credible intervals from all models can cover the true curves, although the coverages are slightly higher than the nominal rate of 95\%. Note that coverage of FoFR is also high even with lower MSPE compared to other spatial models, which is due to the wider credible intervals. We also provide the visualization of the simultaneous credible intervals for predicted curves from each method in the supplementary material (Appendix F). Figure~\ref{fig:SimulKrig} also visualizes the prediction uncertainties for each spatial location. We calculate the length of the 95\% credible intervals for every time grid and compute the mean of them for each test spatial location. In general, prediction uncertainties are larger at the boundaries of the region, which are more difficult to predict due to a lack of neighboring information.


\section{Application}
\label{sec:appli}

In this section, we apply our methods to two real data examples with spatial functional variables: (1) PM2.5 data (continuous domain) and (2) mobility data (discrete domain). For both cases, 
the proposed approach can estimate regression functions and detect the significant region of the estimated functions. Furthermore, our method can provide more (or comparable) accurate kriging results than UK in the point-level data example (PM 2.5 data). 

\subsection{PM2.5 Data} 
\label{sec:pm25}

Particle pollution from fine particulates (PM2.5) is a crucial indicator of air quality. Long-term exposure to PM2.5 can increase the risk of health problems, such as heart disease, low birth weight, and asthma. Note that PM2.5 is influenced by various environmental factors, meteorological conditions, and other pollutant emissions. It is well known that anthropogenic emissions of nitrogen oxides (NOx) and carbon monoxide (CO) can affect PM2.5 concentrations on an intercontinental scale \citep{Leibensperger2011PMNox}. Especially, NOx emissions affect air quality since the reduction of NOx emission consecutively reduces $\text{NO}_{3^{-}}$ and  $\text{NH}_{4^{+}}$ which compose PM2.5 particles. Therefore, studying the functional relationship between PM2.5 and NOx would be useful for establishing public health policy. Furthermore, predicting PM2.5 curves for unobserved locations is crucial for public health management. 

Here, we study the PM2.5 dataset collected from the Ministry of Environment Japan \citep{JapanPMData} observed at 991 monitoring stations. The dataset contains daily mean values of PM2.5 (January 20, 2024, to January 31, 2024) and NOx values (January 8, 2024, to January 19, 2024); therefore, $n_t=n_v=12$ in this example. We use the log PM2.5 values as a functional response and the NOx values as a functional covariate in the proposed model. 
For UK, we use a default Gaussian model for the variogram fitting. For PSFoFR, we use $k_n=11$ and $g_n=11$ orthonormalized B-spline basis functions; the number of basis functions is chosen by minimizing the GCV error. To construct the projection matrix over the continuous spatial domain, we use 1,719 triangular meshes and set $p$ as about 5\% of the sample size (i.e., $p=46$). We run MCMC algorithms until the Monte Carlo standard errors \citep{jones2006fixed} for PSFoFR are at or below 0.04 (50,000 iterations). We use 638 samples for training and the remaining 273 samples for validation. For the same 50,000 iterations, PSFoFR takes 11.80 minutes for computation, while SFoFR takes 94 minutes, indicating that PSFoFR is computationally more efficient. Furthermore, to achieve the same Monte Carlo standard error criteria (0.04), SFoFR needs two times longer iterations; the mixing of the chain is slower than those from PSFoFR. FoFR takes about 5.8 minutes because the model does not consider spatial correlation. Compared to functional regression models, UK only takes about a second.

\begin{figure}[htbp]
     \centering
     \begin{subfigure}[b]{0.47\linewidth}
         \centering
        \medskip         
         \includegraphics[width=\linewidth]{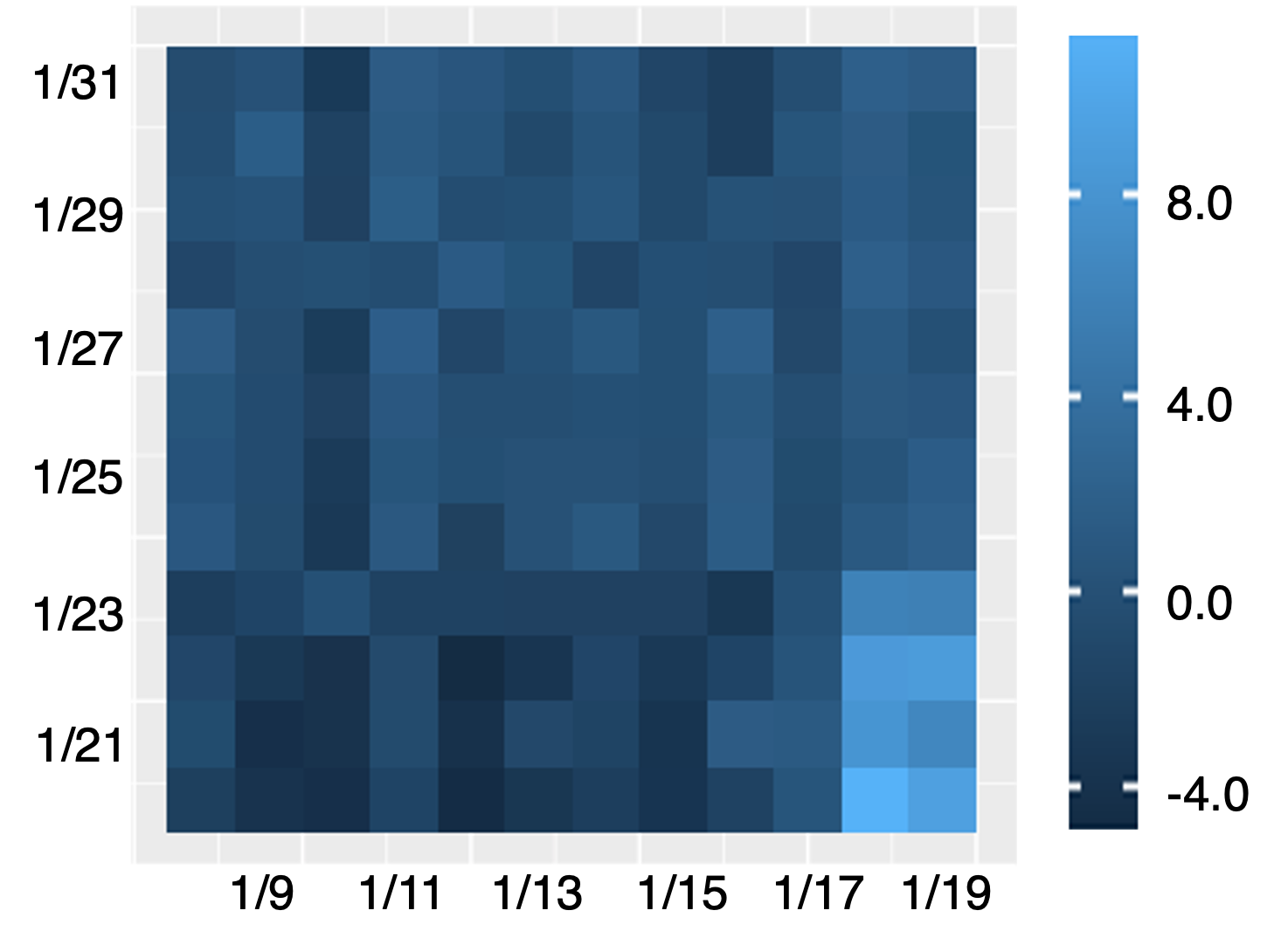}
         \caption{2D plot of estimated $\Psi(r,t)$}
         \label{fig:JapEst2d}
     \end{subfigure}
     \hfill
     \begin{subfigure}[b]{0.45\linewidth}
         \centering
         \includegraphics[width=\linewidth]{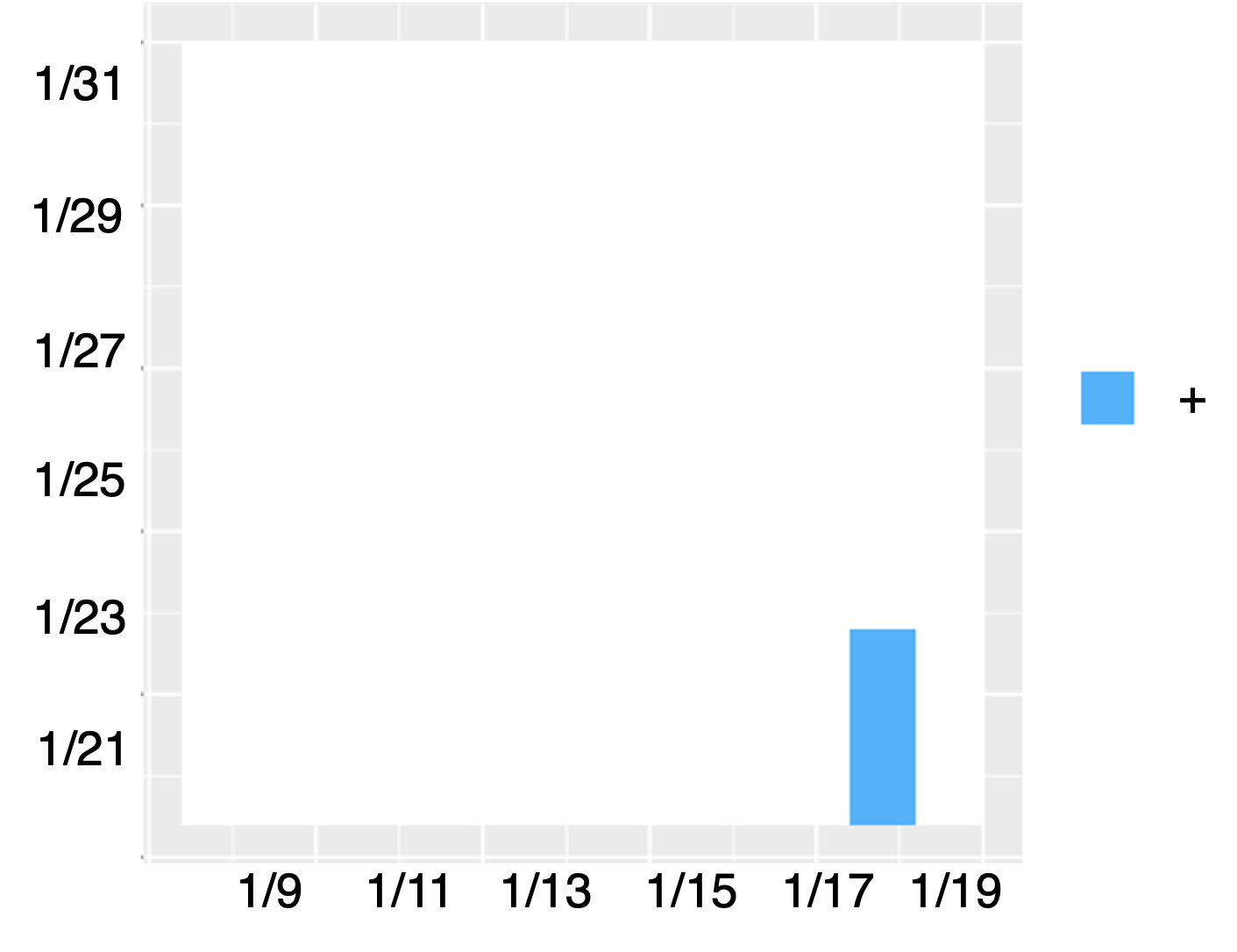}
         \caption{Significant areas}
         \label{fig:JapSigArea}
     \end{subfigure}
     \hfill
     \begin{subfigure}[b]{0.65\linewidth}
         \centering
         \includegraphics[width=\linewidth]{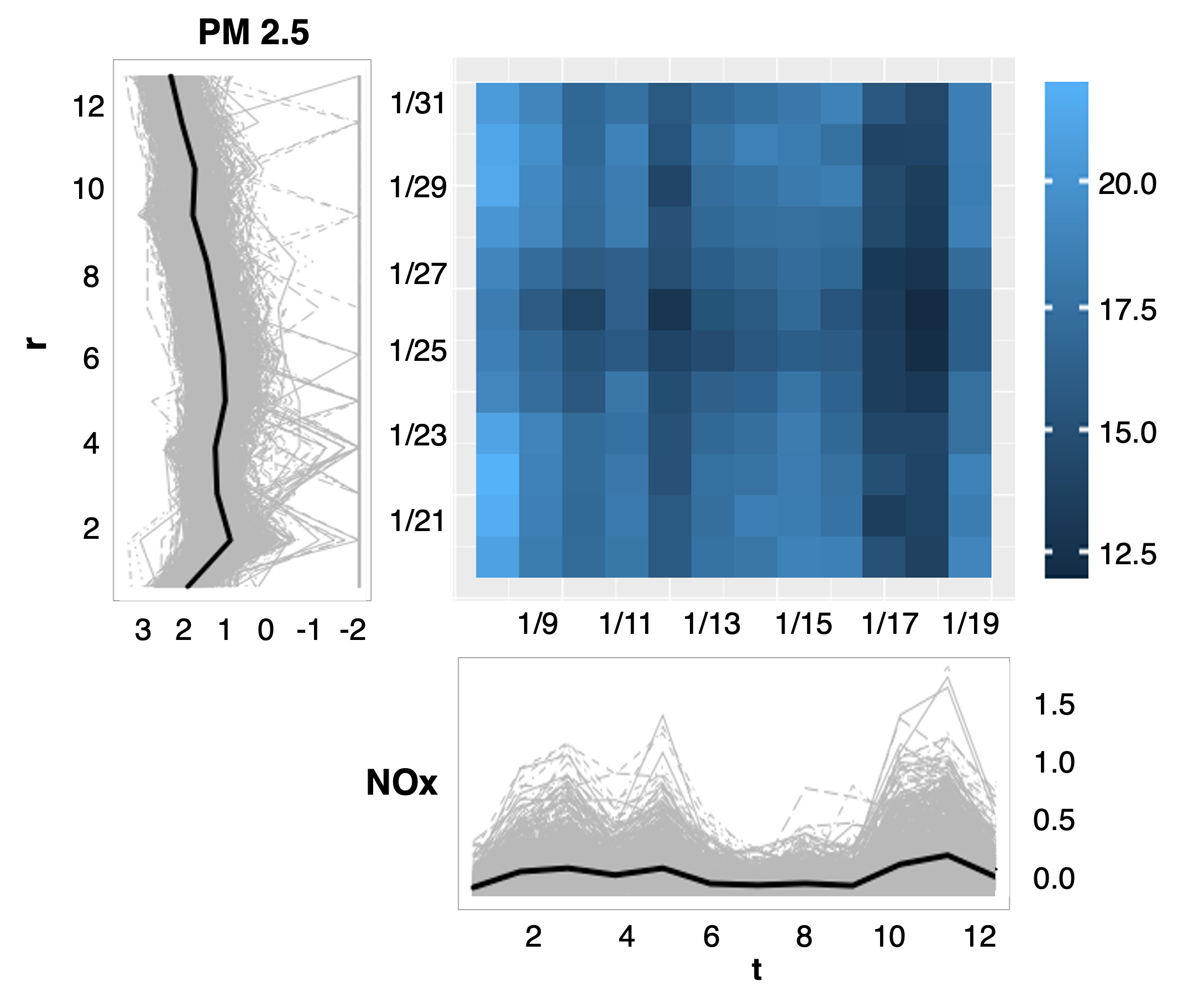}
         \caption{Uncertainty surface of estimated $\Psi(r,t)$}
         \label{fig:Japwidth}
     \end{subfigure}
        \caption{Japan PM2.5 data: Estimated $\Psi(r,t)$ obtained from the posterior mean of PSFoFR. Significant areas in the functions are detected from $\alpha=0.05$. Sky blue colors represent positive significant areas, whereas dark blue colors represent negative significant areas. (c) The uncertainty surface is obtained from the width of 95\% simultaneous credible intervals at each $(r, t)$. Black solid lines indicate the mean curves of the functional variables.    
        }
        \label{fig:Jap3Graphs}
\end{figure}

Figure~\ref{fig:Jap3Graphs} illustrates that the lower right corner parts of the estimated $\Psi(r,t)$ have high values, which implies that there is a positive relationship between NOx values and PM2.5 in the near term. Specifically, it shows that the NOx values on January 18th have a positive relationship with PM2.5 levels from January 20th to January 23rd. This result coincides with the previous studies \citep{fu2023significantly, huang2024assessing}. \cite{fu2023significantly} showed that reducing NOx emissions can effectively decrease PM2.5 levels during winter. \cite{huang2024assessing} demonstrated that reducing NOx emissions significantly decreases $\text{NO}_3^-$ and $\text{NH}_4^+$, which are major components of PM2.5, leading to a substantial reduction in PM2.5 levels during winter. However, such a relationship becomes insignificant in the long term but only lasts for a few days. The reason is that the atmospheric lifetime of NOx is less than half a day on average \citep{liu2016no}; therefore, NOx cannot significantly impact PM2.5 over a long period.

To quantify uncertainties for $\widehat{\Psi}(r,t)$, we provide an uncertainty surface obtained from the width of 95\% simultaneous credible intervals. Figure~\ref{fig:Jap3Graphs} shows that uncertainties of $\widehat{\Psi}(r,t)$ are small when $t=3,5,10,11$ (i.e.,  January 10th, 12th, 17th, and 18th). We observe that the variability of NOx curves these days is larger than that of the others. Note that the variability of $\psi_{gk}$ (i.e., the regression function coefficient) is inversely proportional to the variability of $x_{ig}$ (i.e., the functional covariate coefficient). See the conditional distribution of $\psi_{gk}$ derived in Appendix A in the supplementary material for further details. In general, NOx concentrations in major cities such as Toyko are higher than in other regions, resulting in significant regional variability in NOx concentrations across Japan \citep{kodama2002environmental}. However, we observe that the variability of NOx curves is low when $t=6,7,8,9$ (i.e., from January 13th to 16th). This is because the high wind speed in Tokyo leads to low NOx concentrations in the city, reducing regional variability. Therefore, uncertainties of $\widehat{\Psi}(r,t)$ are relatively large when $t=6,7,8,9$. 

\begin{figure}[htbp]
     \centering
         \centering
         \includegraphics[width=\linewidth]{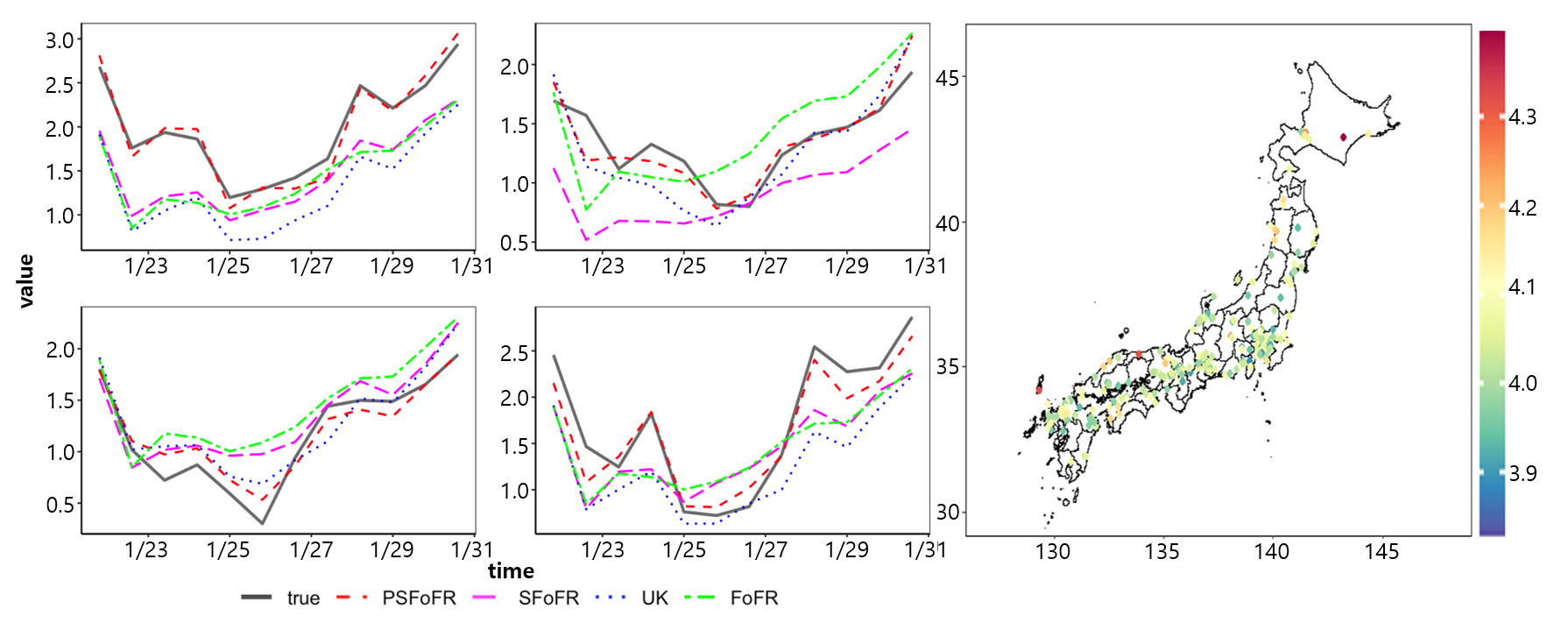}
         \caption{ Visualizations of predicted curves (left panel) and prediction uncertainties (right panel) from PSFoFR.}
         \label{fig:JapKrig}
\end{figure}

\begin{figure}[htbp]
     \centering
         \centering
         \includegraphics[width=0.95\linewidth]{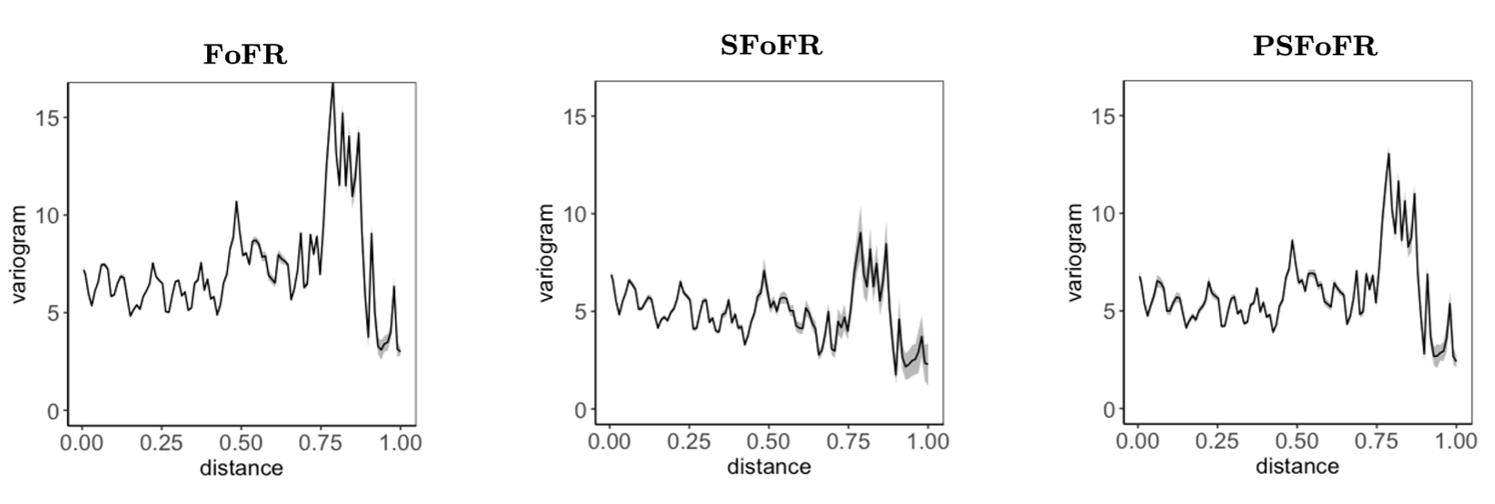}
         \caption{Fitted trace-variograms for the PM2.5 datasets. Black lines indicate the point estimates obtained from posterior means. Grey areas represent the corresponding 95\% credible intervals.}
         \label{fig:JapVario}
\end{figure}

Figure~\ref{fig:JapKrig} provides functional kriging results from PSFoFR and UK at four locations among 273 test observations. We observe that our methods provide more accurate kriging performance than UK. The MSPE of PSFoFR, SFoFR, UK, and FoFR are 0.780, 0.823, 0.806, and 0.838, respectively. PSFoFR shows the lowest MSPE value, while FoFR has the highest MSPE value. The MSPE value from SFoFR is relatively high because high-dimensional random effects $\widetilde{\mathbf{W}}$ cause the poor mixing of the chains. The trace-variograms indicate that residuals from FoFR show a more distinct pattern compared to spatial models (Figure~\ref{fig:JapVario}). We observe that uncertainties about the variogram estimates are small in general. As in the previous section, we compute the mean coverage for 273 predicted curves. We calculate the proportion of $n_t$ grid points where the 95\% simultaneous credible intervals include the true curve. The mean coverage of PSFoFR, SFoFR, and FoFR are 96.55\%, 89.56\%, and 96.46\%, respectively. In summary, PSFoFR provides the most accurate prediction with a coverage close to the 95\% nominal rate. We provide the visualization of the simultaneous credible intervals for predicted curves from each model in the supplementary material (Appendix F).
The right panel of Figure~\ref{fig:JapKrig} summarizes prediction uncertainties for PSFoFR. We compute the length of 95\% simultaneous credible intervals for every grid and take an average of them for each unobserved spatial location. We observe that the uncertainties are larger for the region with fewer observations, which is natural.


\subsection{Mobility Data}
\label{sec:mob}

Population mobility counts the number of individuals in a specific area during a given period based on the location information obtained from mobile phone base stations. It encompasses the total count of the mobile population, accounting for time overlaps, which include both the population not moving within the area and the population actively moving within the area. This population mobility data is sourced from a major mobile service provider in South Korea, SK Telecom (SKT), which holds 40\% of the total mobile telecommunication subscribers. The data is the form of origin-destination data, which has been preprocessed and expanded from the raw data to represent the entire population by adjusting the number of SKT users according to their market share in each region.

Data items are categorized by the daily outflow and inflow populations and segmented by gender, age group, and region. The "origin" is determined based on the previous month, specifically on the area where individuals spent the most time between 00:00 and 06:00. Based on this, individuals leaving the origin for other locations (destination) are defined as the outflow population, and those entering the origin location are defined as the inflow population. Geographic regions are divided into the city (Si), county (Gun), and district level (Gu), and the population count is based on individuals who have spent at least 2 hours in a specific area. If an individual visits multiple locations for over 2 hours daily, the same individual may be counted multiple times. However, duplicate counts are excluded if the same individual visits the same area on the same day, on a daily basis. 

In our analysis, we use 182 weekly mean values of total flow (outflow$+$inflow) populations from April 2019 to September 2022 by the areal unit as a functional response.  We use the population pyramid of each region as a functional covariate, which is collected from KOSIS (\url{https://kosis.kr/statHtml/statHtml.do?orgId=101&tblId=DT_1IN1509&conn_path=I2}). The dataset contains the proportion of the 16 age groups ($\leq$15, 15-19, 20-24, 25-29, 30-34, 35-39, 40-44, 45-49, 50-54, 55-59, 60-64, 65-69, 70-74, 75-79, 80-84 $\geq$85) in 250 geographic regions, city (Si), county (Gun), and district (Gu) in 2022. We use the population mobility values as a functional response and the proportion of the age group as a functional covariate in the proposed model. 
For PSFoFR, we use $k_n=39$ and $g_n=10$  orthonormalized B-spline basis functions; the number of basis functions are chosen by minimizing the GCV error. To construct the projection matrix over the discrete spatial domain, we use $p=24$, which is about 10\% of the sample size. We run MCMC algorithms until the Monte Carlo standard errors \citep{jones2006fixed} for PSFoFR are at or below 0.04 (50,000 iterations). 
Both PSFoFR and SFoFR take about 30 minutes for the same iterations. However, to achieve the same Monte Carlo standard error criteria (0.04), SFoFR needs four times longer iterations, indicating that the mixing of the chain is slower than those from PSFoFR. 


To analyze the spatial correlation among residuals, we use Moran's I, which is a commonly used exploratory method for areal data. For each time grid, we compute Moran's I values and average them. We observe that the averaged Moran's I values (p-values) are 0.000 (0.721), 0.162 (0.019), and 0.211 (0.002) for PSFoFR, SFoFR, and FoFR, respectively. PSFoFR can address the correlation in the data better than other methods, signified by the higher p-values.

\begin{figure}[htbp]
     \centering
     \begin{subfigure}[b]{0.47\linewidth}
         \centering
         \includegraphics[width=\linewidth]{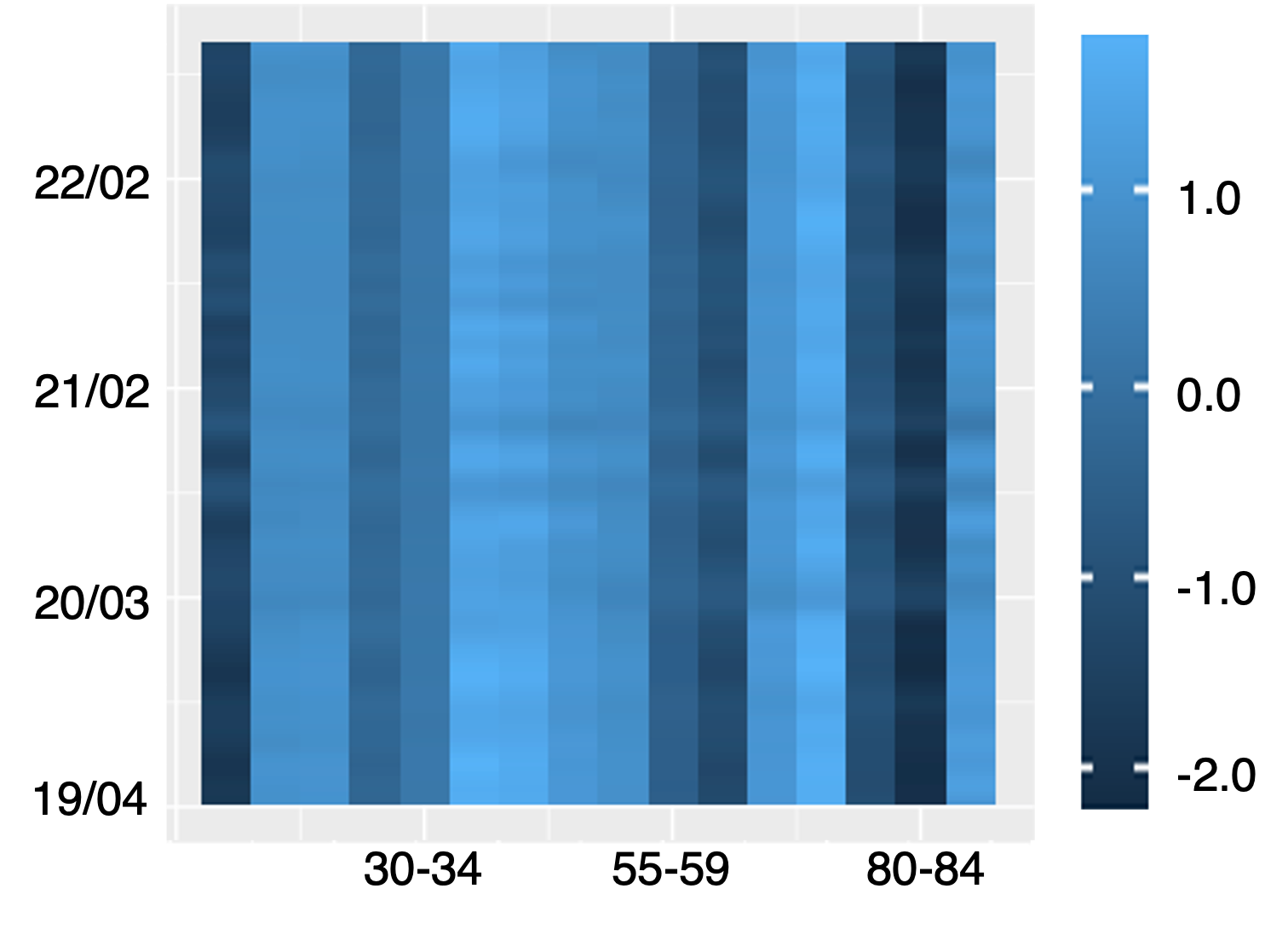}
         \caption{2D plot of estimated $\Psi(r,t)$}
         \label{fig:MobEsF2}
     \end{subfigure}
     \hfill
     \begin{subfigure}[b]{0.45\linewidth}
         \centering
         \includegraphics[width=\linewidth]{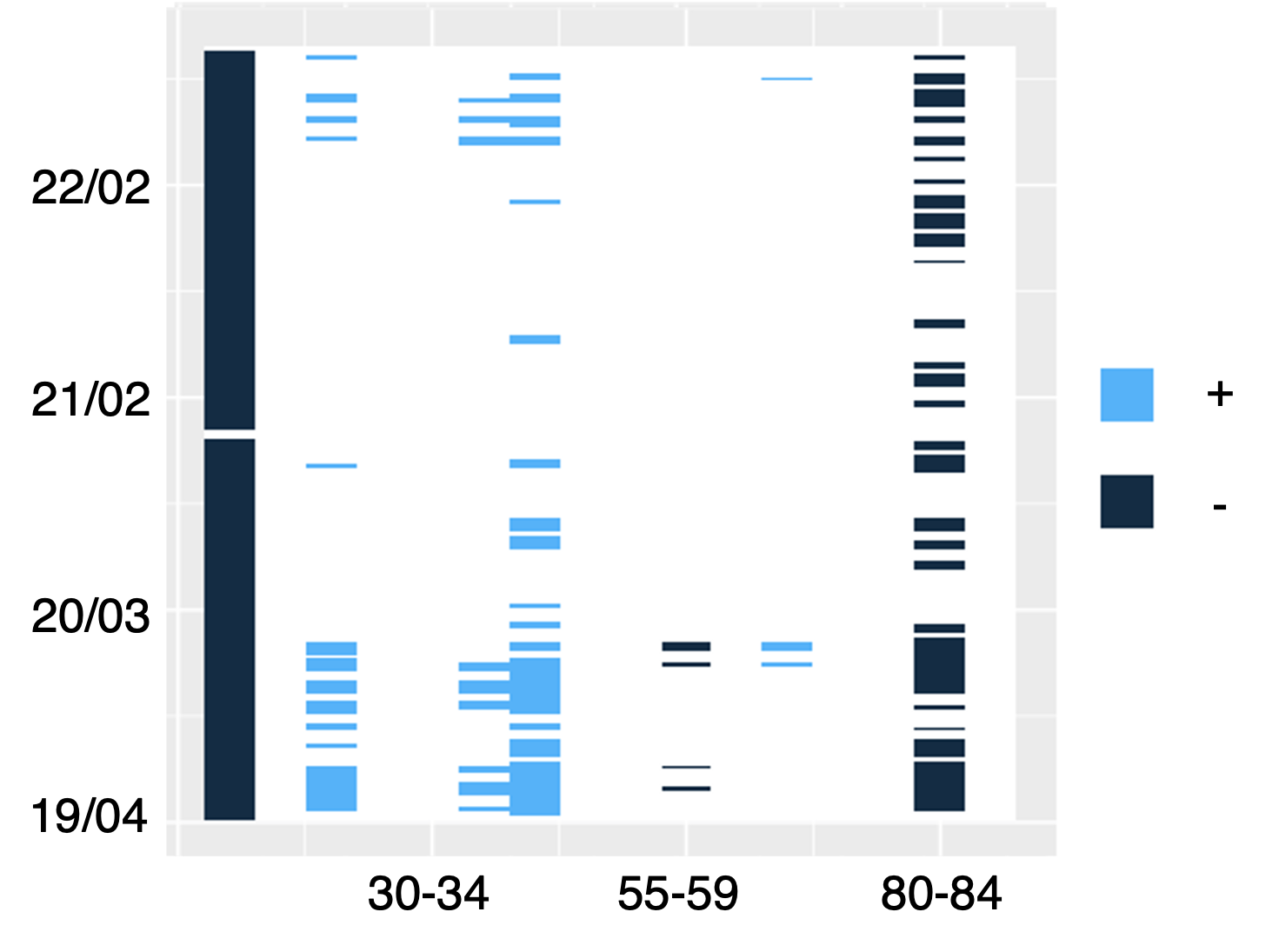}
         \caption{Significant areas}
         \label{fig:MobSigArea}
     \end{subfigure}
     \hfill
     \begin{subfigure}[b]{0.65\linewidth}
         \centering
         \includegraphics[width=\linewidth]{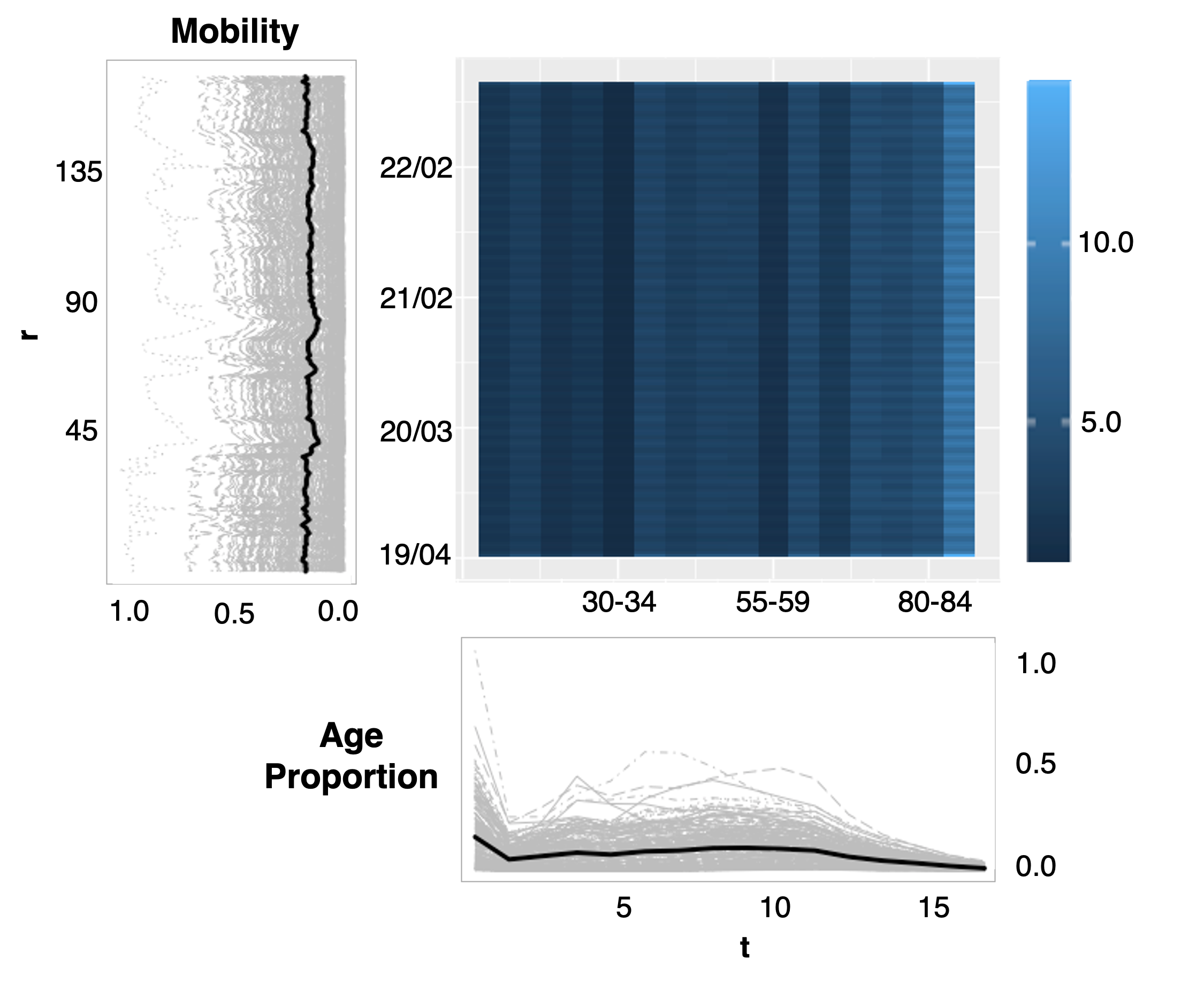}
         \caption{Uncertainty surface of estimated $\Psi(r,t)$}
         \label{fig:Mobwidth}
     \end{subfigure}
        \caption{Korea Mobility data: (a) Estimated $\Psi(r,t)$ from the posterior mean of PSFoFR. (b) Significant areas in the functions are detected from $\alpha=0.05$. Sky blue colors represent positive significant areas, whereas dark blue colors represent negative significant areas. (c) The uncertainty surface is obtained from the width of 95\% simultaneous credible intervals at each $(r, t)$. Black solid lines indicate the mean curves of the functional variables.    
        }
        \label{fig:MobEstimate}
\end{figure}

Figure \ref{fig:MobEstimate} shows the estimated regression function and its significant areas. We observe that the age groups 20-24 ($t=3$) and 40-44 ($t=7$) have a significant positive relationship with mobility in general. On the other hand, the age groups $\leq$15 ($t=1$) and 80-84 ($t=15$) have a significant negative relation with mobility. This implies that the regions with a high population of individuals in the 20-24 and 40-44 ($\leq$15 and 80-84) age groups tend to record a high (low) mobility. However, we observe that a positive relationship between age group 20-24 with mobility is not significant from March 2020 to February 2022. The age group 40-44 also shows less significant relationships in that period. This is due to the social distancing policies implemented in South Korea from March 2020 to April 2022 \citep{ha2023covid}. In that period, the mean of the mobility curves becomes lower than in other periods. The social distancing policy for COVID-19 control decreases overall mobility, resulting in a less significant relationship from March 2020 to February 2022. We observe that the significant area for the age group 20-24 has decreased to a larger degree than the age group 40-44 during social distancing. This is because individuals in the group 20-24 were more affected by the policy; universities were forced to shift to online teaching \citep{hong2020factors}. On the other hand, individuals in the group 40-44 worked for their livelihood, and full remote work was not mandatory by the social distancing period \citep{SYJ2021online}. Furthermore, during the social distancing period, $\widehat{\Psi}(r,t)$ shows consistently low values in the age groups $\leq$15 and 80-84. The mobility of the age group $\leq$15 and 80-84 is inherently low and has little variability because their life radius is relatively small.

As in the previous example, we provide an uncertainty surface of $\widehat{\Psi}(r,t)$. Figure~\ref{fig:MobEstimate} shows that uncertainties of $\widehat{\Psi}(r,t)$ are large when $t=16$ (i.e., age group $\geq$85). We observe that the variability of the age proportion curves is small in the age group $\geq$85. As we pointed out, the variability of the regression function is inversely proportional to the variability of the predictor curve. In general, the proportion of the age group $\geq$85 is low, resulting in low variability across the region. On the other hand, uncertainties of $\widehat{\Psi}(r,t)$ are small in the age group 30-34, which can be attributed to the large regional differences.


\section{Discussion}
\label{sec:discuss}

In this manuscript, we propose a Bayesian framework for FoFR in the presence of spatial correlation. We utilize the basis transformation and projection methods to address within-curve and between-curve dependencies in the model efficiently. By incorporating functional covariates into the model, both SFoFR and PSFoFR can show improved kriging performance and provide interpretations of regression functions. The proposed methods can also be applicable to the spatial functional data over discrete domains. We have also studied the convergence of the regression function with an increasing sample size and MCMC iteration. Uncertainty quantification in estimation and prediction through posterior samples is also one of the advantages over the existing method. PSFoFR is computationally more efficient than SFoFR and shows faster mixing of the chain, while they provide comparable inferential results. 


The proposed framework is flexible and can be extendable to a wide variety of hierarchical functional regression models. For example, we can consider additive function-on-function regressions \citep{kim2018additive} or non-Gaussian functional response models \citep{scheipl2016generalized} for spatial functional data. Once the models include spatially dependent latent processes within the hierarchical model, we can readily use SFoFR or PSFoFR with minor changes to the {\tt nimble} code. As per future research direction, one might consider different spatial basis functions to represent correlated $\mathbf{W}$. Examples include thin plate splines \citep{lee2023scalable} and multi-resolution basis \citep{katzfuss2017multi}. A closer examination of adopting different spatial basis functions can improve the performance of the proposed method.

\section*{Supplementary Material}
Supplementary materials available online contain the full conditionals for MCMC, assessing the convergence of MCMC, details for contour avoiding region, estimated regression functions from SFoFR, additional simulation studies, and simultaneous credible intervals for predicted curves. 

\section*{Acknowledgement}
This work was supported by the National Research Foundation of Korea (RS-2025-00513129, RS-2023-00217705, 2021R1A2C1095639). The authors are grateful to the anonymous reviewers for their careful reading and valuable comments.

\clearpage
\appendix
\begin{center}
\title{\LARGE\bf Supplementary Material for Bayesian Function-on-Function Regression for Spatial Functional Data}\\~\\
\author{\Large{Heesang Lee, Dagun Oh, Sunhwa Choi, and Jaewoo Park}}
\end{center}

\section{Full Conditionals for Markov Chain Monte Carlo}

\subsection{Full Conditionals for SFoFR}

For continuous spatial domains, let $f(y_{ik} | \psi_{gk}, \widetilde{\mathbf{W}}, \tau^2, \boldsymbol{\Sigma}_{\widetilde{\mathbf{W}}})$ be the normal likelihood function with priors $p(\psi_{gk}), p(\tau^2), p(\boldsymbol{\Sigma}_{\widetilde{\mathbf{W}}})$ where $\boldsymbol{\Sigma}_{\widetilde{\mathbf{W}}}$ is covariance structure of $\widetilde{\mathbf{W}}$. Then, we can define the joint posterior distribution for SFoFR (basis space) as $\pi( \psi_{gk}, \widetilde{\mathbf{W}}, \tau^2, \boldsymbol{\Sigma}_{\widetilde{\mathbf{W}}} | y_{ik}) \propto f(y_{ik}|\psi_{gk}, \widetilde{\mathbf{W}}, \tau^2)$ $f(\widetilde{\mathbf{W}} |\boldsymbol{\Sigma}_{\widetilde{\mathbf{W}}})p(\psi_{gk})p(\tau^2)p(\boldsymbol{\Sigma}_{\widetilde{\mathbf{W}}}).$ Here, we use independent normal conjugate priors for $\psi_{gk} \sim \mathcal{N}(0, 10)$, inverse gamma conjugate priors for $\tau^2 \sim \mathcal{IG}(2,0.1)$ and $\sigma^2 \sim \mathcal{IG}(2,0.1)$, and a uniform prior for $\rho \sim \mathcal{\mbox{Unif}}(0,1).$ Let $\boldsymbol{\Sigma}_{\widetilde{\mathbf{W}}}=\sigma^2\Gamma(\rho)$ be the Mat\'{e}rn class \citep{stein2012interpolation} covariance matrix with the range parameter $\rho$ and variance $\sigma^2$. Then, the full conditional distributions are derived as follows.

\begin{itemize}
\item The conditional distribution of $\psi_{gk}$:
\begin{equation}
\begin{split}
\pi(\psi_{gk}|\widetilde{\mathbf{W}}, \tau^2, \boldsymbol{\Sigma}_{\widetilde{\mathbf{W}}}, y_{ik}) & \propto \prod_{i=1}^n \prod_{k=1}^{k_n} f(y_{ik}|\psi_{gk},\widetilde{\mathbf{W}}, \tau^2)\times p(\psi_{gk})\\
& \propto  \exp \left(-\frac{1}{2\tau^2} \sum_{i=1}^n \sum_{k=1}^{k_n}(y_{ik}-x_{ig}\psi_{gk}-\widetilde{w}_{ik})^2\right) \times \exp\left(-\frac{1}{20} \psi_{gk}^2\right),
\end{split}
\end{equation}
\begin{align*}
  \therefore \psi_{gk}|\widetilde{\mathbf{W}}, \tau^2, \boldsymbol{\Sigma}_{\widetilde{\mathbf{W}}}, y_{ik} &\sim \mathcal{N}\left(\frac{M}{V}, \frac{1}{V}\right)\\ &\text{where} \; M =\frac{\sum_{i=1}^n\sum_{k=1}^{k_n}(y_{ik}-\widetilde{w}_{ik})x_{ig}}{\tau^2}, \,V=\frac{\sum_{i=1}^n \sum_{k=1}^{k_n} x_{ig}^2}{\tau^2}+\frac{1}{10}.  
\end{align*}

\item The conditional distribution of $\widetilde{\mathbf{W}}$:
\begin{equation}
\begin{split}
\pi(\widetilde{\mathbf{W}}|\psi_{gk}, \tau^2, \boldsymbol{\Sigma}_{\widetilde{\mathbf{W}}}, y_{ik}) & \propto \prod_{i=1}^n \prod_{k=1}^{k_n} f(y_{ik}|\psi_{gk},\widetilde{\mathbf{W}}, \tau^2)\times f(\widetilde{\mathbf{W}}|\boldsymbol{\Sigma}_{\widetilde{\mathbf{W}}})\\
& \propto  \exp \left(-\frac{1}{2\tau^2} \sum_{i=1}^n \sum_{k=1}^{k_n}(y_{ik}-x_{ig}\psi_{gk}-\widetilde{w}_{ik})^2\right)\\
&\times \exp\left(-\frac{1}{2\sigma^2 } (\text{vec}(\widetilde{\mathbf{W}})'(I_{k_n} \otimes \Gamma(\rho)^{-1})\text{vec}(\widetilde{\mathbf{W}}))\right),
\end{split}
\end{equation}
\begin{align*}
  \therefore \widetilde{\mathbf{W}}|\psi_{gk}, \tau^2, \boldsymbol{\Sigma}_{\widetilde{\mathbf{W}}}, y_{ik} &\sim \mathcal{N}\left(\frac{M}{V}, \frac{1}{V}\right)\\ &\text{where} \; M =\frac{\sum_{i=1}^n\sum_{k=1}^{k_n}(x_{ig}\psi_{gk}-y_{ik})}{\tau^2}, \,V=\frac{n k_n}{\tau^2}+\frac{I_{k_n} \otimes \Gamma(\rho)^{-1}}{\sigma^2}.  
\end{align*}

\item The conditional distribution of $\tau^2$:
\begin{equation}
    \begin{split}
\pi(\tau^2|\psi_{gk}, \widetilde{\mathbf{W}}, \boldsymbol{\Sigma}_{\widetilde{\mathbf{W}}}, y_{ik}) & \propto 
\prod_{i=1}^n \prod_{k=1}^{k_n}f(y_{ik} |\psi_{gk}, \widetilde{\mathbf{W}}, \tau^2) \times p(\tau^2)\\
& \propto (\tau^2)^{-\frac{nk_n}{2}} \exp \left(-\frac{1}{2\tau^2} \sum_{i=1}^n \sum_{k=1}^{k_n}(y_{ik}-x_{ig}\psi_{gk}-\widetilde{w}_{ik})^2\right)\times (\tau^2)^{-3} \exp\left(-\frac{0.1}{\tau^2}\right)\\
& = (\tau^2)^{-\left(\frac{nk_n}{2}+2\right)-1} \exp\left(-\frac{1}{\tau^2} \left(\frac{\sum_{i=1}^n \sum_{k=1}^{k_n}(y_{ik}-x_{ig}\psi_{gk}-\widetilde{w}_{ik})^2}{2}+0.1\right)\right),      
    \end{split}
\end{equation}
$$\therefore \tau^2|\psi_{gk}, \widetilde{\mathbf{W}}, \boldsymbol{\Sigma}_{\widetilde{\mathbf{W}}}, y_{ik} \sim \mathcal{IG}\left(\frac{nk_n}{2}+2, \frac{\sum_{i=1}^n \sum_{k=1}^{k_n}(y_{ik}-x_{ig}\psi_{gk}-\widetilde{w}_{ik})^2}{2}+0.1\right).$$

\item The conditional distribution of $\boldsymbol{\Sigma}_{\widetilde{\mathbf{W}}}$:
\begin{equation}
    \begin{split}
\pi(\boldsymbol{\Sigma}_{\widetilde{\mathbf{W}}}|\psi_{gk}, \widetilde{\mathbf{W}}, \tau^2, y_{ik}) & \propto 
f(\widetilde{\mathbf{W}} | \boldsymbol{\Sigma}_{\widetilde{\mathbf{W}}}) \times f(\boldsymbol{\Sigma}_{\widetilde{\mathbf{W}}}|\sigma^2, \rho) \times p(\sigma^2) \times p(\rho)\\
& \propto  (\sigma^2)^{-\frac{k_n}{2}} \exp\left(-\frac{1}{2\sigma^2 } (\widetilde{\mathbf{W}}'\Gamma(\rho)^{-1}\widetilde{\mathbf{W}})\right) \times (\sigma^2)^n \exp\left(-\sum_{i=1}^n\sum_{j=1}^n\frac{d_{ij}}{\rho}\right)\\ &\times (\sigma^2)^{-3} \exp \left(-\frac{0.1}{\sigma^2}\right) \times 1_{\left[0,1\right]}(\rho)\\
&= (\sigma^2)^{n-\frac{k_n}{2}-2} \exp\left(-\left(\frac{1}{2\sigma^2 } (\widetilde{\mathbf{W}}'\Gamma(\rho)^{-1}\widetilde{\mathbf{W}})+\sum_{i=1}^n\sum_{j=1}^n\frac{d_{ij}}{\rho}+\frac{0.1}{\sigma^2}\right)\right)\times 1_{\left[0,1\right]}(\rho),    
    \end{split}
\end{equation}
where $d_{ij}$ is distance between location $i$ and location $j$.
\end{itemize}

Similarly, for discrete spatial domains, we can define the normal likelihood function $f(y_{ik} | \psi_{gk}, \widetilde{\mathbf{W}}, \tau^2, \boldsymbol{\Sigma}_{\widetilde{\mathbf{W}}})$ and priors $p(\psi_{gk}), p(\tau^2)p(\boldsymbol{\Sigma}_{\widetilde{\mathbf{W}}})$. Then, the joint posterior distribution for SFoFR (basis space) as 
$\pi( \psi_{gk}, \widetilde{\mathbf{W}}, \tau^2, \boldsymbol{\Sigma}_{\widetilde{\mathbf{W}}} | y_{ik}) \propto \prod_{i=1}^{n}\prod_{k=1}^{k_n} f(y_{ik} |\psi_{gk}, \widetilde{\mathbf{W}}, \tau^2)$ $f(\widetilde{\mathbf{W}} | \nu)p(\psi_{gk})p(\tau^2)p(\boldsymbol{\Sigma}_{\widetilde{\mathbf{W}}}).$ Here, we use independent normal conjugate priors for $\psi_{gk} \sim \mathcal{N}(0, 10)$, an inverse gamma conjugate prior for $\tau^2 \sim \mathcal{IG}(2,0.1)$ and a Gamma conjugate piror for $\nu \sim \mathcal{G}(0.5,\frac{1}{2,000})$. Let $\mathbf{Q}=\text{diag}(\mathbf{D1})-\mathbf{D}$ be a precision matrix, and $\boldsymbol{\Sigma}_{\widetilde{\mathbf{W}}}=\nu\mathbf{Q}$. Then, the full conditional distributions are derived as 
\begin{itemize}
\item The conditional distribution of $\psi_{gk}$:
\begin{equation}
\begin{split}
\pi(\psi_{gk}|\widetilde{\mathbf{W}}, \tau^2, \boldsymbol{\Sigma}_{\widetilde{\mathbf{W}}}, y_{ik}) & \propto \prod_{i=1}^n \prod_{k=1}^{k_n} f(y_{ik}|\psi_{gk},\widetilde{\mathbf{W}}, \tau^2)\times p(\psi_{gk})\\
& \propto  \exp \left(-\frac{1}{2\tau^2}(y_{ik}-x_{ig}\psi_{gk}-\widetilde{w}_{ik})^2\right) \times \exp\left(-\frac{1}{20} \psi_{gk}^2\right),
\end{split}
\end{equation}
\begin{align*}
    \therefore \psi_{gk}|\widetilde{\mathbf{W}}, \tau^2, \boldsymbol{\Sigma}_{\widetilde{\mathbf{W}}}, y_{ik} &\sim \mathcal{N}\left(\frac{M}{V}, \frac{1}{V}\right)\\ &\text{where} \; M =\frac{\sum_{i=1}^n\sum_{k=1}^{k_n}(y_{ik}-\widetilde{w}_{ik})x_{ig}}{\tau^2}, \,V=\frac{\sum_{i=1}^n \sum_{k=1}^{k_n} x_{ig}^2}{\tau^2}+\frac{1}{10}.
\end{align*}

\item The conditional distribution of $\widetilde{\mathbf{W}}$:
\begin{equation}
\begin{split}
\pi(\widetilde{\mathbf{W}}|\psi_{gk}, \tau^2, \boldsymbol{\Sigma}_{\widetilde{\mathbf{W}}}, y_{ik}) & \propto \prod_{i=1}^n \prod_{k=1}^{k_n} f(y_{ik}|\psi_{gk},\widetilde{\mathbf{W}}, \tau^2)\times f(\widetilde{\mathbf{W}}|\nu)\\
& \propto  \exp \left(-\frac{1}{2\tau^2} \sum_{i=1}^n \sum_{k=1}^{k_n}(y_{ik}-x_{ig}\psi_{gk}-\widetilde{w}_{ik})^2\right) \\
&\times \exp\left(\frac{\nu}{2} (\text{vec}(\widetilde{\mathbf{W}})'(I_{k_n}\otimes \mathbf{Q})\text{vec}(\widetilde{\mathbf{W}}))\right),
\end{split}
\end{equation}
$$\therefore \widetilde{\mathbf{W}}|\psi_{gk}, \tau^2, \boldsymbol{\Sigma}_{\widetilde{\mathbf{W}}}, y_{ik} \sim \mathcal{N}\left(\frac{M}{V}, \frac{1}{V}\right)\; \text{where} \; M =\frac{\sum_{i=1}^n\sum_{k=1}^{k_n}(x_{ig}\psi_{gk}-y_{ik})}{\tau^2}, \,V=\frac{n k_n}{\tau^2}-\nu(I_{k_n}\otimes \mathbf{Q}).$$

\item The conditional distribution of $\tau^2$:
\begin{equation}
    \begin{split}
\pi(\tau^2|\psi_{gk}, \widetilde{\mathbf{W}}, \boldsymbol{\Sigma}_{\widetilde{\mathbf{W}}}, y_{ik}) & \propto 
\prod_{i=1}^n \prod_{k=1}^{k_n}f(y_{ik} |\psi_{gk}, \widetilde{\mathbf{W}}, \tau^2) \times p(\tau^2)\\
& \propto (\tau^2)^{-\frac{nk_n}{2}} \exp \left(-\frac{1}{2\tau^2} \sum_{i=1}^n \sum_{k=1}^{k_n}(y_{ik}-x_{ig}\psi_{gk}-\widetilde{w}_{ik})^2\right)\times (\tau^2)^{-3} \exp\left(-\frac{0.1}{\tau^2}\right)\\
& = (\tau^2)^{-\left(\frac{nk_n}{2}+2\right)-1} \exp\left(-\frac{1}{\tau^2} \left(\frac{\sum_{i=1}^n \sum_{k=1}^{k_n}(y_{ik}-x_{ig}\psi_{gk}-\widetilde{w}_{ik})^2}{2}+0.1\right)\right),     
    \end{split}
\end{equation}
$$\therefore \tau^2|\psi_{gk}, \widetilde{\mathbf{W}}, \boldsymbol{\Sigma}_{\widetilde{\mathbf{W}}}, y_{ik} \sim \mathcal{IG}\left(\frac{nk_n}{2}+2, \frac{\sum_{i=1}^n \sum_{k=1}^{k_n}(y_{ik}-x_{ig}\psi_{gk}-\widetilde{w}_{ik})^2}{2}+0.1\right).$$

\item The conditional distribution of $\boldsymbol{\Sigma}_{\widetilde{\mathbf{W}}}$:
\begin{equation}
    \begin{split}
\pi(\boldsymbol{\Sigma}_{\widetilde{\mathbf{W}}}|\psi_{gk}, \widetilde{\mathbf{W}}, \tau^2, y_{ik}) & \propto f(\widetilde{\mathbf{W}} | \boldsymbol{\Sigma}_{\widetilde{\mathbf{W}}}) \times
p(\boldsymbol{\Sigma}_{\widetilde{\mathbf{W}}})\\
& \propto (\nu)^{-\frac{n}{2}} \exp\left(\frac{\nu}{2} (\widetilde{\mathbf{W}}'\mathbf{Q}\widetilde{\mathbf{W}})\right) \times (\nu)^{\frac{n}{2}} \exp\left(-\frac{n}{2,000}\nu\right)\\
& = (\nu)^{1-1} \exp\left(-\left(\frac{\widetilde{\mathbf{W}}'\mathbf{Q}\widetilde{\mathbf{W}}}{2}+\frac{n}{2,000}\right)\nu\right),         
    \end{split}
\end{equation}
$$\therefore \boldsymbol{\Sigma}_{\widetilde{\mathbf{W}}}|\psi_{gk}, \widetilde{\mathbf{W}}, \tau^2, y_{ik} \sim \mathcal{G}\left(1, \frac{\widetilde{\mathbf{W}}'\mathbf{Q}\widetilde{\mathbf{W}}}{2}+\frac{n}{2,000}\right).$$
\end{itemize}

\subsection{Full Conditionals for PSFoFR}

Let $f(y_{ik} | \psi_{gk}, \boldsymbol{\delta}, \tau^2, \boldsymbol{\Sigma}_{\boldsymbol{\delta}})$ be the normal likelihood function with priors $p(\psi_{gk}), p(\tau^2), p(\boldsymbol{\Sigma}_{\boldsymbol{\delta}})$ where $\boldsymbol{\Sigma}_{\boldsymbol{\delta}}=\sigma^2 \mathbf{P}'\mathbf{QP}$ is covariance structure of $\boldsymbol{\delta}$. Then we can define the joint posterior distribution for PSFoFR (basis space) as $\pi( \psi_{gk}, \boldsymbol{\delta}, \tau^2, \boldsymbol{\Sigma}_{\boldsymbol{\delta}} | y_{ik}) \propto \prod_{i=1}^{n}\prod_{k=1}^{k_n} f(y_{ik} |\psi_{gk}, \mathbf{P},  \boldsymbol{\delta}, \tau^2)$ $f(\boldsymbol{\delta} | \boldsymbol{\Sigma}_{\boldsymbol{\delta}})p(\psi_{gk})p(\tau^2)p(\boldsymbol{\Sigma}_{\boldsymbol{\delta}}).$ Here, we use independent normal conjugate priors for $\psi_{gk} \sim \mathcal{N}(0, 10)$, an inverse gamma conjugate priors for $\tau^2 \sim \mathcal{IG}(2,0.1)$ and $\sigma^2 \sim \mathcal{G}(0.5, \frac{1}{2,000})$. We define $(\mathbf{P}\boldsymbol{\delta})_{ik}$ as an element of $\mathbf{P}\boldsymbol{\delta}$ at $i$-th row and $k$-th column. From this, we can obtain the full conditionals of PSFoFR as follows:

\begin{itemize}
\item The conditional distribution of $\psi_{gk}$:
\begin{equation}
\begin{split}
\pi(\psi_{gk}|\boldsymbol{\delta}, \tau^2, \boldsymbol{\Sigma}_{\boldsymbol{\delta}}, y_{ik}) & \propto \prod_{i=1}^n \prod_{k=1}^{k_n} f(y_{ik}|\psi_{gk},\mathbf{P}, \boldsymbol{\delta}, \tau^2)\times p(\psi_{gk})\\
& \propto  \exp \left(-\frac{1}{2\tau^2}\sum_{i=1}^n\sum_{k=1}^{k_n}(y_{ik}-x_{ig}\psi_{gk}-(\mathbf{P}\boldsymbol{\delta})_{ik})^2\right) \times \exp\left(-\frac{1}{20} \psi_{gk}^2\right),
\end{split}
\end{equation}
\begin{align*}
    \therefore \psi_{gk}|\boldsymbol{\delta}, \tau^2, \boldsymbol{\Sigma}_{\boldsymbol{\delta}}, y_{ik} &\sim \mathcal{N}\left(\frac{M}{V}, \frac{1}{V}\right)\\ &\text{where} \; M =\frac{\sum_{i=1}^n\sum_{k=1}^{k_n}(y_{ik}-(\textbf{P}\boldsymbol{\delta})_{ik})x_{ig}}{\tau^2}, \,V=\frac{\sum_{i=1}^n \sum_{k=1}^{k_n} x_{ig}^2}{\tau^2}+\frac{1}{10}.
\end{align*}

\item The conditional distribution of $\tau^2$:
\begin{equation}
    \begin{split}
\pi(\tau^2|\psi_{gk}, \boldsymbol{\delta},\boldsymbol{\Sigma}_{\boldsymbol{\delta}}, y_{ik}) & \propto 
\prod_{i=1}^n \prod_{k=1}^{k_n}f(y_{ik} |\psi_{gk}, \mathbf{P}, \boldsymbol{\delta}, \tau^2) \times p(\tau^2)\\
& \propto (\tau^2)^{-\frac{nk_n}{2}} \exp \left(-\frac{1}{2\tau^2}\sum_{i=1}^{n}\sum_{k=1}^{k_n}(y_{ik}-x_{ig}\psi_{gk}-(\mathbf{P}\boldsymbol{\delta})_{ik})^2\right) \times (\tau^2)^{-3} \exp\left(-\frac{0.1}{\tau^2}\right)\\     
& = (\tau^2)^{-\left(\frac{nk_n}{2}+2\right)-1} \exp\left(-\frac{1}{\tau^2} \left(\frac{\sum_{i=1}^n \sum_{k=1}^{k_n}(y_{ik}-x_{ig}\psi_{gk}-(\mathbf{P}\boldsymbol{\delta})_{ik})^2}{2}+0.1\right)\right),      
    \end{split}
\end{equation}
$$\therefore \tau^2|\psi_{gk}, \boldsymbol{\delta},\boldsymbol{\Sigma}_{\boldsymbol{\delta}}, y_{ik} \sim \mathcal{IG}\left(\frac{nk_n}{2}+2, \frac{\sum_{i=1}^n \sum_{k=1}^{k_n}(y_{ik}-x_{ig}\psi_{gk}-(\mathbf{P}\boldsymbol{\delta})_{ik})^2}{2}+0.1\right).$$

\item The conditional distribution of $\boldsymbol{\delta}$:
\begin{equation}
\begin{split}
\pi(\boldsymbol{\delta}|\psi_{gk}, \tau^2, \boldsymbol{\Sigma}_{\boldsymbol{\delta}}, y_{ik}) & \propto \prod_{i=1}^{n}\prod_{k=1}^{k_n}f(y_{ik}|\psi_{gk}, \mathbf{P}, \boldsymbol{\delta}, \tau^2)\times f(\boldsymbol{\delta} | \boldsymbol{\Sigma}_{\boldsymbol{\delta}})\\
& \propto \exp \left(-\frac{1}{2\tau^2}\sum_{i=1}^{n}\sum_{k=1}^{k_n}(y_{ik}-x_{ig}\psi_{gk}-(\mathbf{P}\boldsymbol{\delta})_{ik})^2\right) \times \exp\left(-\frac{\sigma^2}{2} \boldsymbol{\delta}' \mathbf{P}'\mathbf{QP}\boldsymbol{\delta}\right),\\
\end{split} 
\end{equation}
\begin{align*}
    \therefore \boldsymbol{\delta}|\psi_{gk}, \tau^2, \boldsymbol{\Sigma}_{\boldsymbol{\delta}}, y_{ik} &\sim \mathcal{N}\left(\frac{M}{V}, \frac{1}{V}\right)\\ &\text{where} \; M =\frac{\sum_{i=1}^n\sum_{k=1}^{k_n}(x_{ig}\psi_{gk}-y_{ik})\mathbf{P}_i}{\tau^2}, \,V=\frac{n k_n \mathbf{P}'\mathbf{P}}{\tau^2}+\sigma^2\mathbf{P}'\mathbf{QP}.
\end{align*}

\item The conditional distribution of $\boldsymbol{\Sigma}_{\boldsymbol{\delta}}$:
\begin{equation}
    \begin{split}
\pi(\boldsymbol{\Sigma}_{\boldsymbol{\delta}}|\psi_{gk}, \boldsymbol{\delta}, \tau^2, y_{ik}) & \propto 
f(\boldsymbol{\delta} | \boldsymbol{\Sigma}_{\boldsymbol{\delta}}) \times p(\boldsymbol{\Sigma}_{\boldsymbol{\delta}})\\
& \propto (\sigma^2)^{\frac{p}{2}} \exp\left(-\frac{\sigma^2}{2} \boldsymbol{\delta}' \mathbf{P}'\mathbf{QP}\boldsymbol{\delta}\right) \times (\sigma^2)^{-\frac{p}{2}} \exp\left(-\frac{p}{2,000}\sigma^2\right)\\
&=(\sigma^2)^{1-1} \exp\left(-\left(\frac{\boldsymbol{\delta}' \mathbf{P}'\mathbf{QP}\boldsymbol{\delta}}{2}+\frac{p}{2,000}\right)\sigma^2\right),        
    \end{split}
\end{equation}
$$\therefore \boldsymbol{\Sigma}_{\boldsymbol{\delta}}|\psi_{gk}, \boldsymbol{\delta}, \tau^2, y_{ik} \sim \mathcal{G}(1, \frac{\boldsymbol{\delta}' \mathbf{P}'\mathbf{QP}\boldsymbol{\delta}}{2}+\frac{p}{2,000}).$$
\end{itemize}

\clearpage
\section{Assessing Convergence and Mixing of MCMC}

We assess the improvement in MCMC mixing from the projection step. As described in the manuscript, we run MCMC algorithms until the Monte Carlo standard errors are at or below 0.04 level. In simulation studies, we run 70,000 iterations for both PSFoFR and SFoFR. For the PM2.5 data example, we run 50,000 iterations for PSFoFR and 100,000 iterations for SFoFR. For the mobility data example, we run 50,000 iterations for PSFoFR and 200,000 iterations for SFoFR.  

Here, we provide selected trace plots of the 1,000 thinned posterior samples for each application (Figure~\ref{tofsignal1} and Figure~\ref{tofsignal2}). PSFoFR has better mixing compared to SFoFR, especially for random effect parameters. We also observe that PSFoFR shows a larger ESS, indicating better Markov chain mixing (Table~\ref{tabess}). 

\begin{table}[htbp]
  \centering
    \begin{tabular}{M{25mm}M{25mm}M{25mm}M{25mm}} \toprule
    $\boldsymbol{\delta}$ or $\widetilde{\textbf{W}}$ &\multicolumn{1}{M{25mm}}{Simulated data} &\multicolumn{1}{M{25mm}}{PM2.5 data} &\multicolumn{1}{M{25mm}}{Mobility data} \\ \midrule
        PSFoFR & 115 & 110 & 539 \\ \cmidrule(l){2-4}
        SFoFR & 17 & 16 & 31  \\  \cmidrule(l){1-4}
    $\widetilde{\boldsymbol{\psi}}$  &\multicolumn{1}{M{25mm}}{Simulated data} &\multicolumn{1}{M{25mm}}{PM2.5 data} &\multicolumn{1}{M{25mm}}{Mobility data} \\ \midrule
        PSFoFR & 226 & 389 & 725 \\ \cmidrule(l){2-4}
        SFoFR & 169 & 475 & 383  \\  \cmidrule(l){1-4}
    \end{tabular}
    \caption{
    Effective sample sizes for random effect parameters ($\boldsymbol{\delta}$ for PSFoFR or $\widetilde{\textbf{W}}$ for SFoFR) and regression function parameters ($\widetilde{\boldsymbol{\psi}}$). Effective sample sizes are computed from the 1,000 thinned posterior samples. 
    }
    \label{tabess}
\end{table}

\begin{figure}[htbp!]
\subfloat[
Simulated data example: $\boldsymbol{\delta}$ for PSFoFR.]{{\includegraphics[width=0.5\linewidth]{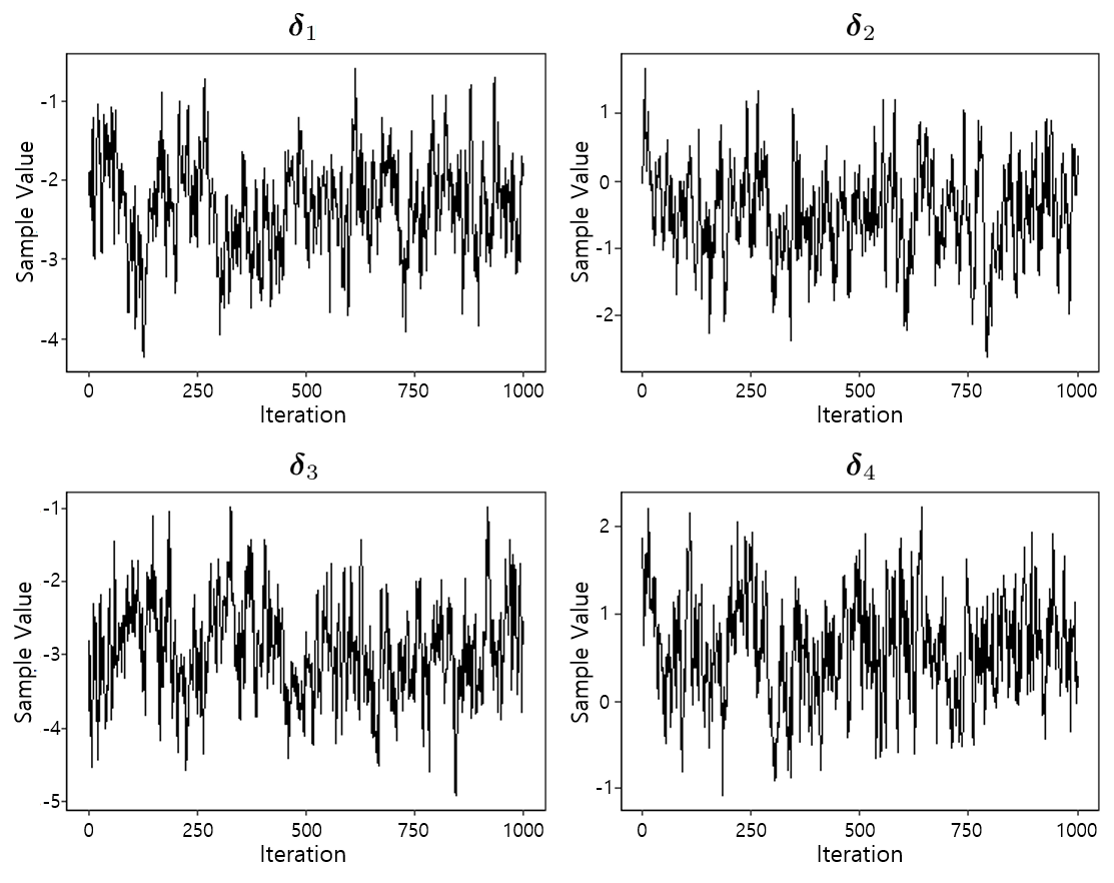} }}%
\subfloat[Simulated data example: $\widetilde{\textbf{W}}$ for SFoFR.]{{\includegraphics[width=0.5\linewidth ]{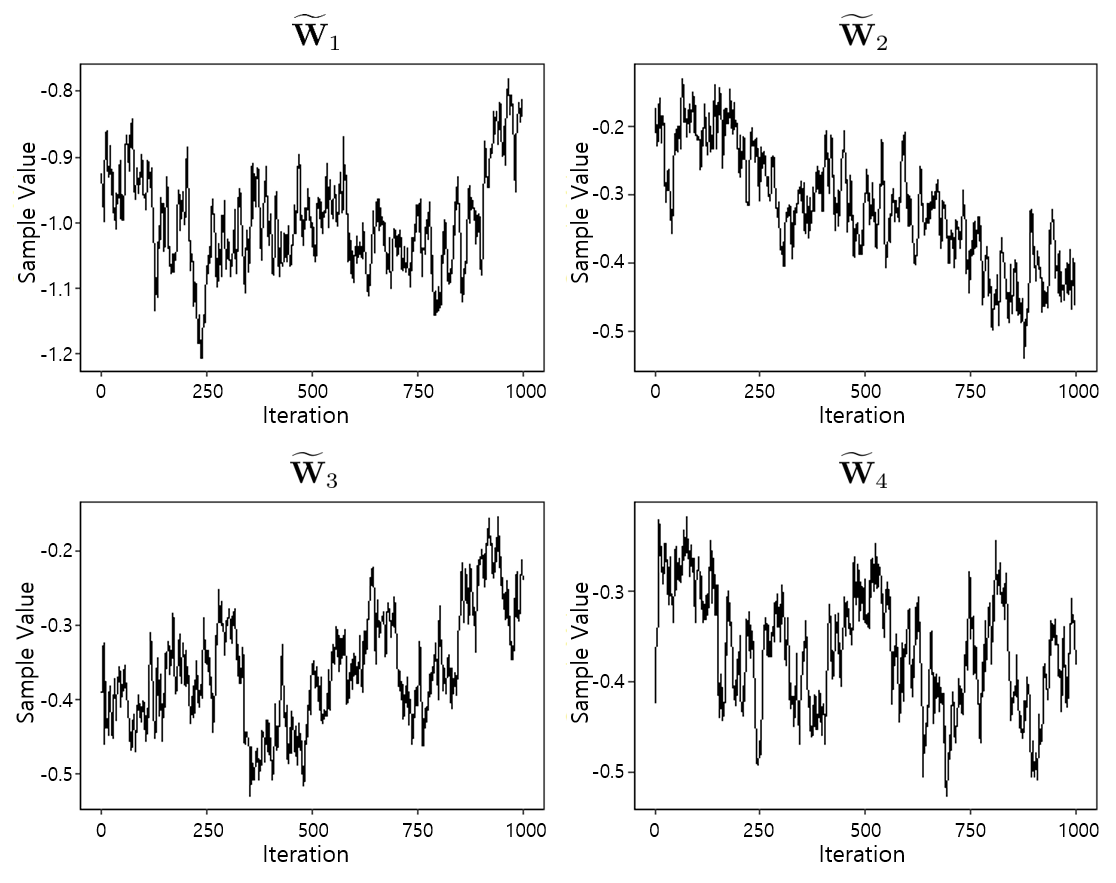} }}%
\hfill 
\subfloat[PM2.5 data example: $\boldsymbol{\delta}$ for PSFoFR.]{{\includegraphics[width=0.5\linewidth ]{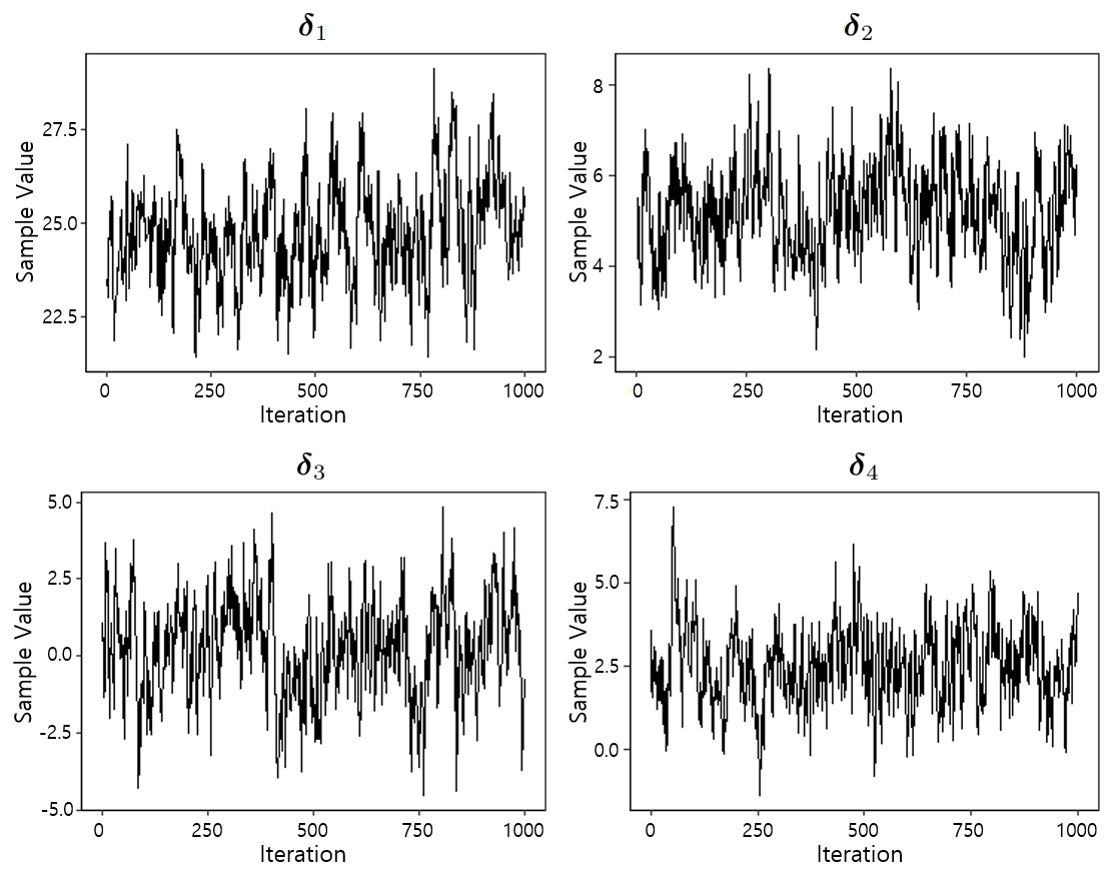} }}%
\subfloat[PM2.5 data example: $\widetilde{\textbf{W}}$ for SFoFR.]{{\includegraphics[width=0.5\linewidth ]{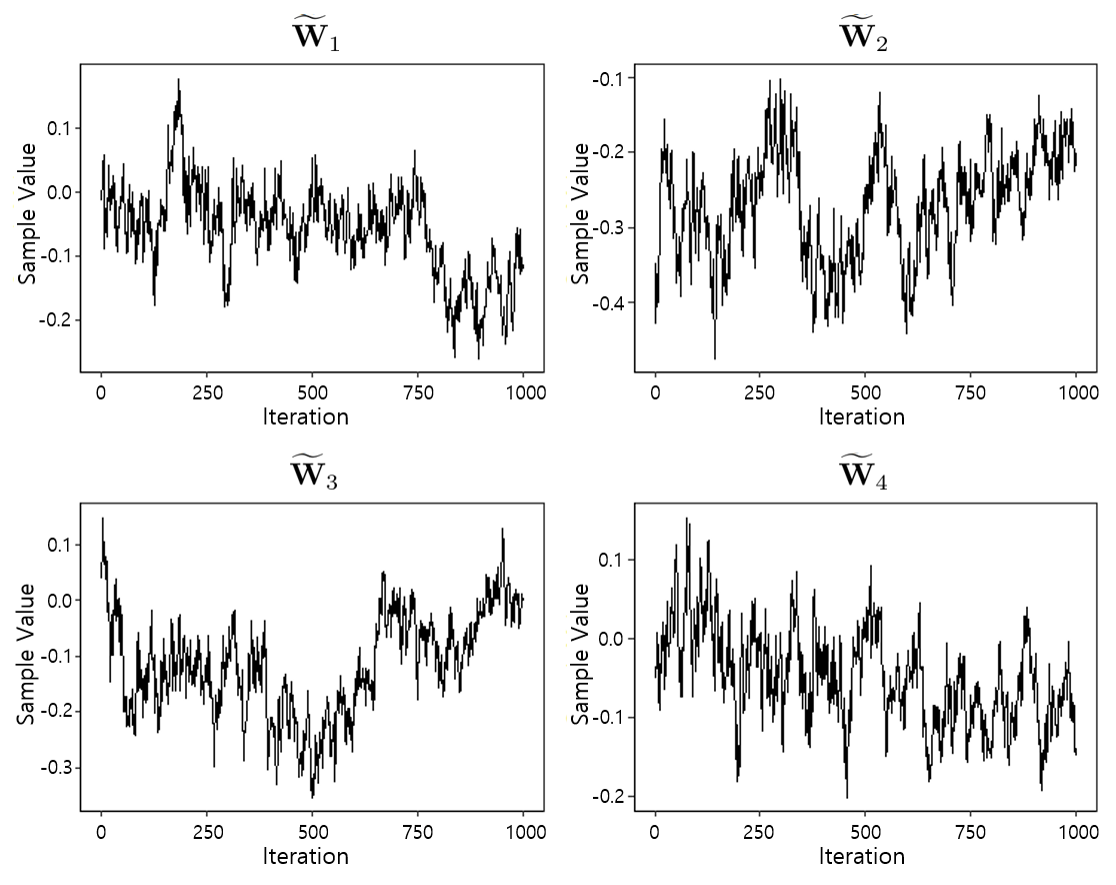} }}%
\hfill 
\subfloat[Mobility data example: $\boldsymbol{\delta}$ for PSFoFR.]{{\includegraphics[width=0.5\linewidth ]{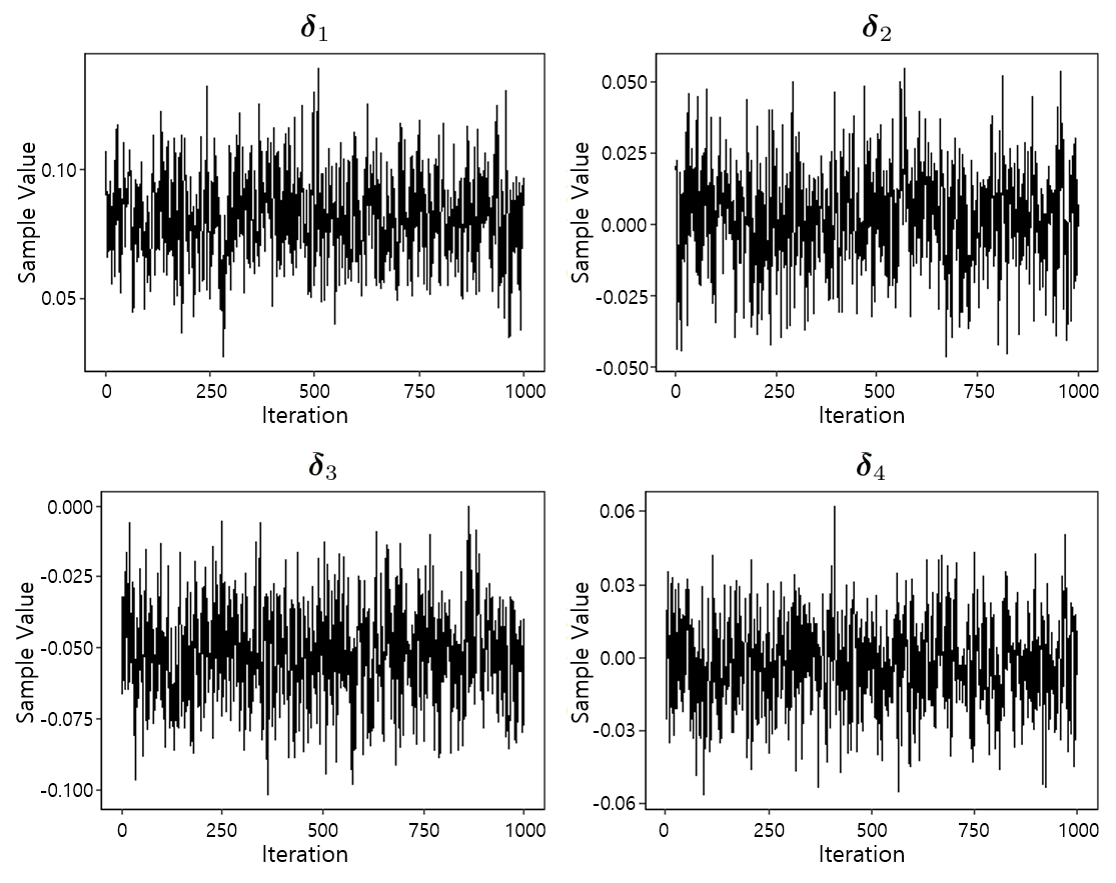} }}%
\subfloat[Mobility data example: $\widetilde{\textbf{W}}$ for SFoFR.]{{\includegraphics[width=0.5\linewidth ]{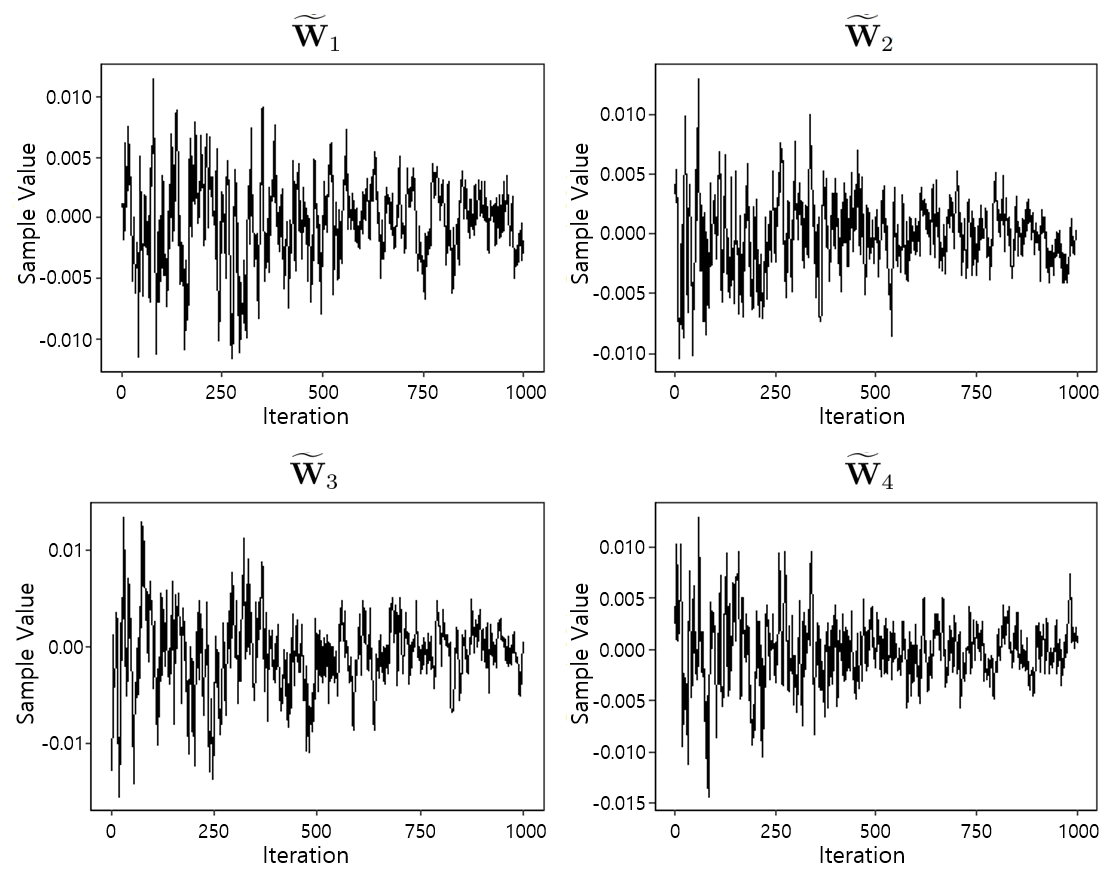} }}%
\caption{Selected trace plots for random effect parameters.}%
\label{tofsignal1}%
\end{figure}

\begin{figure}[htbp!]
\subfloat[Simulated data example: $\widetilde{\boldsymbol{\psi}}$ for PSFoFR.]{{\includegraphics[width=0.49\linewidth]{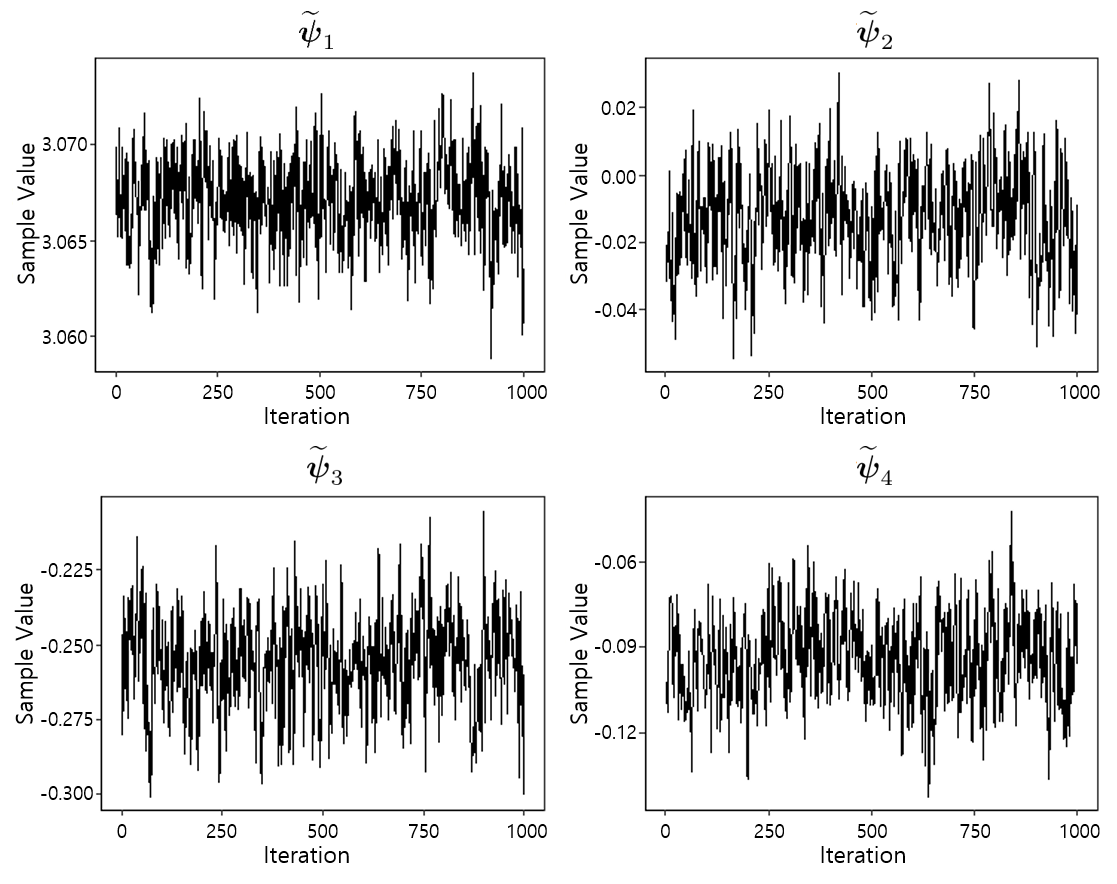} }}%
\subfloat[Simulated data example: $\widetilde{\boldsymbol{\psi}}$ for SFoFR.]{{\includegraphics[width=0.49\linewidth ]{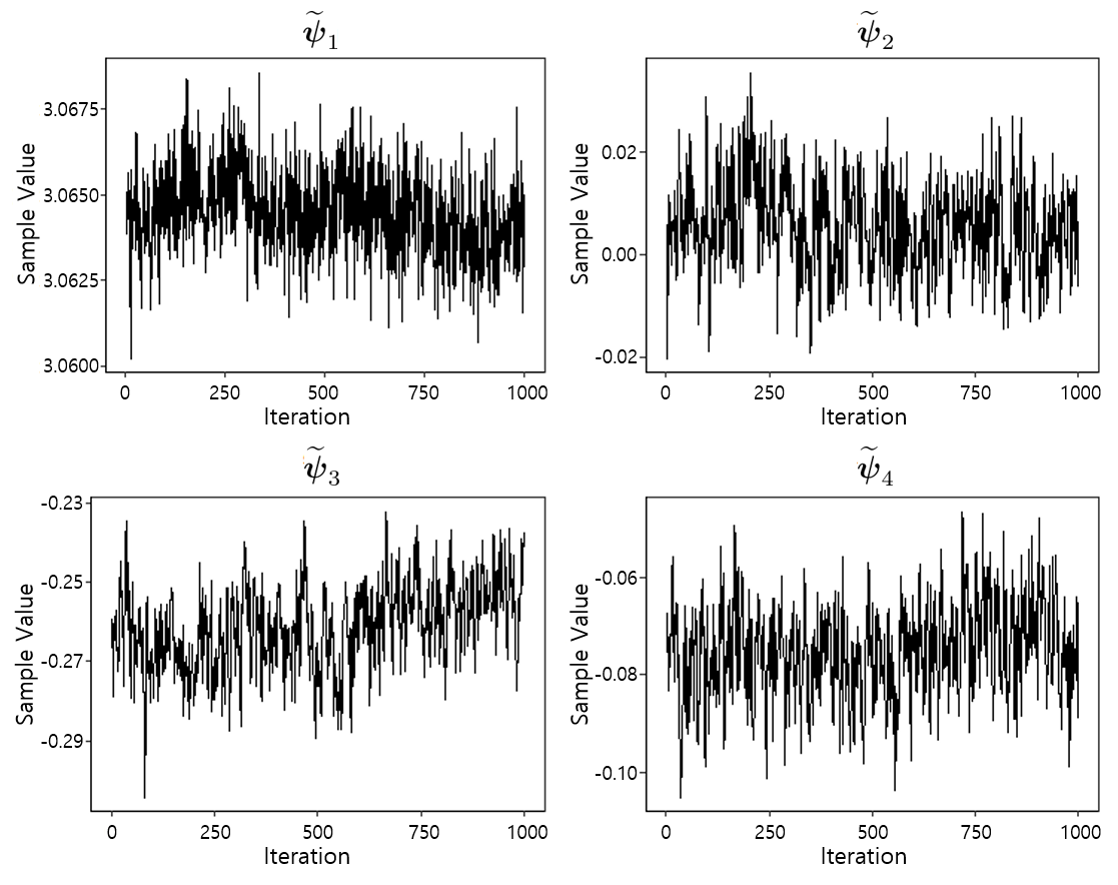} }}%
\hfill 
\subfloat[PM2.5 data example: $\widetilde{\boldsymbol{\psi}}$ for PSFoFR.]{{\includegraphics[width=0.49\linewidth ]{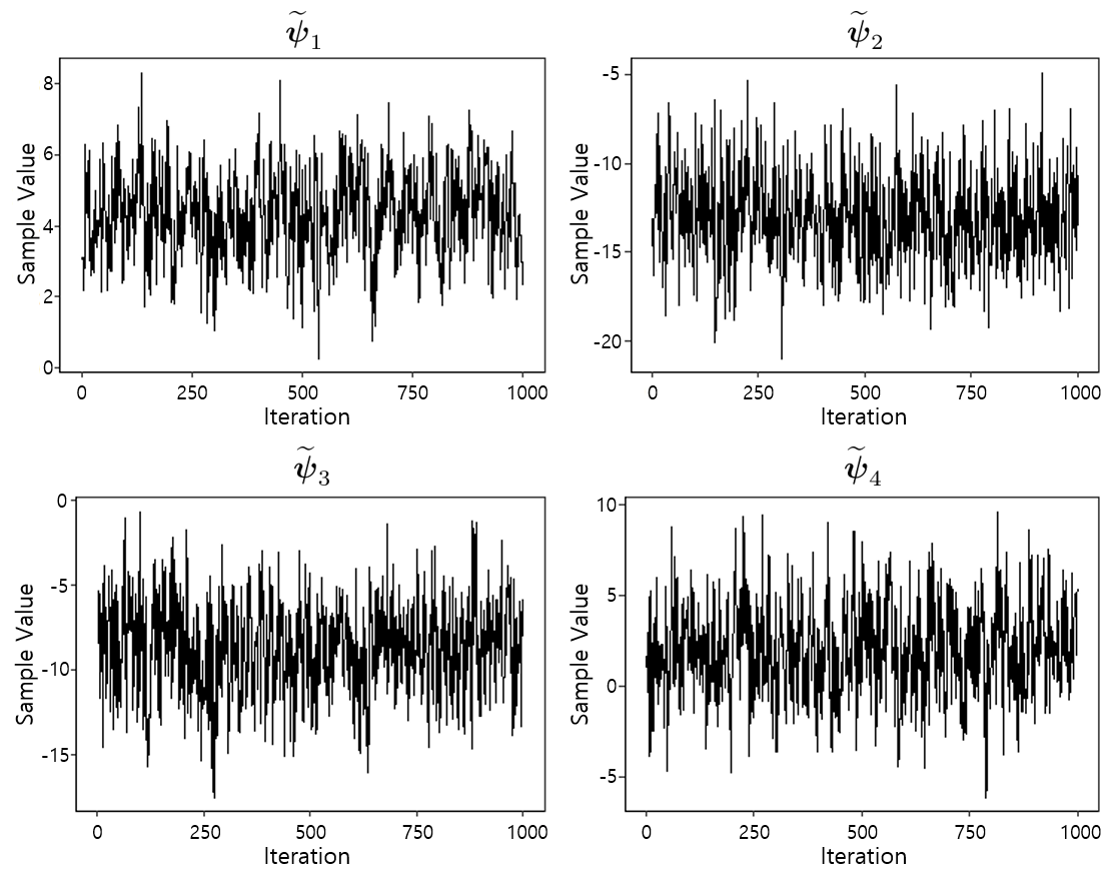} }}%
\subfloat[PM2.5 data example: $\widetilde{\boldsymbol{\psi}}$ for SFoFR.]{{\includegraphics[width=0.49\linewidth ]{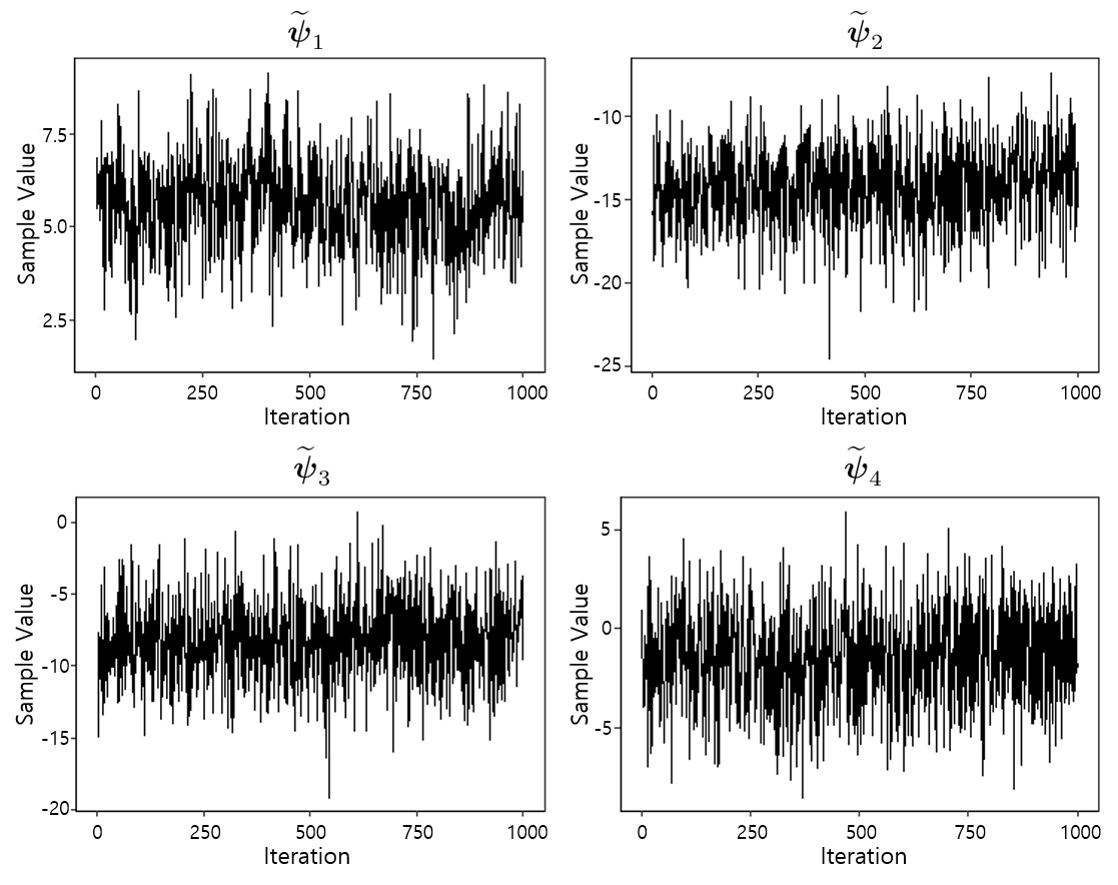} }}%
\hfill 
\subfloat[Mobility data example: $\widetilde{\boldsymbol{\psi}}$ for PSFoFR.]{{\includegraphics[width=0.49\linewidth ]{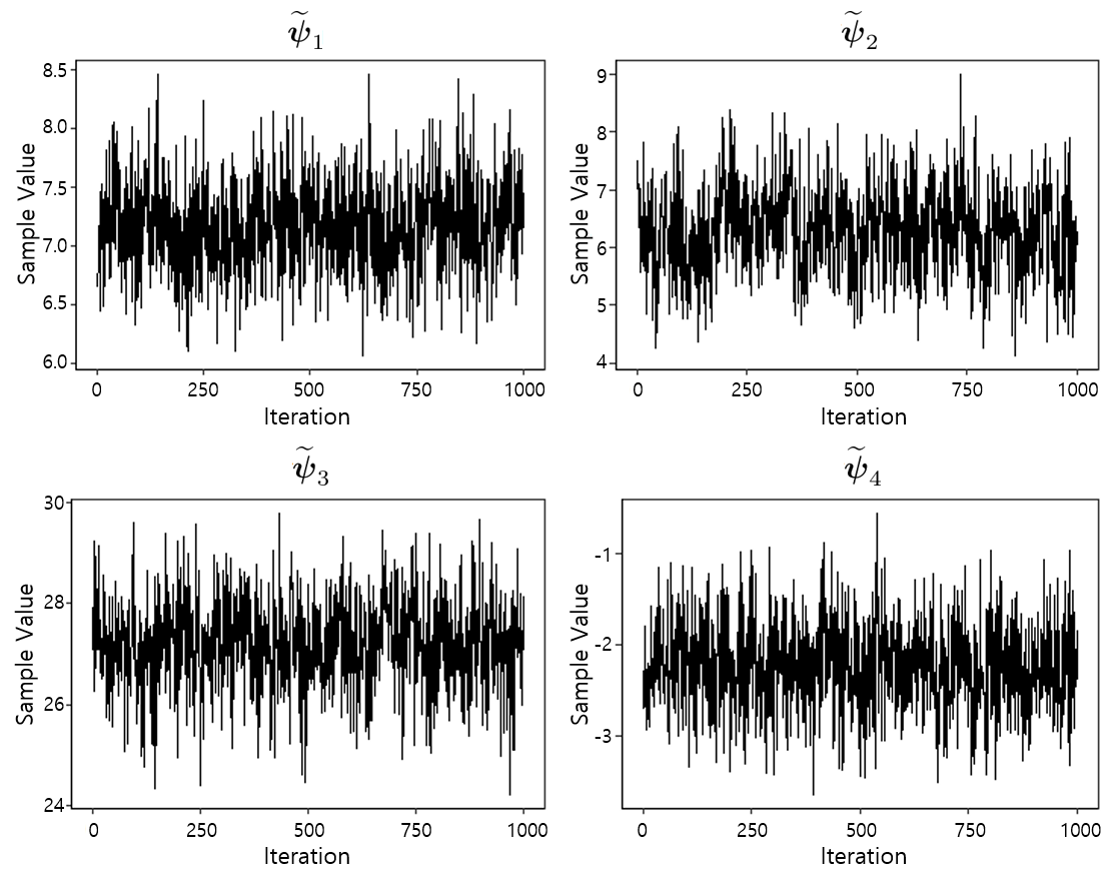} }}%
\subfloat[Mobility data example: $\widetilde{\boldsymbol{\psi}}$ for SFoFR.]{{\includegraphics[width=0.49\linewidth ]{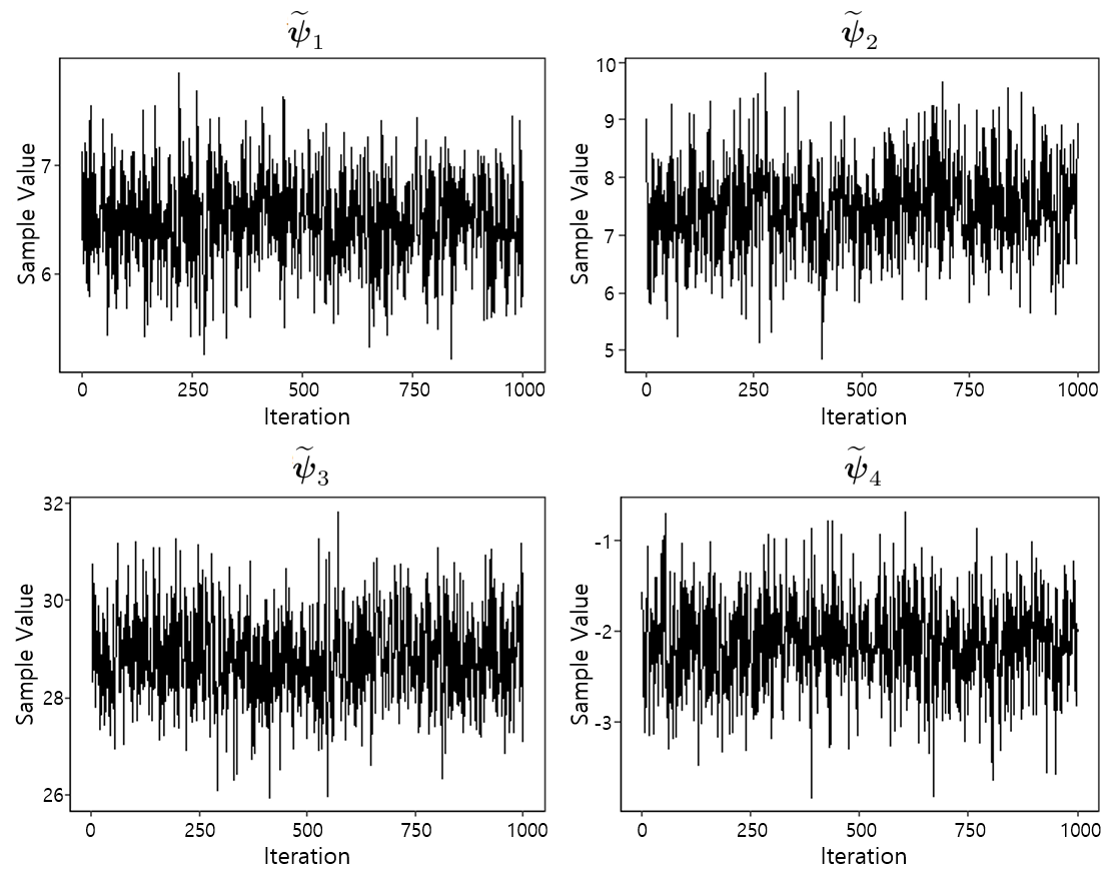} }}%
\caption{Selected trace plots for regression function parameters.
}%
\label{tofsignal2}%
\end{figure}

\clearpage
\section{Comparison with Contour Avoiding Region}
We obtain the contour avoiding region \citep{bolin2015excursion} for the regression function as follows. All the notations follow the main manuscript.   
\begin{enumerate}
\item We construct $\Psi_{g_nk_n}^{(u)}(r,t)=\sum_{k=1}^{k_n}\sum_{g=1}^{g_n} \psi_{gk}^{(u)}\phi_{gk}(r,t)$ from the MCMC samples $(\psi^{(u)}_{gk})_{g_n \times k_n}$, for $u=1,\cdots,U$.
\item We evaluate $\Psi_{g_nk_n}^{(u)}(r, t)$ for each grid point $r = 1,\cdots, n_v$ and $t = 1, \cdots, n_t$, where $1 < 2 < \cdots < n_v$ and $1 < 2 < \cdots < n_t$, respectively.
\item For each $(r,t)$, we compute $P(\Psi_{g_nk_n}(r,t)>0)$ as $\rho_U(r,t)=\frac{1}{U}\sum_{u=1}^{U}\mathbf{1}(\Psi_{g_nk_n}^{(u)}(r,t) > 0)$ and $P(\Psi_{g_nk_n}(r,t)<0)$ as $\rho_L(r,t)=\frac{1}{U}\sum_{u=1}^{U}\mathbf{1}(\Psi_{g_nk_n}^{(u)}(r,t) < 0)$.
\item We compute $\rho(r,t) = \max(\rho_U(r,t), \rho_L(r,t))$.
\item If $\rho_U(r,t) \leq 0.5$, we set limits $a(r,t)=-\infty$ and $b(r,t)=\text{SD}(\Psi_{g_nk_n}(r,t))\times \Phi^{-1}(\rho(r,t))$; otherwise $a(r,t)=\text{SD}(\Psi_{g_nk_n}(r,t))\times \Phi^{-1}(1-\rho(r,t))$ and $b(r,t)=\infty$, where $\Phi$ is the CDF of the standard normal distribution.
\item We compute $F_0(r,t)$ as $\frac{1}{U}\sum_{u=1}^{U}\mathbf{1}( \Psi_{g_nk_n}^{(u)}(r,t) \in [a(r,t), b(r,t)])$.
\item If $F_0(r,t)$ is exceed $1-\alpha$, we assign 1 to $(r,t)$ grid point; otherwise 0. 
\end{enumerate}
In Step 1 and Step 2, we construct the regression functions from posterior samples. In Step 3, we compute the marginal probability that $\Psi_{g_nk_n}(r,t)$ exceeds 0 or below 0 through a Monte Carlo approximation. In Step 5, we obtain the upper and lower limits to compute the contour avoidance function, $F_0(r,t)$. Note that $\rho_U(r,t) > 0.5$ implies that $P(\Psi_{g_nk_n}(r,t)>0) > P(\Psi_{g_nk_n}(r,t)<0)$ (i.e., the positive excursion set). On the other hand, $\rho_U(r,t) < 0.5$ implies that $P(\Psi_{g_nk_n}(r,t)>0) < P(\Psi_{g_nk_n}(r,t)<0)$ (i.e., the negative excursion set). The limits are computed using the CDF of the standard normal distribution because \cite{bolin2015excursion} assumed that the process follows a latent Gaussian model. In Step 7, the contour avoiding (significantly different from 0) regions are identified based on $F_0(r,t)$. We obtain contour avoiding regions ($\alpha = 0.05$) for the simulated and the real data examples using {\tt{excursions}} package \citep{excursion2023} in {\tt{R}}. Figure~\ref{tofsignal22} shows that the simultaneous credible intervals and the contour avoiding regions show similar patterns in general, though the contour avoiding regions are slightly wider.
\begin{figure}[htbp!]
\centering
\subfloat[Simulated data example.]{{\includegraphics[width=0.83\linewidth]{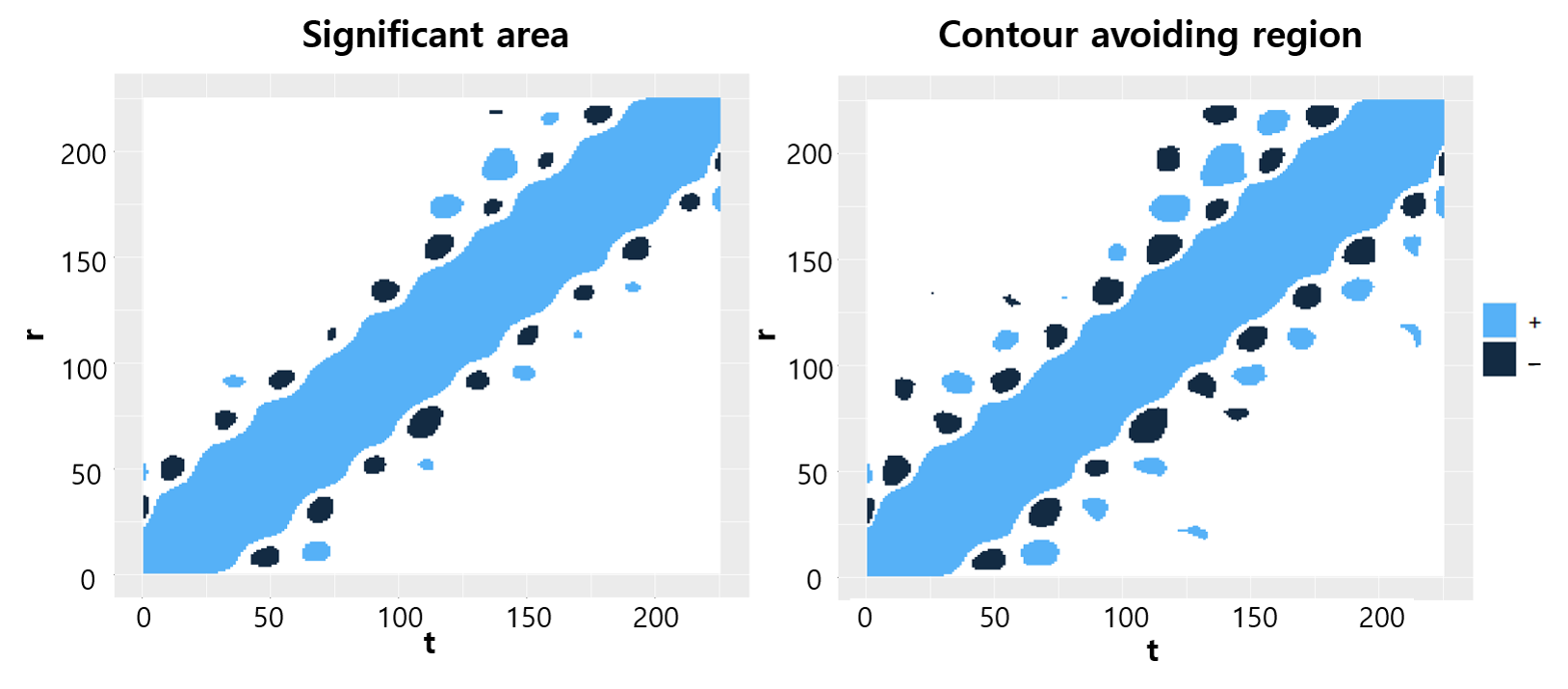} }}%
\hfill 
\subfloat[PM2.5 data example.]{{\includegraphics[width=0.83\linewidth ]{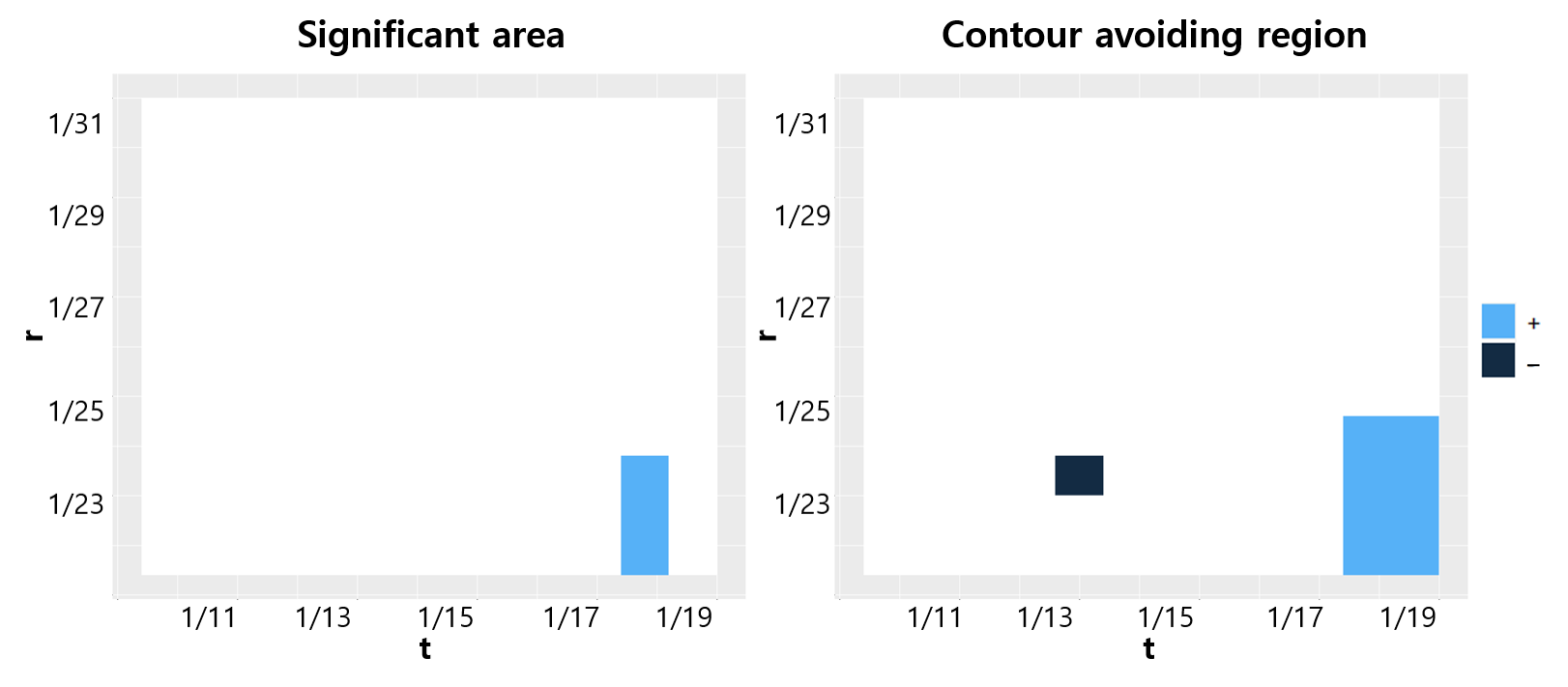} }}%
\hfill 
\subfloat[Mobility data example.]{{\includegraphics[width=0.83\linewidth ]{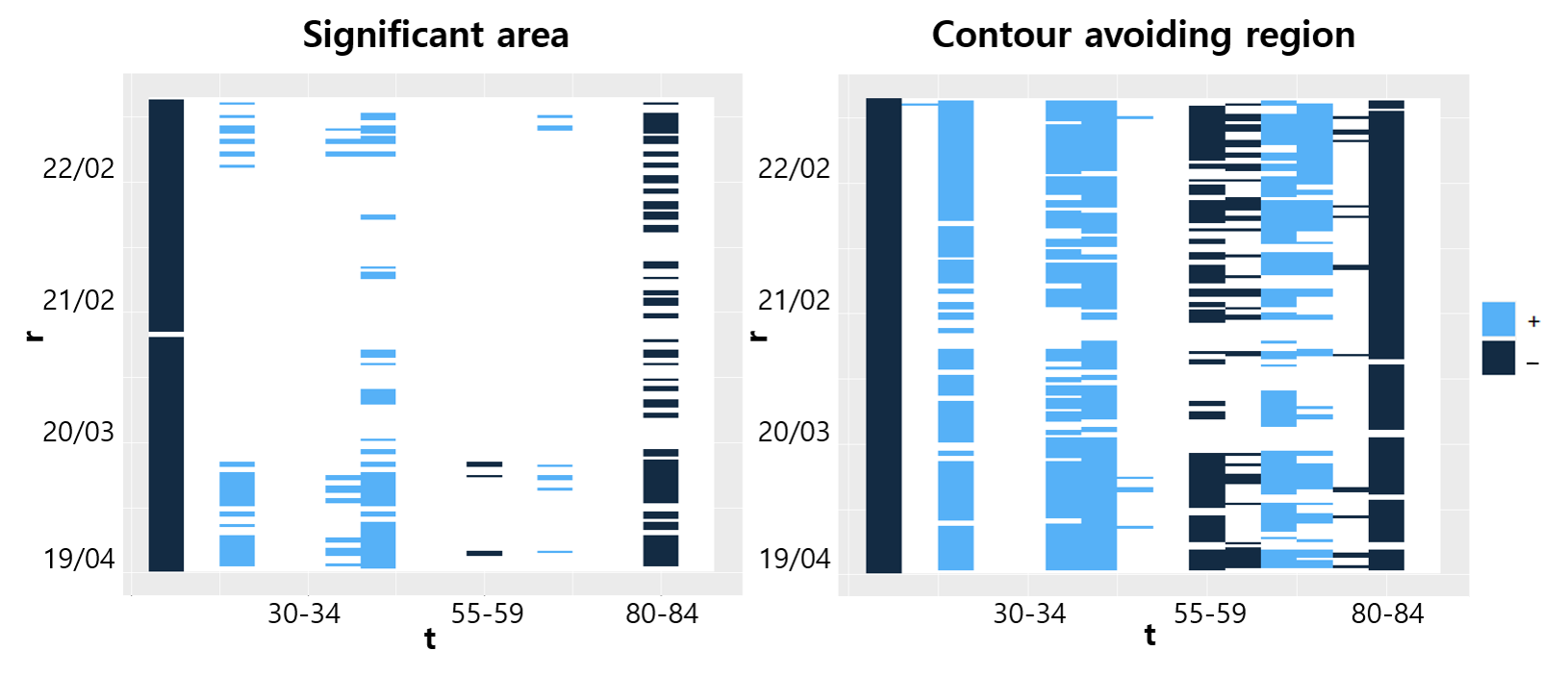} }}%
\caption{Significant areas detected from the simultaneous credible intervals (left panel) and contour avoiding regions (right panel).
}%
\label{tofsignal22}%
\end{figure}

\clearpage
\section{Simulation Results with Different Basis Functions}
\subsection{Fourier Basis}
We use $n=700$ samples for training and $n_{cv}=300$ samples for prediction. For PSFoFR, we run the MCMC algorithm for 70,000 iterations with 50,000 discarded for burn-in, and 1,000 thinned samples are obtained from the remaining 20,000. We use $k_n=19$ and $g_n=13$; the number of basis functions is chosen by minimizing GCV error. We set the rank values as 5\% of the sample size.
\begin{figure}[htbp!]
    \centering
    \includegraphics[width=0.8\linewidth]{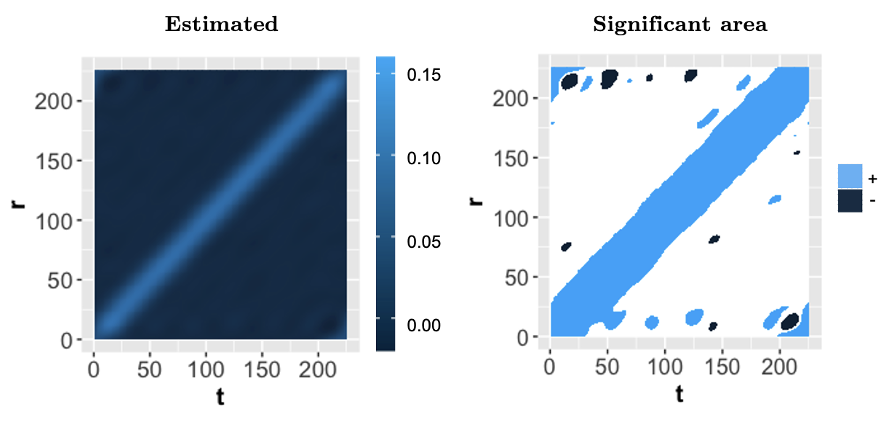}
    \caption{The estimated regression functions from PSFoFR with Fourier basis. Significant areas in the functions are detected from $\alpha=0.05$. Sky blue colors represent positive significant areas, whereas dark blue colors represent negative significant areas. 
         }
    \label{fig:AppWithSimulFourier}
\end{figure}

\subsection{FPC Basis}
We use $n=700$ samples for training and $n_{cv}=300$ samples for prediction. For PSFoFR, we run the MCMC algorithm for 70,000 iterations with 50,000 discarded for burn-in, and 1,000 thinned samples are obtained from the remaining 20,000. We use $k_n=11$ and $g_n=171$. We set $k_n$ and $g_n$ to achieve the function of variance explained at 90\%. As in the other basis functions, we set the rank values as 5\% of the sample size.
\begin{figure}[htbp!]
    \centering
    \includegraphics[width=0.8\linewidth]{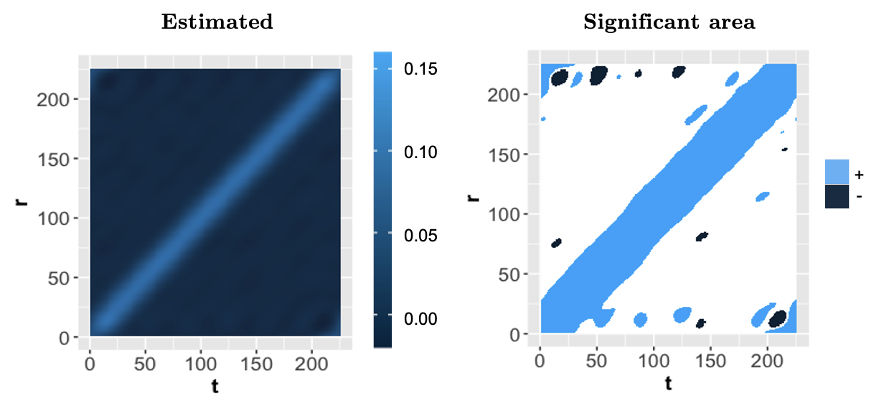}
    \caption{The estimated regression function from PSFoFR with FPC basis. Significant areas in the functions are detected from $\alpha=0.05$. Sky blue colors represent positive significant areas, whereas dark blue colors represent negative significant areas.}
    \label{fig:AppWithSimulPC}
\end{figure}

\subsection{Comparison}
\begin{table}[htbp]
  \centering
    \begin{tabular}{M{25mm}M{25mm}M{25mm}M{25mm}} \toprule
    \textbf{Basis Type} &\multicolumn{1}{M{25mm}}{\textbf{MSPE}} &\multicolumn{1}{M{25mm}}{\textbf{MSE }} &\multicolumn{1}{M{25mm}}{\textbf{Time}} \\ \midrule
        Fourier  & 0.124 &0.004& 45  \\ \cmidrule(l){1-4}
        FPC   &0.125 &0.004& 27 \\  \bottomrule
    \end{tabular}
    \caption{Inference results for the simulated datasets. MSPE, MSE, and computing time (min) are reported.}
    \label{tab2}
\end{table}
Figures \ref{fig:AppWithSimulFourier}, \ref{fig:AppWithSimulPC} indicate that the estimated regression functions are quantitatively similar to those from the B-spline basis expansion in the main manuscript. Although the overall trend is similar, we observe that the estimated functions from FPC are rough compared to the other two basis functions. Table \ref{tab2} shows that the results using Fourier and FPC bases are similar. However, the MSPE is higher compared to the B-spline basis.

\clearpage
\section{Simulation Results with a Different Regression Function}

Consider a regression function defined over the [0,225] $\times$ [0,225] domain as 
$$\Psi(r,t) =
\begin{cases} 
0, & \text{if } -0.03 \leq \Psi^*(r,t) \leq 0.03 \\
\Psi^*(r,t) - 0.03, & \text{if } \Psi^*(r,t) > 0.03 \\
\Psi^*(r,t) + 0.03, & \text{if } \Psi^*(r,t) < -0.03
\end{cases}
$$
where $$\Psi^*(r,t) = \frac{1}{10} \left( \sin\left(10r\right) \cos\left(10t\right) 
+ \exp\left(-5 \left(r^2 + t^2 \right)\right) 
+ 0.5 \sin\left(5 \left(r + t\right) \right) \right).$$
The rest of the simulation settings are identical to those in Section 4. We use $n=700$ samples for training and $n_{cv}=300$ samples for prediction. For PSFoFR, we run 70,000 iterations with 50,000 discarded for burn-in and 1,000 thinned samples obtained from the remaining 20,000. We choose $k_n=29$ and $g_n=12$ by minimizing the generalized
cross-validation (GCV). We use 1,520 triangular meshes and set $p$ as 5\% of the sample size. Training PSFoFR takes 103.52 minutes, which is longer than the training time in Section 4 due to the complexity of the regression function.

\begin{figure}[htbp]
     \centering
         \centering
         \includegraphics[width=\linewidth]{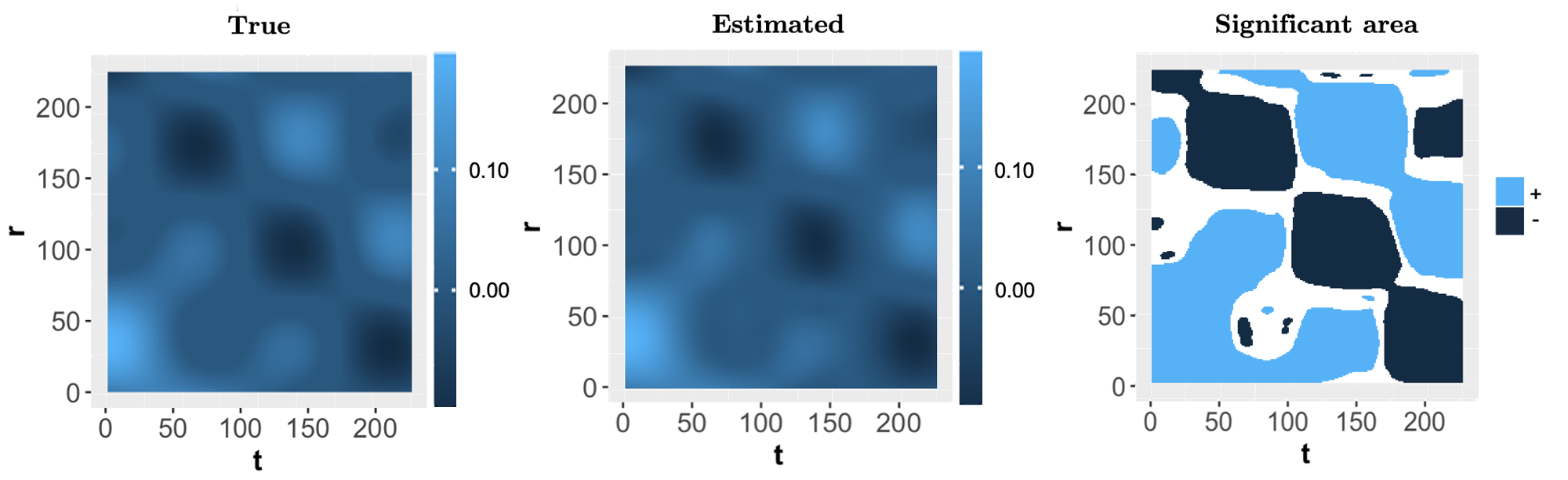}
         \caption{The estimated $\Psi(r,t)$ obtained from the posterior mean of PSFoFR. Significant areas in the functions are detected from $\alpha=0.05$. Sky blue colors represent positive significant areas, whereas dark blue colors represent negative significant areas.} 
         \label{fig:SimulcomParam}
\end{figure}

Figure~\ref{fig:SimulcomParam} indicates that the estimated regression functions are similar to the true regression function. In addition, the significant areas are well aligned with the important regions of the true $\Psi(r,t)$. MSE of the estimated regression function is computed as 0.003.
\clearpage
\begin{figure}[htbp]
     \centering
         \centering
         \includegraphics[width=\linewidth]{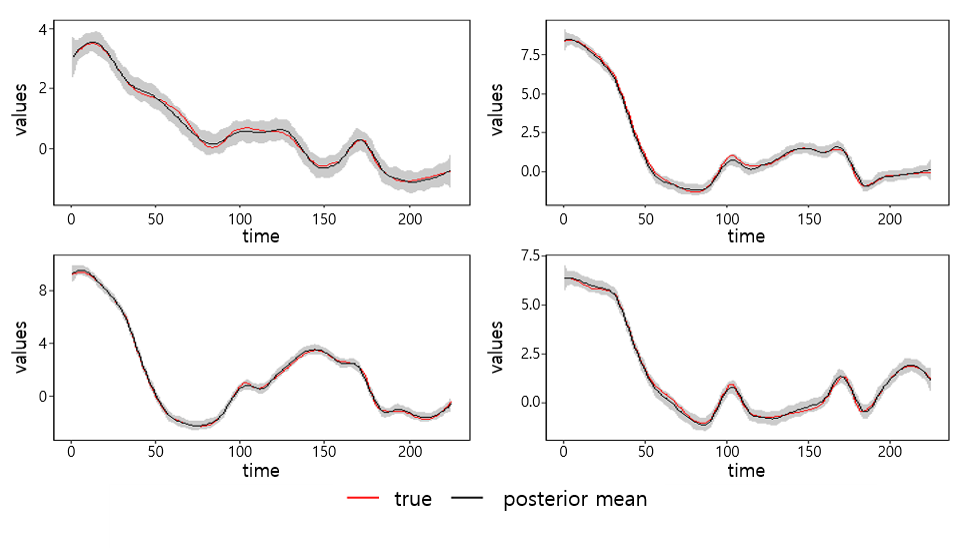}
         \caption{Visualizations of predicted curves (black lines) with 95\% simultaneous credible intervals (grey areas) from PSFoFR. Predicted curves are obtained from the sample mean of the posterior predictive distribution. Redlines indicate true curves at unobserved locations.         
         } 
         \label{fig:SimulcomPred}
\end{figure}

Figure~\ref{fig:SimulcomPred} provides four predicted curves from 300 test observations. We observe that the predicted curves are aligned with the true test functions; the MSPE of the predicted curves is computed as 0.108. As in Section 4, we compute the mean coverage for 300 predicted curves. Figure~\ref{fig:SimulcomPred} shows that the credible intervals include the true curves well, and the coverage is computed as 99\%. 

\clearpage

\clearpage
\section{Simultaneous Credible Intervals for Predicted Curves}

\subsection{Simulated Data}

\begin{figure}[htbp!]
    \centering
    \includegraphics[width=0.94\linewidth]{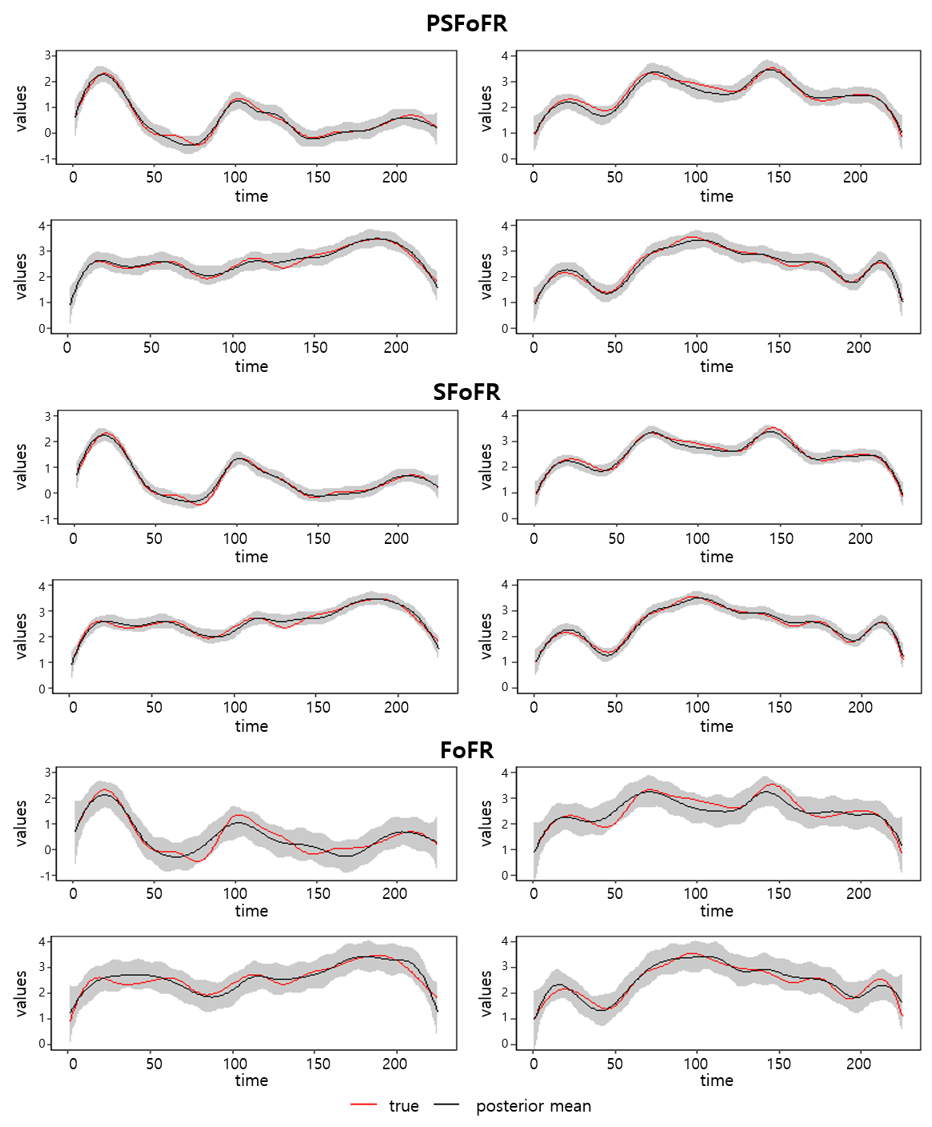}
    \caption{Visualizations of predicted curves (black lines) with 95\% simultaneous credible intervals (grey areas). Predicted curves are obtained from the sample mean of the posterior predictive distribution. Redlines indicate true curves at unobserved locations.}
    \label{fig:simu_credible}
\end{figure}

\clearpage
\subsection{PM2.5 Data}

\begin{figure}[htbp!]
    \centering
    \includegraphics[width=1\linewidth]{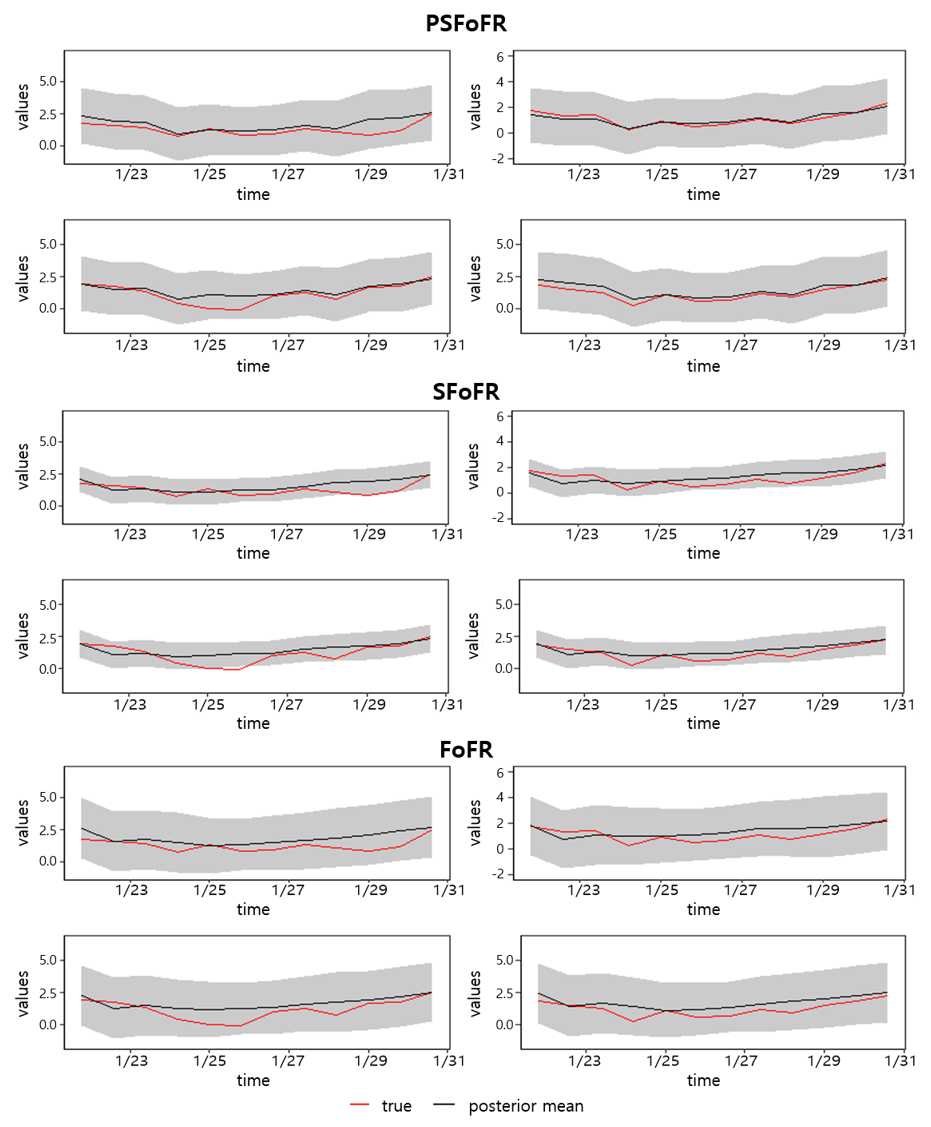}
    \caption{Visualizations of predicted curves (black lines) with 95\% simultaneous credible intervals (grey areas). Predicted curves are obtained from the sample mean of the posterior predictive distribution. Redlines indicate true curves at unobserved locations.}
    \label{fig:japan_credible}
\end{figure}

\bibliography{ref}

\end{document}